%% file: main.tex
\documentclass[12pt]{article}
\usepackage[top=1.25in,bottom=1in,left=1in,right=1in]{geometry}
\usepackage{graphicx}
\usepackage{amsmath}
\usepackage{bbm}
\usepackage{amssymb}
\usepackage{amsthm}
\usepackage{wrapfig}
\usepackage{multicol}
\usepackage{float}
\usepackage{pdfpages}
\usepackage[colorlinks=true, urlcolor=blue, linkcolor=blue,citecolor=blue]{hyperref}
\usepackage{natbib}

\usepackage[utf8]{inputenc}
\usepackage{subfiles}
\usepackage{blindtext}
\usepackage{enumitem}
\usepackage{titlesec}
\usepackage{lipsum}
\usepackage[title,toc,titletoc]{appendix}
\usepackage{etoolbox}
\usepackage{tikz}
\usepackage[font=footnotesize, textfont=it]{caption}

\newtheorem{definition}{Definition}
\newtheorem{proposition}{Proposition}
\newtheorem{theorem}{Theorem}
\newtheorem{corollary}{Corollary}
\newtheorem{lemma}{Lemma}

\newtheorem{assumption}{Assumption}

\theoremstyle{remark}
\newtheorem{example}{Example}
\theoremstyle{remark}
\newtheorem{remark}{Remark}
\DeclareMathOperator*{\argmax}{argmax}
\DeclareMathOperator*{\argmin}{argmin}

\DeclareMathOperator*{\diag}{diag}

\newtheoremstyle{named}{}{}{\itshape}{}{\bfseries}{.}{.5em}{\thmnote{#3}}
\theoremstyle{named}

\title{Dynamically Aggregating Diverse Information\footnote{We are grateful to Cuimin Ba, Yash Deshpande, Mira Frick, Drew Fudenberg, Boyan Jovanovic, Erik Madsen, George Mailath, Konrad Mierendorff, Lan Min, Jacopo Perego, Peter Norman S\o rensen, Jakub Steiner, Philipp Strack and Can Urgun for helpful comments and suggestions. Xiaosheng Mu acknowledges the hospitality of Columbia University and the Cowles Foundation at Yale University, which hosted him during parts of this research. Annie Liang thanks National Science Foundation Grant SES-1851629 for financial support.} \\}

\author{Annie Liang\thanks{Northwestern University} \quad \quad Xiaosheng Mu\thanks{Princeton University} \quad \quad Vasilis Syrgkanis\thanks{Microsoft Research}}

\begin{document}

\pagenumbering{gobble}

\maketitle
\begin{abstract}
An agent has access to multiple information sources, each of which provides information about a different attribute of an unknown state. Information is acquired continuously\textemdash where the agent chooses both which sources to sample from, and also how to allocate attention across them\textemdash until an endogenously chosen time, at which point a decision is taken. We provide an exact characterization of the optimal information acquisition strategy under weak conditions on the agent's prior belief about the different attributes. We then apply this characterization to derive new results regarding: (1) endogenous information acquisition for binary choice, (2) strategic information provision by biased news sources, and (3) the dynamic consequences of attention manipulation. 
\end{abstract}
\clearpage



\pagebreak

\pagenumbering{arabic}

\section{Introduction}

We study dynamic acquisition of information when a decision-maker has access to multiple sources of information, and limited resources with which to acquire that information. Our decision-maker seeks to learn a Gaussian state, and each information source is modeled as a diffusion process whose drift is an unknown ``attribute" that contributes linearly to the state. Attributes are potentially correlated. This structure captures information acquisition in many economic settings, including for example:
\begin{itemize}
 \item A governor wants to learn the number of cases of a disease outbreak, and can acquire information about the incidence rate of the disease in different cities.
    \item An investor wants to assess the value of an asset portfolio, and acquires information about the value of each asset included in the portfolio.
    \item An analyst wants to forecast a macroeconomic variable such as GDP growth, and needs to aggregate recent economic activities across different industries and locations.
\end{itemize}

\noindent  At every instant of continuous time, the decision-maker allocates a fixed budget of attention/resources across the information sources, where this allocation determines the precision of information extracted from each source. For example, the governor may have a limited number of tests to allocate across testing centers each day, where more tests lead to a more precise estimate of the incidence rate. The decision-maker acquires information until an endogenously chosen stopping time, at which point he makes a decision whose payoff depends on the unknown state.

Our contribution is to demonstrate that the optimal dynamic acquisition strategy can  be explicitly characterized under weak conditions on the prior belief, and to explain what those conditions are. Under the optimal strategy, the decision-maker initially exclusively acquires information from the single most informative source, where ``more informative" is evaluated with respect to his prior belief over the unknown attribute values. At fixed times, the decision-maker begins learning from additional sources, and divides attention over these new sources as well as the ones he was learning from previously. Eventually, the decision-maker acquires information from all sources using a final and constant mixture. Notably, the optimal information acquisition strategy is not only history-independent but also robust across all decision problems. This implies, for example, that one does not need to solve for the optimal stopping time and information acquisition strategy jointly in this problem---optimal information acquisition is independent of when the decision-maker stops. We make use of this implication in Section \ref{sec:binary choice} to derive new results about the optimal stopping behavior in a binary choice problem.

To gain some intuitions for the optimal information acquisition strategy, it is useful to compare our problem with a simpler one, in which the decision-maker acquires information for a decision at a fixed end date. Since the payoff-relevant state and all sources of information are Gaussian in our setting, the Blackwell-optimal solution would then be to acquire information in any way that minimizes posterior variance of the state at the known end date \citep{Blackwell, HansenTorgersen}. We show that under certain conditions on the prior belief, it is possible to ``string together" these solutions across different end dates using a single history-independent dynamic strategy, which thus minimizes posterior variance at every moment of time. Generalizing a result of \citet{Greenshtein}, we show that this strategy---which we call the ``uniformly optimal" strategy---is best for every decision problem and every distribution over stopping times, including those that are endogenously chosen. 

When a uniformly optimal strategy does not exist, the variance-minimizing strategies for some end dates are in conflict with one another. In these environments the decision-maker must trade off across possible end dates, where the optimal way of doing this in general depends on the stopping time distribution and details of the payoff function. Thus, the existence of a uniformly optimal strategy is key to guaranteeing the properties of history independence and robustness across decision problems that we have outlined for our solution. 

The question of whether a uniformly optimal strategy exists turns out to relate to a classic problem in consumer theory regarding the normality of demand---i.e., whether a consumer's demand for various goods is weakly increasing in income. In our setting, the decision-maker's ``utility function" is the negative of the posterior variance function, and his ``income" is the budget constraint on attention. When a uniformly optimal strategy exists, this means that the decision-maker's demand for information from each source is weakly increasing in his overall attention budget. One of our sufficient conditions for existence of a uniformly optimal strategy---``perpetual complementarity" of the different sources---directly connects to a known sufficient condition in the literature for normality of demand. We additionally utilize the specific structure of our problem to provide two new sufficient conditions, and show that all of these conditions can be simply stated in terms of the decision-maker's prior belief. See Section \ref{sec:proof sketch} for further details.

Beyond the specific statements of the results, a main contribution of this paper is demonstrating that in the present framework (i) the study of endogenous information acquisition is quite tractable, permitting explicit and complete characterizations; and (ii) there is enough richness in the setting to accommodate various economically interesting questions (e.g., about the role of primitives such as correlation across attributes). This makes the characterizations useful for deriving new substantive results in settings motivated by particular economic questions. We illustrate this with three applications, where we use our main results to generalize (Application 1) and complement (Application 3) existing results in the literature, as well as to solve for the equilibrium in a new game between competing news sources (Application 2). In all three applications, we discover new economic insights.

\paragraph{Application 1: Binary Choice.} A large literature in economics and neuroscience (originating with \citet{RatcliffMcKoon}) considers a consumer's decision process for choosing between two goods with unknown payoffs. Recently,  \citet{FudenbergStrackStrzalecki} proposed a model in which a decision-maker endogenously allocates attention across learning about two normally distributed, but i.i.d., payoffs. This model is nested in our framework. We use our main result to generalize \citet{FudenbergStrackStrzalecki}'s Proposition 3 and Theorem 5 beyond i.i.d. payoffs to settings with (1) correlation in the payoffs and (2) asymmetry in the consumer's initial knowledge about the two payoffs. These generalizations bring important realism to the setting, since correlation and asymmetry are common features of choice environments. We characterize the optimal attention allocation given an arbitrary normal prior about the payoffs, and show that \citet{FudenbergStrackStrzalecki}'s main economic insight regarding the relationship between choice speed and accuracy holds in this general setting.

\paragraph{Application 2: Biased News Sources.} In our next application, we consider a stylized game between a liberal and a conservative news source that report on a common unknown (e.g., the fiscal cost of a policy proposal), where their reporting is biased in opposite directions. The sources choose the size of their bias, as well as the informativeness of their reporting, and compete over readers' attention. Using our characterization of information acquisition, we are able to derive the complete time path of readers' attention allocations given the sources' choices of precision and bias levels, which allows us to characterize equilibrium news provision in this model. One particular insight that emerges from this analysis is that incentives for bias not only lead to greater polarization in equilibrium, but also lead to lower-quality news provision (i.e., larger noise choices). This analysis contributes to a literature about how competition across news sources affects the quality of news \citep{GentzkowShapiro,GalpertiTrevino,ChenSuen,PeregoYuksel}, where our work is distinguished in considering the role of the time path of information demand.

\paragraph{Application 3: Attention Manipulation.} Finally, we use our framework to study the dynamic consequences of a one-time attention manipulation.
Recently, \citet{GossnerSteinerStewart} studied this question in a model where a decision-maker chooses between goods with independent payoffs. Under some assumptions, they show that a one-time manipulation of attention towards a given good leads to persistently higher cumulative attention devoted to that good, and persistently lower cumulative attention to every other good. We derive a complementary result in our setting, focusing on how correlation across the unknown attributes affects the consequences of attention manipulation. We show that with two sources, \citet{GossnerSteinerStewart}'s insights hold under flexible patterns of correlation. On the other hand, with more than two sources, the nature of correlation matters. We identify a property of the prior belief under which all sources provide substitutable information, and show that under this property (but not in general), attention manipulation leads to persistently higher attention for that source and lower attention for others. 

\subsection{Related Literature} \label{sec:related lit}

Our work builds on a large literature regarding dynamic acquisition of information. One part of this literature considers choice between unconstrained information structures at entropic (or more generally, ``posterior-separable") costs, see e.g.\ \citet{Yang}, \citet{SteinerStewartMatejka}, \citet{HebertWoodford}, \citet{MorrisStrack}, and \citet{Zhong}.\footnote{It is interesting that \citet{SteinerStewartMatejka} also showed how the solution to their dynamic problem reduces to a series of static optimizations, similar to our multi-stage characterization. However, their argument is based on the additive property of entropy and differs from ours.} Under this modeling approach, the cost to acquiring information depends on the decision-maker's current belief, and is often interpreted as a mental processing cost \citep{RISurvey}. In contrast, a second set of papers---to which our paper belongs---models agents as dynamically allocating a fixed budget of resources across a prescribed (and finite) set of experiments, see e.g.\ \citet{CheMierendorff}, \citet{Mayskaya}, \citet{FudenbergStrackStrzalecki}, \citet{GossnerSteinerStewart}, and \citet{AzevedoEtAl}. These papers, and ours, assume that the cost of information is independent of what the decision-maker currently knows. We view such information costs as a better match for applications in which the cost to information acquisition is physical, e.g., a limit on the number of available COVID tests that can be administered in a given day.

Relative to this latter strand of literature, a distinguishing feature of our work is the presence of flexible correlation. Dynamic learning about correlated unknowns is generally intractable, so there has been relatively little work done in this area. An exception is a model introduced by \citet{Callander}, where the available signals are the realizations of a single Brownian motion path at different points, and the agent (or a sequence of agents) chooses myopically. This informational setting has since been extended by \citet{GarfagniniStrulovici}, which considers the optimal experimentation strategy for a forward-looking agent with acquisition costs, and \citet{Bardhi}, which introduces a potential conflict between an agent acquiring the information and a principal making the decision. These models differ from ours in that agents can perfectly observe any of an infinite number of attributes, and the correlation structure across the attributes is derived from a primitive notion of similarity or distance. We show that in a different model with  flexible correlation across a finite number of sources (and with noisy observations), it is sometimes possible to exactly characterize the optimal forward-looking solution.\footnote{Also related are  \citet{KlabjanOlszewskiWolinsky} and \citet{Sanjurjo}, which study learning about multiple attributes. Besides having noisy Gaussian signals, the main distinction of our informational setting is again that we allow for correlation across attributes and focus on what this correlation implies for the optimal learning strategy.}

Our work additionally connects to a large literature on sequential sampling in statistics and operations research. Since the information acquisition decisions in our model are not directly linked to flow payoffs, our model does not fall under the classic multi-armed bandit framework \citep{Gittins1979,Bergemann2008}. This feature also distinguishes our results relative to a classic literature on ``learning by experimentation" \citep{EasleyKiefer,Aghion,KellerRadyCripps}. The ``best-arm identification" problem \citep{Bubeck2009, Russo} is more closely related to us, as it considers a decision-maker who samples for a number of periods before selecting an arm and receiving its payoff. Indeed, the special case of two arms with jointly normal payoffs is nested in our framework under the case of two attributes and equal payoff weights. Our Theorem \ref{thm:K=2} thus builds on a prior result of \citet{FrazierEtAl2008}, which showed that myopic information acquisition---or the ``knowledge gradient" policy in the language of that literature---is optimal when the two arms have independent normal payoffs. Our result generalizes \citet{FrazierEtAl2008}'s result by allowing for correlated payoffs and a broader class of decision problems.

The best-arm identification problem between three or more arms falls outside of our framework, since payoffs in that problem depend on a multi-dimensional unknown.\footnote{With two arms, the difference in their payoffs is a sufficient statistic for choosing which arm is better. Such reduction to a one-dimensional unknown is not available in the case of many arms.} From a number of papers including \cite{ChickFrazier} and \cite{KeVillas-Boas}, it is well-understood that characterizing the optimal strategy in those problems is quite challenging (although \citet{FrazierEtAl2008} and \cite{FrazierEtAl2009} showed that the knowledge gradient policy performs well asymptotically). Our setting also involves multi-dimensional uncertainty, but we assume that the unknowns are linearly aggregated into a one-dimensional payoff-relevant variable. We show that under this restriction, exact characterization of the optimal strategy is feasible, and in fact it is the knowledge gradient policy (suitably defined in continuous time). We also discover new properties of the knowledge gradient policy in our continuous-time setting: In each of a finite number of stages, the policy attends to a fixed set of attributes with a constant ratio of attention, until this set expands and the next stage commences.



A key technical tool behind our characterization builds on a literature about the comparison of normal experiments. Following the classic work of \citet{Blackwell}, \citet{HansenTorgersen} showed that in a static decision problem, different normal signals about a one-dimensional payoff-relevant state can be Blackwell-ranked based on how much they reduce the variance of the state.  \cite{Greenshtein} subsequently derived comparisons between deterministic \emph{sequences} of conditionally independent normal signals about an unknown state. His Theorem 3.1 implies that one sequence Blackwell-dominates another if and only if it leads to lower posterior variances about the state at every time. Our Lemma \ref{lemma:Greenshtein} shows that \citet{Greenshtein}'s characterization is valid in a more general setting, in which time is continuous, and the sequence of signals can be chosen in a history-dependent manner (i.e., the first signal's realization can determine which signal is chosen next). 
\section{Model} \label{sec:model}

An agent has access to $K \geq 2$ \emph{information sources}, each of which is a diffusion process that provides information about an unknown \emph{attribute} $\theta_i \in \mathbb{R}$. The random vector $(\theta_1, \dots, \theta_K)$ is jointly normal with a known prior mean vector $\mu$ and prior covariance matrix $\Sigma$. We assume $\Sigma$ has full rank, so the attributes are linearly independent. 

As we describe in more detail below, the agent's decision depends on a \emph{payoff-relevant state} $\omega \in \mathbb{R}$. We assume the state is a linear combination of the attributes:
\begin{assumption} \label{assmp:linear}
$\omega = \alpha_1 \theta_1 + \dots + \alpha_K \theta_K$ for known weights $\alpha_1, \dots, \alpha_K \in \mathbb{R}$. 
\end{assumption}
\noindent It is equivalent (up to a constant) to assume that the vector $(\omega, \theta_1, \dots, \theta_K)$ is jointly normal, and that there is no residual uncertainty about $\omega$ conditional on the attribute values.\footnote{If $\omega, \theta_1, \dots, \theta_K$ are jointly normal, then $\omega$ can be rewritten as a linear combination of the $\theta_i$ plus a residual term that is independent of each $\theta_i$. The assumption of no residual uncertainty means that the residual term is a constant, returning Assumption \ref{assmp:linear} up to an additive constant (which can be normalized to zero in our problem).} Because any attribute value can be replaced with its negative, assuming $\alpha_i \geq 0$ is without loss. For ease of exposition, we will further assume that each weight $\alpha_i$ is strictly positive. 

Time is continuous, and the agent has a unit budget of attention to allocate at every instant of time. Formally, at each $t \in [0, \infty)$, the agent chooses an attention vector $\beta_1(t), \dots, \beta_K(t)$ subject to the constraints $\beta_i(t) \geq 0$ (attentions are positive) and $\sum_{i} \beta_i(t) \leq 1$ (allocations respect the budget constraint).

Attention choices influence the diffusion processes $X_1, \dots, X_K$ observed by the agent in the following way:
\begin{equation}\label{eq:diffusion processes}
dX^t_i = \beta_i(t) \cdot \theta_i \cdot dt + \sqrt{\beta_i(t)} \cdot d \mathcal{B}^t_i.
\end{equation}
Above, each $\mathcal{B}_i$ is an independent Brownian motion, and the term $\sqrt{\beta_i(t)}$ is a standard normalizing factor to ensure constant informativeness per unit of attention devoted to each source. Thus, devoting $T$ units of time to observation of source $i$ is equivalent to observation of the normal signal $\theta_i + \mathcal{N}(0,1/T)$, or $T$ independent observations of the standard normal signal $\theta_i + \mathcal{N}(0,1)$. This formulation treats ``attention" and ``time" in the same way, in the sense that devoting $1/2$ attention to source $i$ for a unit of time provides the same amount of information about $\theta_i$ as devoting full attention to source $i$ for a $1/2$ unit of time. We also note that since all sources are assumed to be equally informative about their corresponding attributes, it is \emph{with loss} to further normalize the payoff weights $\alpha_i$ to be equal to one another.\footnote{In fact, our subsequent results indicate that the case of equal weights is special. For example, with two sources, the conclusions of Theorem \ref{thm:K=2} always hold when $\alpha_1 = \alpha_2$, but do not hold in general.}



\begin{remark} \label{remark:discreteTime} 
As these comments suggest, there is a natural discrete-time analogue to our continuous-time model: At each period $t\in \mathbb{Z}_+$, the agent has a unit budget of precision to allocate across $K$ normal signals. Choice of attention vector $(\pi_1(t), \dots, \pi_K(t))$ results in one observation of the normal signal $\theta_i + \mathcal{N}\left(0, 1/{\pi_i(t)}\right)$ for each source $i=1, \dots, K$. All of our main results have an immediate corollary in that model. See Section \ref{sec:discussion} for further discussion.
\end{remark}

Let $(\Omega, \mathbb{P}, \{\mathcal{F}_t\}_{t\in \mathbb{R}_+})$ describe the relevant probability space, where the information $\mathcal{F}_t$ that the agent observes up to time $t$ is the collection of paths $\left\{X_i^{\leq t}\right\}_{i =1}^{K}$ (with $X_i^{\leq t}$ representing the sample path of $X_i$ from time $0$ to time $t$). An \emph{information acquisition strategy} $S$ is a map from $\left\{X_i^{\leq t}\right\}_{i, t}$ into $\Delta(\{1,\dots, K\})$, representing how the agent divides attention at each instant as a function of the observed diffusion processes.\footnote{We assume that given the agent's attention strategy, the stochastic differential equations in (\ref{eq:diffusion processes}) have a solution. This is true for example if each $\beta_i(t)$ is a deterministic function of $t$ (as in the optimal strategy that we describe in Theorems \ref{thm:K=2} and \ref{thm:general}), or if $\sqrt{\beta_i(t)}$ satisfies standard Lipschitz conditions (see Section 6.1 of \citet{YongZhou}).} In addition to allocating his attention, the agent chooses how long to acquire information for; that is, at each instant he determines (based on the history of observations) whether to continue acquiring information, or to stop and take an action. Formally, the agent chooses a \emph{stopping time} $\tau$, which is a map from $\Omega$ into $[0, +\infty]$ satisfying the measurability requirement $\{\tau \leq t\} \in \mathcal{F}_t$ for all $t$.  


At the endogenously chosen end time $\tau$, the agent will choose from a set of actions $A$ and receive the payoff $u(\tau, a,\omega$), where $u$ is a known payoff function that depends on the stopping time $\tau$, the action taken $a$ and the payoff-relevant state $\omega$. This formulation allows for additively separable waiting costs, $u(\tau, a, \omega) = u_1(a, \omega) - c(\tau)$, as well as geometric discounting, $u(\tau, a, \omega) = \delta^\tau \cdot u_2(a, \omega)$. The agent's posterior belief about $\omega$ at time $\tau$ determines the action that maximizes his expected flow payoff $\mathbb{E}[u(\tau, a,\omega)]$. We will only impose the following weak assumption on the payoff function: 
\begin{assumption} \label{assmp:waiting is costly}
Given any (normal) belief about $\omega$, $\max_{a} \mathbb{E}[u(\tau, a,\omega)]$ is decreasing in $\tau$. 
\end{assumption}
\noindent That is, we assume that holding fixed the agent's belief at the time of decision, an earlier decision is better. In the case of $u(\tau, a, \omega) = u_1(a, \omega) - c(\tau)$, this assumption requires the waiting cost $c(\tau)$ to be non-decreasing in $\tau$; in the case of $u(\tau, a, \omega) = \delta^\tau \cdot u_2(a, \omega)$, the assumption is that the optimal flow payoff $\max_{a} \mathbb{E}[u_2(a,\omega)]$ is non-negative (which is satisfied, for example, if there is a default action that always yields zero payoff).

To summarize, the agent chooses his information acquisition strategy and stopping time $(S,\tau)$ to maximize
$
 \mathbb{E}\left[ \max_a~ \mathbb{E}[u(\tau, a,\omega) \vert \mathcal{F}_\tau]\right].
$
In this paper we primarily focus on characterizing the optimal information acquisition strategy $S$. In general the strategies $S$ and $\tau$ should be determined jointly, but our results will show that in many cases the optimal $S$ can be characterized independently from the choice of $\tau$.

\section{Main Results}

In Section \ref{sec:K=2} we consider the case of two attributes, where we are able to derive a slightly stronger result. In Section \ref{sec:general K} we characterize the optimal attention allocation strategy for any finite number of attributes. All proofs appear in the appendix, and we provide an extended explanation of these results in Section \ref{sec:proof sketch}.

\subsection{Two Attributes} \label{sec:K=2}
Suppose there are two attributes $\theta_1$ and $\theta_2$, and the payoff-relevant state is $\omega = \alpha_1 \theta_1 + \alpha_2 \theta_2$, with each $\alpha_i > 0$. The agent's prior over the unknown attributes is 
$(\theta_1,\theta_2) \sim \mathcal{N}(\mu, \Sigma)$. Then the prior covariance between each attribute $i$ and the payoff-relevant state $\omega$ is
$cov_i:=\mbox{Cov}(\omega, \theta_i) = \alpha_i \Sigma_{ii} + \alpha_j \Sigma_{ji}$, 
and we assume that these covariances satisfy the following relationship:
\begin{assumption} \label{assmp:K=2}
$cov_1 + cov_2 = \alpha_1 (\Sigma_{11} + \Sigma_{12}) + \alpha_2 (\Sigma_{21} + \Sigma_{22}) \geq 0$.
\end{assumption}

\noindent Since both variances $\Sigma_{11},\Sigma_{22}$ are positive, this property holds if the covariance $\Sigma_{12}$ is not too negative relative to the size of either variance. If the weights on the two attributes are equal (i.e., $\alpha_1=\alpha_2$), this property holds for all prior beliefs over the attributes (since $2 \cdot \vert \Sigma_{12} \vert \leq 2\cdot \sqrt{\Sigma_{11} \cdot \Sigma_{22}} \leq \Sigma_{11} + \Sigma_{22}$). If the attributes are positively correlated ($\Sigma_{12}=\Sigma_{21} \geq 0$), then this property holds for all payoff weights $\alpha_1$ and $\alpha_2$. 


\begin{theorem}\label{thm:K=2}
Suppose $K = 2$ and Assumption \ref{assmp:K=2} is satisfied. W.l.o.g.\ let $cov_i \geq cov_j$.  Define
\[t_i^*:= \frac{cov_i-cov_j}{\alpha_j \det(\Sigma)}.\]
Then an optimal information acquisition strategy is history-independent and hence can be expressed as a deterministic path of attention allocations $(\beta_1(t),\beta_2(t))_{t\geq 0}$. This path consists of two stages: 

\begin{itemize}[noitemsep]
\item \textbf{Stage 1:} At all times $t < t_i^*$, the agent optimally allocates all attention to attribute $i$ (that is, $\beta_i(t)=1$ and $\beta_j(t)=0$). 
\item \textbf{Stage 2:} At all times $t \geq t_i^*$, the agent optimally allocates attention in the constant proportion
$(\beta_1(t),\beta_2(t))= \left(\frac{\alpha_1}{\alpha_1 + \alpha_2}, \frac{\alpha_2}{\alpha_1 + \alpha_2}\right)$.
\end{itemize}
\end{theorem}

There are two stages of information acquisition. In the first stage, which ends at some time $t^*_i$, the agent allocates all of his attention to the attribute $i$ with higher prior covariance with the payoff-relevant state. After time $t^*_i$, he divides his attention across the attributes in a constant ratio, where the long-run instantaneous attention allocation is proportional to the weights $\alpha$. Note that depending on the agent's stopping rule, Stage 2 of information acquisition may never be reached along some histories of the diffusion processes. But as long as the agent continues acquiring information, his attention allocations are as given above. In Online Appendix \ref{appx:uniqueness}, we show that under mild technical assumptions, the optimal attention strategy is in fact unique up to the stopping time $\tau$ (after which attention allocations obviously do not matter).

The characterization reveals that the optimal information acquisition strategy is completely determined from the prior covariance matrix $\Sigma$ and the payoff weight vector $\alpha$ (the prior mean vector $\mu$ does not play a role). In particular, the strategy does not depend on details of the agent's payoff function $u(\tau, a, \omega)$, including his time preferences. When the prior belief satisfies Assumption \ref{assmp:K=2}, the optimal information acquisition strategy is thus constant across different objectives and also across different stopping rules. Relatedly, as we demonstrate in Section \ref{sec:binary choice}, we can solve for the optimal stopping rule in this setting as if information acquisition were \emph{exogenously} given by Theorem \ref{thm:K=2}. In Online Appendix \ref{appx:counterexample}, we provide an example to illustrate that these properties can fail when Assumption \ref{assmp:K=2} is violated. Online Appendix \ref{appx:K=2 necessity} further shows that for the case of two attributes, Assumption \ref{assmp:K=2} is not only sufficient but also necessary for our characterization to hold independently of the agent's payoff criterion.

Below we illustrate this optimal strategy using a few simple examples. 

\begin{example}[Independent Attributes] \label{ex:independent} Suppose $(\theta_1,\theta_2) \sim \mathcal{N}(\mu, \left(\begin{smallmatrix} 6 & 0 \\ 0 & 1 \end{smallmatrix} \right))$ and the payoff-relevant state is $\omega = \theta_1 + \theta_2$. Then, applying Theorem \ref{thm:K=2}, the agent initially puts all attention towards learning $\theta_1$. At time $t_1^*= \frac{5}{6}$, his posterior covariance matrix is the identity matrix. After this time he optimally splits attention equally between the two attributes, which are now symmetrically distributed.
\end{example}

\smallskip

\begin{example}[Correlated Attributes]\label{ex:equalWeights} Suppose 
$(\theta_1,\theta_2) \sim \mathcal{N}(\mu, \left(\begin{smallmatrix} 6 & 2 \\ 2 & 1\end{smallmatrix}\right))$, and the payoff-relevant state is still $\omega =\theta_1 + \theta_2$. Applying Theorem \ref{thm:K=2}, the agent initially puts all attention towards learning $\theta_1$. At time $t_1^*=\frac{5}{2}$, his posterior covariance matrix becomes $ \left(\begin{smallmatrix} 3/8 & 1/8 \\ 1/8 & 3/8 \end{smallmatrix} \right)$. Compared to the previous example, it takes longer for the agent's uncertainty about the two attributes to equalize, since information about $\theta_1$ also reduces the agent's uncertainty about $\theta_2$. After $t_1^*=\frac{5}{2}$ he optimally splits attention equally between the two attributes.
\end{example}

\smallskip

\begin{example}[Unequal Payoff Weights]\label{ex:unequalWeights} Consider the prior belief given in the previous example, but suppose now that the payoff-relevant state is $\omega =\theta_1 + 2\theta_2$. As before, the agent initially puts all attention towards learning $\theta_1$. Stage 1 ends at time $t_1^*=\frac{3}{2}$, when the posterior covariance matrix is  $\left(\begin{smallmatrix} 3/5 & 1/5 \\ 1/5 & 2/5 \end{smallmatrix} \right)$. Note that because of the asymmetry in the payoff weights, the agent's posterior uncertainty about the two attributes is not the same at this switch point. As we will discuss in Section \ref{sec:proof sketch}, however, the \emph{marginal values} of learning about the two attributes are equal to one another at time $t_1^*$. After this time, the agent optimally acquires information in the mixture $(1/3, 2/3)$.
\end{example}

\smallskip

\subsection{$K$ Attributes} \label{sec:general K}
We now consider the case of multiple attributes. We provide three different sufficient conditions under which the optimal information acquisition strategy can be exactly characterized. 

\begin{assumption}[Perpetual Substitutes] \label{assmp:Substitutes} $\Sigma^{-1}$ has non-positive off-diagonal entries. 
\end{assumption}

\begin{assumption}[Perpetual Complements] \label{assmp:Complements} $\Sigma$ has non-positive off-diagonal entries and $\Sigma \cdot \alpha$ has non-negative coordinates.
\end{assumption}

\begin{assumption}[Diagonal Dominance]\label{assmp:diagonal dominance}
$\Sigma^{-1}$ is diagonally-dominant. That is, 
$
[\Sigma^{-1}]_{ii} \geq \sum_{j \neq i} \vert [\Sigma^{-1}]_{ij} \vert$ for all $ 1 \leq i \leq K.
$
\end{assumption}

\begin{figure}[h]
    \centering
    \includegraphics[scale=0.4]{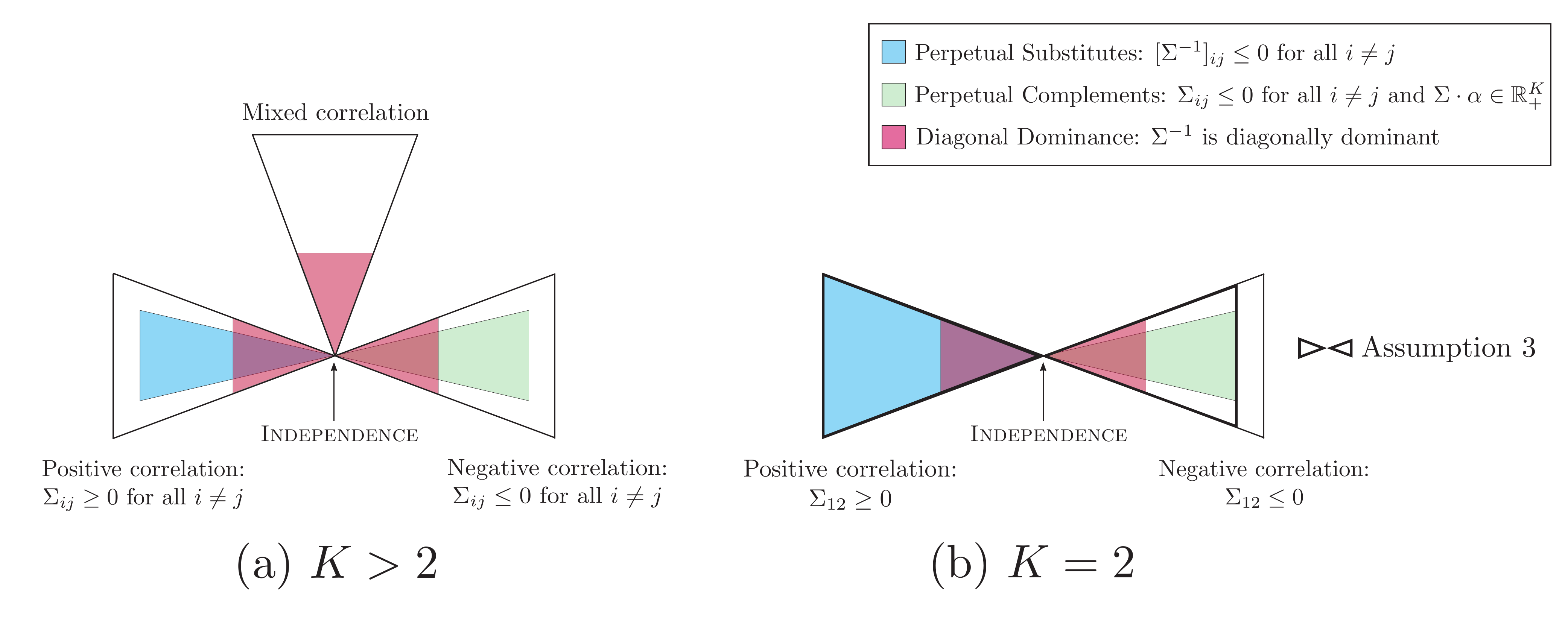}
    \caption{Relationship between Assumptions \ref{assmp:K=2} - \ref{assmp:diagonal dominance}.}
    \label{fig:Venn}
\end{figure}

These conditions, their relationship to one another, and their relationship to our previous Assumption \ref{assmp:K=2} for the case of $K=2$, are depicted in Figure \ref{fig:Venn}. Assumption \ref{assmp:Substitutes} generalizes a previous sufficient condition $\Sigma_{12} \geq 0$ to more than two attributes. It requires that the \emph{partial correlation} between any two attributes---controlling for all other attributes---is positive.\footnote{The partial correlation between attributes $\theta_i$ and $\theta_j$ under covariance matrix $\Sigma$ is equal to $-[\Sigma^{-1}]_{ij}/\sqrt{[\Sigma^{-1}]_{ii} [\Sigma^{-1}]_{jj}}$.}$^,$\footnote{Positive-definite matrices with non-positive off-diagonal entries are known as $M$-matrices \citep{Plemmons1977}. It is well-known that the inverse of such matrices has positive entries everywhere. Thus Assumption \ref{assmp:Substitutes} implies positive correlation $\Sigma_{ij} \geq 0$, but is strictly stronger when $K > 2$. We mention that assuming $\Sigma_{ij} \geq 0$ for all $i \neq j$ does not guarantee Theorem \ref{thm:general} to hold. A counterexample with three sources is presented in Online Appendix \ref{appx:K=3 counterexample}.} Proposition \ref{prop:substitutes} in Online Appendix \ref{appx:substitutes/complements} shows that Assumption \ref{assmp:Substitutes} is an if and only if condition for information from any pair of sources to be ``perpetually substitutable." By this we mean that the value of acquiring information from one source is decreasing in the amount of information from the other source, and that this property holds not only at time $t=0$ but also along any path of information acquisitions.

Assumption \ref{assmp:Complements} imposes that all attributes are negatively correlated with one another, but that each attribute is initially positively correlated with the payoff-relevant state. Under this assumption, prior covariances are ``mildly" negative. Proposition \ref{prop:complements} in Online Appendix \ref{appx:substitutes/complements} shows that Assumption \ref{assmp:Complements} is an if and only if condition for information from any pair of sources to be perpetually complementary, in the sense that the value of acquiring information from one source is increasing in the amount of information from the other source, again along the entire path of information acquisitions.

The last Assumption \ref{assmp:diagonal dominance} requires the inverse of the prior covariance matrix to be diagonally-dominant. Roughly speaking, this assumption allows for some pairs of attributes to be complements and other pairs to be substitutes, but puts an upper bound on the magnitude of any complementarity or substitution effects. When $K = 2$, Assumption \ref{assmp:diagonal dominance} reduces to the simple condition $\vert \Sigma_{12} \vert \leq \min \{\Sigma_{11}, \Sigma_{22}\}$, which is sufficient for our previous Assumption \ref{assmp:K=2}. For general $K$, Assumption \ref{assmp:diagonal dominance} is implied by a similar condition $\vert \Sigma_{ij} \vert \leq \frac{1}{2K-3} \Sigma_{ii}$ for all $i \neq j$ (see Appendix \ref{appx:2K-3 implies diagonal dominance}) We explain the role of these assumptions in Section \ref{sec:proof sketch}.

\begin{theorem}\label{thm:general}
Suppose any of Assumption \ref{assmp:Substitutes}, \ref{assmp:Complements}, or \ref{assmp:diagonal dominance} is satisfied.\footnote{We point out that while Assumption \ref{assmp:Substitutes}, \ref{assmp:Complements}, and \ref{assmp:diagonal dominance} are each sufficient for the theorem to hold, they are not in general necessary (unlike Assumption \ref{assmp:K=2} in the $K=2$ case).} Then, there exist times 
\[0 = t_0 < t_1 < \cdots < t_m = +\infty\]
and nested sets 
\[\emptyset \subsetneq B_1  \subsetneq \cdots \subsetneq B_m = \{1, \dots, K\},\]
such that an optimal information acquisition strategy is described by a deterministic path of attention allocations $(\beta(t))_{t \geq 0}$. This attention path consists of $m \leq K$ stages: For each $1 \leq k \leq m$, $\beta(t)$ is constant at all times $t \in [t_{k-1}, t_k)$ and supported on the sources in $B_k$. In particular, the optimal attention allocation at any time $t \geq t_{m-1}$ is proportional to $\alpha$. 
\end{theorem}

The full path of attention allocations $(\beta_1(t),\dots,\beta_K(t))$ (including the times $t_k$, the nested sets $B_k$, and the constant attention allocation at each stage $k$) can again be determined directly from the primitives $\Sigma$ and $\alpha$. In Appendix \ref{appx:algorithm}, we provide an algorithm for computing this path. Theorem \ref{thm:general} thus tells us that the agent can reduce the dynamic information acquisition problem to a sequence of $m \leq K$ static problems, each of which involves finding the optimal constant division of attention for a fixed period of time (from $t_{k-1}$ to $t_k$). Moreover, as in the $K=2$ case, the optimal information acquisition strategy does not depend on the agent's payoff function, and is history-independent. 

We note that each of Assumption \ref{assmp:Substitutes}, \ref{assmp:Complements}, or \ref{assmp:diagonal dominance} is ``absorbing" in the following sense: If the prior covariance matrix satisfies any of these conditions, then so does any posterior covariance matrix. Propositions \ref{prop:substitutes} and \ref{prop:complements} respectively show that this is true for Assumptions \ref{assmp:Substitutes} and \ref{assmp:Complements}, whereas diagonal dominance is absorbing because any information acquisition strategy only increases the diagonal entries of the precision matrix $\Sigma^{-1}$.\footnote{By directly computing $cov_1$ and $cov_2$ at posterior beliefs, it can be shown that our previous Assumption \ref{assmp:K=2} for the case of two attributes is also absorbing.} The absorbing property implies that our characterization not only applies to the prior belief, but also to \emph{any} posterior belief even if the history involves sub-optimal attention allocations. This feature enables us to study the effect of attention manipulation, an application that we pursue in Section \ref{sec:manipulation}.


Finally, we mention that starting from any prior belief (including those that fail all three of the sufficient conditions we have provided), so long as the agent does not stop learning about any attribute, his posterior belief must eventually satisfy Assumption \ref{assmp:diagonal dominance}.\footnote{This is because any attention devoted to attribute $i$ simply increases the $i$-th diagonal entry of the precision matrix $\Sigma^{-1}$. With sufficient attention devoted to each attribute, the posterior precision matrix must be diagonally-dominant.} Theorem \ref{thm:general} applies at these posterior beliefs, implying in particular that the agent's optimal attention allocation is eventually constant and proportional to the weight vector $\alpha$.



\section{Explanation of Results}\label{sec:proof sketch}

\subsection{Fixed Stopping Time $t$} \label{sec:prelim}

Consider the simpler problem in which the agent makes a decision at an exogenously fixed and known time $t$. 
Because normal signals can be completely Blackwell-ordered based on their precisions \citep{HansenTorgersen}, different mixtures over the sources can be compared based on how much they reduce the variance of the payoff-relevant state. 

Formally, the agent's past attention allocations integrate to a \emph{cumulative attention vector}
$
q(t)=(q_1(t) \dots, q_K(t))' \in \mathbb{R}_+^K
$
at time $t$, describing how much attention has been paid to each attribute thus far. The agent's posterior variance of $\omega$ is
\begin{equation}\label{eq:V(q)}
V(q) = \alpha'(\Sigma^{-1} + \diag(q(t)))^{-1}\alpha
\end{equation}
where $\Sigma$ is the prior covariance matrix over the attribute values, and $\diag(q(t))$ is the diagonal matrix with entries $q_1(t), \dots, q_K(t)$. This posterior variance depends solely on the payoff weights $\alpha$, the prior covariance matrix $\Sigma$, and the cumulative acquisitions $q(t)$. It does not depend on the realizations of the diffusion processes or the order of information acquisitions. So the problem of optimizing for a fixed end date $t$ reduces to a static problem of optimally allocating $t$ units of attention. 


 
Define the \emph{$t$-optimal} attention vector
\[
n(t) = \argmin_{q_1, \dots, q_K \geq 0, ~ \sum_i q_i = t} V(q_1, \dots, q_K).
\]
to be the allocation of $t$ units of attention that minimizes posterior variance of $\omega$, which is unique by Lemma \ref{lemma:t-opt unique} in the appendix. Every information acquisition strategy that cumulates to the attention vector $n(t)$ at time $t$ is optimal for any decision problem at that time.

\subsection{Uniformly Optimal Strategy: Definition} 


In general, the family of solutions $\{n(t)\}_{t \geq 0}$ corresponding to optimal allocation of $t$ units of attention does not determine the solution to the dynamic problem. To see this, suppose $n(1)=(0,1)$, so that the optimal way to allocate attention for a decision at time $t=1$ is to allocate it all to the second attribute, but $n(2)=(2,0)$, so that the optimal way to allocate attention for a decision at time $t=2$ is to allocate it all to the first attribute. Clearly these $t$-optimal attention vectors are incompatible under a single dynamic strategy: Optimal attention allocation for a decision at time $t=1$ precludes achieving the optimal attention allocation for time $t=2$. So when the agent faces the possibility of stopping at either time $t = 1$ or $t = 2$ (depending on the realizations of the diffusion processes), he must trade off between achieving more precise information about the payoff-relevant state at either time. The optimal resolution of this trade-off depends on the specific decision problem.

If, however, it were possible to continuously string together the $t$-optimal attention vectors $n(t)$ along the path of one information acquisition strategy, then such intertemporal trade-offs would not arise, and we might further conjecture such a strategy to be optimal for all stopping problems. This turns out to be true.

\begin{definition}
Say that an information acquisition strategy $S$ is \emph{uniformly optimal} if it is deterministic (i.e., history-independent) and its induced cumulative attention vector at each time $t$ is the $t$-optimal vector $n(t)$.
\end{definition}

A uniformly optimal strategy, by definition, minimizes posterior variance at every instant $t$. A result of \citet{Greenshtein} implies that such a strategy is best among all history-independent information acquisition strategies. In Lemma \ref{lemma:Greenshtein}, we extend \citet{Greenshtein}'s result to show that a uniformly optimal strategy is best among \emph{all} strategies, including those that condition future attention allocations on past signal realizations. In brief, we first observe that compared to any alternative information acquisition strategy, the uniformly optimal strategy achieves the same precision of beliefs about the state at earlier times. We then use this observation to show that any decision rule (i.e., any stopping time and final action) achievable under the alternative strategy can be replicated under the uniformly optimal strategy in a way that makes the agent stop earlier, but maintains his belief at the time of stopping. This ``replicating" decision rule, together with the uniformly optimal attention strategy, thus yields a higher expected payoff. We note that this proof relies crucially on the normal environment, which allows us to capture the agent's uncertainty through the single statistic of posterior variance. 

Thus,  whenever a uniformly optimal strategy exists, it must be the optimal strategy in our problem.\footnote{Our argument based on Blackwell comparisons gets to the optimal policy (i.e., attention allocation) without going through the HJB equation and value function, which may be difficult to solve for explicitly.} It remains to show that under the assumptions we provided previously, a uniformly optimal strategy does exist, and has the structure described in Theorems \ref{thm:K=2} and \ref{thm:general}.



\subsection{Uniformly Optimal Strategy: Existence}

\begin{lemma}[Monotonicity] \label{lemma:monotonicity}
A uniformly optimal strategy exists if and only if the $t$-optimal attention vector $n(t)$ weakly increases (in each coordinate) over time.
\end{lemma}

In words, a uniformly optimal strategy exists if and only if for every $t' > t$, the optimal allocation of $t'$ units of attention devotes a (weakly) higher amount of attention to \emph{each} source compared to the optimal allocation of $t$ units. Thus, monotonicity of $n(t)$ and existence of a uniformly optimal strategy are equivalent.

Whether $n(t)$ is monotone turns out to be closely related to a classic problem in consumer theory: Suppose a consumer has a utility function $U(q_1, \dots, q_K)$ over consumption of $q_k$ units of each of $K$ goods, and let $D(\mathbf{p}, w)$ denote his Marshallian demand subject to the budget constraint $\mathbf{p} \cdot \mathbf{q} \leq w$. Then, the consumer's demand is \emph{normal} if each coordinate of $D(\mathbf{p},w)$ increases with income $w$. In our setting, we can set $U = -V(q_1, \dots, q_K)$ to be the negative of the posterior variance, so that minimizing $V$ is the same as maximizing $U$. Our $t$-optimal attention vector $n(t)$ is then precisely the Marshallian demand, when prices are identically equal to $1$ and income is equal to $t$. Thus, normality of the consumer's demand under utility function $U=-V$ is equivalent to existence of a uniformly optimal strategy.

A necessary and sufficient condition on $U$ for normality of demand is given in \citet{AlarieBronsard} and \citet{BilanciniBoncinelli}. When there are just two goods, demand is normal if and only if
$\partial_{ij} U \cdot \partial_{j} U \geq \partial_{i} U \cdot \partial_{jj} U$ for $i, j \in \{1, 2\}$. In our setting, this condition can be simplified to our Assumption \ref{assmp:K=2}, which as discussed is necessary and sufficient for a uniformly optimal strategy to exist when $K = 2$ (although our proof of Theorem \ref{thm:K=2} is more direct, as it characterizes $n(t)$ explicitly).

With many goods, this necessary and sufficient condition on $U$ is  too complex to reduce to a statement on the primitives in our problem ($\Sigma$ and $\alpha$). However, a useful sufficient condition for normality of demand has been given by \citet{Chipman} and more recently generalized by \citet{Quah}: If $U$ is concave and super-modular (also called ``ALEP-complementary" to distinguish from Marshallian and Hicksian complementarity), then the consumer's demand is normal. In our environment, $U = -V$ is concave (see Lemma \ref{lemma:V partials} in the appendix and the subsequent corollary). Super-modularity of $U$ requires the cross-partials of $V$ to be negative at every posterior belief, namely that the different sources are \emph{complements}. Our Proposition \ref{prop:complements} in Online Appendix \ref{appx:substitutes/complements} shows that it is in fact possible to reduce supermodularity of $-V$ to a condition on the prior belief only. This is our first sufficient condition in Assumption \ref{assmp:Complements}, ``perpetual complementarity."

Generalizing beyond this condition, we use the special form of $V$ to develop alternative sufficient conditions. Perhaps surprisingly, we show that if the sources are substitutes at every posterior belief (i.e., $-V$ is sub-modular), then $n(t)$ is also monotone. This ``perpetual substitution" property, too, can be stated as a simple condition on the prior covariance matrix (our Assumption \ref{assmp:Substitutes}). Additionally, if correlation is not too strong---as implied by Assumption \ref{assmp:diagonal dominance}, which bounds the size of the covariances between the attributes relative to the size of their variances---then again we obtain monotonicity. In our proof of Theorem \ref{thm:general}, we show that these three different economic conditions are each sufficient to imply the following technical property: At every moment of time, those attributes that covary most strongly with the payoff-relevant state $\omega$ all have \emph{positive} covariance with $\omega$ (Lemma \ref{lemma:gamma positive}). As demonstrated in Lemma \ref{lemma:final stage}, this key technical property delivers the monotonicity of $n(t)$. Further exploration of our sufficient conditions, and whether they imply results for consumer demand theory beyond our specific setting, is an interesting topic for subsequent work.  

\subsection{Uniformly Optimal Strategy: Structure}

When $n(t)$ is indeed monotone in $t$, the uniformly optimal strategy that achieves these vectors is simply the one that sets each attention allocation $\beta(t)$ to be the time-derivative of $n(t)$. Under this strategy, the agent divides attention at every moment across those attributes that maximize the \emph{instantaneous} marginal reduction of posterior variance $V$. While the set of attributes is pinned down by first-order conditions, the precise mixture over those attributes is pinned down by second-order conditions, which ensure that this set of attributes continue to have equal and marginal values at future instants. Specifically, for each set of attributes, there is a unique linear combination that corresponds to the ``learnable component" of $\omega$ given those attributes (formally, it is a projection of $\omega$). It turns out that at each stage, it is optimal to acquire information in a mixture proportional to the weights of this linear combination, thus producing an unbiased signal about the learnable component of $\omega$. In the final stage, when the agent pays attention to every attribute---so that the learnable component of $\omega$ is $\omega$ itself---the optimal mixture is proportional to the payoff weights $\alpha$. 

As beliefs about a set of attributes become more precise, their shared marginal value decreases continuously relative to the marginal value of learning about other attributes. Eventually the marginal value of learning about some other attribute ``catches up" and joins the set of maximizers. This yields the nested-set property in Theorems \ref{thm:K=2} and \ref{thm:general}.

\section{Applications}

The characterizations in Theorems \ref{thm:K=2} and \ref{thm:general} suggest that the study of dynamic information acquisition in our setting is quite tractable. We now apply this characterization to derive new results in a diverse set of applications.

In Section \ref{sec:binary choice}, we consider optimal attention allocation for choice between two goods and generalize recent results from \citet{FudenbergStrackStrzalecki}. In Section \ref{sec:media} we develop a game between biased news sources, and characterize the degree of polarization and the quality of information in equilibrium. In Section \ref{sec:manipulation}, we consider the dynamic implications of a one-time attention manipulation, complementing a recent exercise in \cite{GossnerSteinerStewart}. The applications in Sections \ref{sec:binary choice} and \ref{sec:manipulation} show that we can use our main results to tractably introduce correlation in settings that have been previously studied under strong assumptions of independence. Our applications in Sections \ref{sec:binary choice} and \ref{sec:media} show how our main results can be used as an intermediate step to derive results about other economic behaviors.  

\subsection{Application 1: Binary Choice} \label{sec:binary choice}

Building on a large literature regarding ``binary choice" problems, \citet{FudenbergStrackStrzalecki} (henceforth FSS) recently proposed the \emph{uncertain drift-diffusion} model: 
An agent has a choice between two goods with unknown payoffs $v_1$ and $v_2$, and can learn about those payoffs before making a choice. The two payoffs are jointly normal and i.i.d.:
\begin{equation} \label{exp:prior}
(v_1,v_2)' \sim \mathcal{N}\left(\mu, \left(\begin{array}{cc}
\sigma^2 & 0 \\
0 & \sigma^2 \end{array}\right)\right)
\end{equation}
The agent observes two Brownian processes whose drifts are the unknown payoffs. He then chooses a stopping time $\tau$ to maximize the objective
\[
\mathbb{E}[\mathbb{E}[\max\{v_1, v_2\} \mid \mathcal{F}_{\tau} ] - c\tau ],
\]
where  $c\tau$ is a linear waiting cost. 


FSS's main economic insight is that earlier decision times are associated with more accurate decisions. Formally, let $p(t)$ be the probability of choosing the higher-value good conditional on stopping at time $t$. FSS's Proposition 3 shows that $p(t)$ is monotonically (weakly) decreasing over time. This comparative static is not obvious because two forces push in opposite directions: On the one hand, the agent has more information at later times, suggesting that later decisions may be more accurate. On the other hand, because the stopping time is endogenously chosen, the agent is more likely to stop earlier when the decision is easy (i.e., when one good's value is much higher than the other). FSS's result implies that this second force dominates.

FSS show moreover that this result is robust to endogenous attention allocation under a budget constraint. Specifically, suppose that at each moment of time, the agent has one unit of attention to allocate across learning either $v_1$ or $v_2$. Then, FSS's Theorem 5 shows that the agent optimally divides attention equally between learning about the two payoffs at every moment of time, similar to the exogenous process specified in their main model.

FSS's model of endogenous attention and binary choice is nested within our framework. To map this setting back into our main model, define $\theta_1 = v_1$ and $\theta_2 = -v_2$. Then, since the payoff difference $v_1 - v_2$ is a sufficient statistic for the agent's decision, this problem corresponds to payoff-relevant state $\omega = v_1 - v_2 = \theta_1 + \theta_2$ in our framework. 

We now show that we can use our results to go beyond the case of independent and identically distributed payoffs (as imposed in the prior in (\ref{exp:prior})). Different from FSS, suppose that the agent's prior over $(\theta_1, \theta_2)$ is normal with an arbitrary covariance matrix
$\Sigma := \left(\begin{array}{cc}
\Sigma_{11} & \Sigma_{12} \\
\Sigma_{21} & \Sigma_{22} \end{array}\right)$. Since the payoff weights are $\alpha_1 = \alpha_2 = 1$, our Theorem \ref{thm:K=2} applies and characterizes the agent's optimal attention allocations over time given any $\Sigma$. The following corollary is an immediate generalization of FSS's Theorem 5:

\begin{corollary}\label{corr:binary choice attention}
Suppose $\Sigma_{11} \geq \Sigma_{22}$. The agent's optimal information acquisition strategy $(\beta_1(t),\beta_2(t))_{t \geq 0}$ in this binary choice problem consists of two stages: 
\begin{itemize}[noitemsep]
\item \textbf{Stage 1:} At all times 
\[
t< t_1^* = \frac{\Sigma_{11} - \Sigma_{22}}{\det(\Sigma)},
\]
the agent optimally allocates all attention to $\theta_1$.
\item \textbf{Stage 2:} At times $t \geq t_1^*$, the agent optimally allocates equal attention to $\theta_1$ and $\theta_2$.
\end{itemize}
\end{corollary}


Thus, when the agent is initially more uncertain about one of the two payoffs, he spends a period of time exclusively learning about that payoff. Starting at time $t_1^*$, the agent divides attention equally across learning about the two goods, as in FSS's i.i.d.\ case. From the closed-form expression for $t_1^*$, it is straightforward to show that the length of the first stage is increasing in the asymmetry of initial uncertainty, and also increasing in the degree of correlation between the two payoffs.

We now use this characterization of optimal attention allocation to further generalize FSS's main economic insight regarding the relationship between choice speed and accuracy. 

\begin{proposition}\label{prop:binary choice FSS}
For any $\Sigma$, $p(t)$ is weakly decreasing in $t$. Thus, choice accuracy is weakly higher at earlier stopping times.
\end{proposition}

The logic of this result is as follows (see Online Appendix \ref{appx:binary choice} for detailed analysis). First, Corollary \ref{corr:binary choice attention} implies that we can separate the problem of optimal stopping from the problem of optimal information acquisition. That is, we can take information as exogenously given by the process described in Corollary \ref{corr:binary choice attention}, and characterize properties of optimal stopping within this problem. In particular, Corollary \ref{corr:binary choice attention} pins down the evolution of the agent's posterior covariance matrix $\Sigma_t$, which will be important for the subsequent arguments.

While $\Sigma_t$ is a deterministic process, the agent's posterior expectation for the payoff difference $\theta_1+\theta_2$ evolves according to a random process $Y_t$. As in FSS, the symmetric stopping boundary at time $t$ is given by a function $
k^*(\Sigma_t)$ of the agent's posterior covariance matrix $\Sigma_t$; that is, the agent optimally stops at time $t$ if and only if $\vert Y_t \vert \geq k^*(\Sigma_t)$. Given these stopping boundaries, the choice accuracy $p(t)$ conditional on stopping at time $t$ has the following simple form: 
\begin{equation} \label{eq:pt}
p(t) = \Phi\left(\frac{k^*(\Sigma_t)}{\sigma_t}\right),
\end{equation}
where $\sigma_t^2$ is the agent's posterior variance of $\theta_1 +\theta_2$ at time $t$, and $\Phi$ is the normal c.d.f.\ function. 

So it remains to understand how (\ref{eq:pt}) evolves. There are two forces, which turn out to go in the same direction. First, uncertainty about the payoff difference $\theta_1 + \theta_2$ decreases over time. As FSS already showed in the i.i.d.\ case, this effect weakly decreases the ratio $k^*(\Sigma_t)/\sigma_t$. Roughly speaking, stopping at an earlier time when there is more residual uncertainty requires the agent to have received disproportionately stronger signals to forgo the option value. In our Lemma \ref{lemma:binary choice uncertainty change}, we generalize this insight to arbitrary prior beliefs.


Second, our characterization in Corollary \ref{corr:binary choice attention} reveals that the optimal attention strategy continuously reduces the asymmetry in uncertainty between the two attributes. We show in Lemma \ref{lemma:binary choice asymmetry change} that holding fixed uncertainty about the sum $\theta_1 + \theta_2$, asymmetry in uncertainty about the two attributes $\theta_1$ and $\theta_2$ allows the agent to learn faster; that is, the agent has lower posterior variance of $\theta_1 + \theta_2$ at every future time compared to a symmetric prior.\footnote{A simple informal argument is as follows. Given any prior, we can upper-bound the posterior variance under the optimal strategy by the posterior variance under a strategy that devotes equal attention to $\theta_1$ and $\theta_2$ at all times, which is equivalent to receiving the signals $\theta_1 + \mathcal{N}(0,2)$ and $\theta_2 + \mathcal{N}(0,2)$ over every unit of time. The information that those signals provide about the sum $\theta_1 + \theta_2$ is \emph{at least} as informative as the single signal $\theta_1 + \theta_2 + \mathcal{N}(0,4)$, which is the sum of the previous two signals. Thus the agent learns about $\theta_1 + \theta_2$ at least as fast as if he received the signal $\theta_1 + \theta_2 + \mathcal{N}(0,4)$ over every unit of time. When the prior is symmetric, then the equal-attention strategy considered above is optimal, and this lower bound on the speed of learning is tight. But when the prior is asymmetric, then the agent can improve upon this bound. This suggests that the agent can learn more quickly under an asymmetric prior than a symmetric one, holding fixed his prior uncertainty about $\theta_1 + \theta_2$.} So asymmetric uncertainty increases the option value to waiting, and thus also the stopping boundary relative to the symmetric baseline. This effect, too, causes the ratio $k^*(\Sigma_t)/\sigma_t$ to decrease over time. Combining both effects yields the result that $p(t)$ decreases over time.

In fact, using similar arguments, we can further generalize Proposition \ref{prop:binary choice FSS} to asymmetric learning speeds about the two unknown payoffs. See Online Appendix \ref{appx:binary choice informativeness}.  



\subsection{Application 2: Biased News Sources} \label{sec:media}

Next, we apply our results to study information provision in a setting with strategic information providers. Specifically, we are interested in how incentives for news bias interact with incentives for high-quality news provision. 

To study this, we consider a stylized game in which  a representative news reader seeks to learn a payoff-relevant unknown $\omega \sim \mathcal{N}(\mu_\omega, \sigma_\omega^2)$, for example the expected fiscal cost of a certain policy proposal. Two sources $i\in \{1,2\}$ (a left-leaning and right-leaning news source) each report on $\omega$, but bias their reporting in the direction helpful to their political party. We consider issues where the partisan implications are not precisely known by the general public (although they are understood by the sources)---for example, new limits on short selling in financial markets or trade deals with countries in Southeast Asia. We define $b$ to be the \emph{benefit} to source 1's party when the reader believes that $\omega$ is large (i.e., the cost of the policy is high). This benefit $b$ is a random variable from the perspective of the reader, and we assume it has distribution $b \sim \mathcal{N}(\mu_b, \sigma_b^2)$. The larger $\sigma_b$ is, the less well understood the partisan implications of the issue are.

Source 1 biases its reporting of $\omega$ in the direction of $b$, and source 2 biases away from it, but these sources can choose the intensity of their biases $\phi_i > 0$. The sources also choose the quality of their reporting, parametrized by $\zeta_i > 0$, where a larger $\zeta_i$ represents greater noise in the reporting. Specifically, we assume that a unit of time spent on source $1$ is informationally equivalent to a realization of 
\[Z_1 \sim \mathcal{N}(\omega + \phi_1 b, \zeta_1^2),\]
while a unit of time spent on source 2 is informationally equivalent to a realization of 
\[Z_2 \sim \mathcal{N}(\omega - \phi_2 b, \zeta_2^2).\]
Both choices $\phi_i$ and $\zeta_i$ are fixed across time. For example, a source may develop a reputation for providing very biased information but having high-quality reporting.

We suppose that the reporters and editors at these news sources are subject to soft pressure to bias their reporting, and model the strength of those pressures by a payoff of
$-\lambda (\phi_i - \kappa)^2$
for choice of bias intensity $\phi_i$. The parameters $\lambda \in \mathbb{R}_+$ and $\kappa \in \mathbb{R}_+$ respectively determine the strength of incentives for bias, and the bliss point for bias intensity. We consider an interior bliss point $\kappa$ to be realistic, since very biased sources lose credibility.

Sources are not explicitly incentivized to provide informative news, but will do so in equilibrium to attract attention from the news reader whose attention allocations we now describe. For any given choices of precisions and bias intensities, the reader's optimal attention allocations $(\beta_1(t), \beta_2(t))_{t \geq 0}$ can be derived from the characterization in Theorem \ref{thm:K=2}. To apply this theorem, we can transform the current setting to our main model: Define $\theta_1  = \frac{1}{\zeta_1}(\omega + \phi_1 b)$ and $\theta_2 = \frac{1}{\zeta_2}(\omega - \phi_2 b)$, so that a unit of time spent on each source $i$ produces an equally informative (standard normal) signal about $\theta_i$. The payoff-relevant state can be rewritten as $\omega = \alpha_1 \theta_1 + \alpha_2 \theta_2$ with payoff weights $\alpha_1  = \zeta_1 \cdot \frac{\phi_2}{\phi_1 + \phi_2}$ and $\alpha_2 = \zeta_2 \cdot \frac{\phi_1}{\phi_1 + \phi_2}$. It can be checked that $\mbox{Cov}(\omega, \theta_i) = \sigma_{\omega}^2/\zeta_i > 0$, so our Assumption \ref{assmp:K=2} is satisfied and Theorem \ref{thm:K=2} holds. 

The optimal path of attention allocations determines the discounted average attention paid to source $i$, namely $\int_{0}^{\infty} r e^{-rt} \beta_i(t) ~dt$, where $r$ is a common discount rate for the sources.\footnote{Note that in this formulation, we implicitly assume that the reader never stops information acquisition, which simplifies our subsequent analysis. However, $\int_{0}^{\infty} r e^{-rt} \beta_i(t) ~dt$ can be interpreted as the limiting discounted average attention that source $i$ receives when the reader chooses an endogenous stopping time under vanishingly small information acquisition costs. Never stopping can also be justified in an extension of our main model where the agent faces multiple decisions across time (see Section \ref{sec:discussion}).} We use the discounted average attention as a reduced form for advertising revenue, where each news source receives a profit proportional to its viewership. Each source $i$'s total payoff is then the sum of discounted attention and the reward for bias intensity
\[
U_i = \int_{0}^{\infty} r e^{-rt} \beta_i(t) ~dt - \lambda (\phi_i - \kappa)^2.
\]
The following proposition derives the equilibrium in this game between the two sources. For technical reasons, we require an assumption that the incentives for bias are not too weak.

\begin{proposition}\label{prop:eqm} 
Suppose $\lambda \kappa^2 \geq 1.6$.\footnote{This assumption is imposed to guarantee the existence of a pure strategy equilibrium. Our analysis shows that $\lambda \kappa^2 \geq 1.6$ is sufficient for existence, whereas a weaker condition $\lambda \kappa^2 \geq \frac{9}{16}$ is necessary.} The unique pure strategy equilibrium is $(\phi_1^*,\zeta_1^*; \phi_2^*,\zeta_2^*)$ where 
\[
\phi_1^*=\phi_2^*= \frac12 \left(\kappa + \sqrt{\kappa^2 - \frac{1}{2\lambda}}\right)
\]
and 
\[
\zeta_1^*=\zeta_2^* = \frac{\sigma_b}{2\sqrt{r}} \cdot \left(\kappa + \sqrt{\kappa^2 - \frac{1}{2\lambda}}\right).
\]
Given these equilibrium choices, the reader optimally devotes equal attention to the two sources at every moment.
\end{proposition}



The subsequent corollary regarding the informativeness of news in equilibrium follows immediately.\footnote{It can be computed that when the sources choose $\zeta_1 = \zeta_2 = \zeta^*$ and $\phi_1 = \phi_2$, the news reader's posterior variance of the payoff-relevant state $\omega$ at time $t$ is $\left(\frac{1}{\sigma_{\omega}^2} + \frac{t}{(\zeta^*)^2}\right)^{-1}$. This confirms why $\zeta^*$ is a sufficient statistic for equilibrium informativeness.} 

\begin{corollary}[Informativeness of News] \label{corr:eqm noise}
The equilibrium noise level, $\zeta^* = \frac{\sigma_b}{2\sqrt{r}} \left(\kappa + \sqrt{\kappa^2 - \frac{1}{2\lambda}}\right)$, 
\begin{itemize}[noitemsep]
    \item [(a)] is increasing in the incentive for bias, $\lambda$, and the bias intensity bliss point, $\kappa$;
    \item [(b)] is increasing in the prior uncertainty about partisan implications, $\sigma_b$;
    \item [(c)] is decreasing in the discount rate, $r$.
\end{itemize}
\end{corollary}

Part (a) says that incentives for greater bias not only increase polarization, which is expected, but also lead to a reduction in the quality of news. To understand this result, consider the incentives for source $i$'s choice of precision. Applying our characterization in Theorem \ref{thm:K=2}, there are up to two stages of information acquisition: In Stage 1, if there is a strictly more informative source, then that source receives all viewership; in Stage 2, both sources receive a constant proportion of viewership. We show that, for any equal bias intensity choices $\phi_1 = \phi_2$, source $i$'s long-run share is $\frac{\zeta_i}{\zeta_1 + \zeta_2}$, while source $i$ is chosen in Stage 1 if and only if its noise term is smaller ($\zeta_i<\zeta_{3-i}$). Thus, sources face a trade-off between optimizing for greater long-run viewership\textemdash where a larger noise choice $\zeta_i$ increases the long-run share\textemdash versus competing to be chosen in the short-run\textemdash which encourages smaller $\zeta_i$. Intuitively, more precise information improves the competitive value of a source at the beginning of time, but reduces the value of continual engagement with that source. In equilibrium, sources choose the same $\zeta_i$, thus washing out the first stage of information acquisition. 

The size of this common noise level, however, depends on the incentives for bias. When sources provide biased news, the reader must attend to both sources to learn the truth. Polarized news sources thus live in symbiosis, where the extremity of bias on one side increases the value of information from the other. In the language of our paper, two sufficiently polarized news sources on opposite sides provide complementary information (while in contrast, two unbiased sources about $\omega$ provide fully substitutable information). The strength of complementarity increases monotonically with the degree of polarization.            

Since the reader has stronger preferences for mixing over the two sources when they are complements, this means that the length of Stage 1 (when it exists) is decreasing in the degree of polarization. Thus, the more polarized the news sources are, the more emphasis these sources place on the long run, which in turn leads to lower quality news provision as we have discussed. This gives the conclusion in Part (a) of Corollary \ref{corr:eqm noise}. In addition, larger prior uncertainty about $b$ implies higher value of de-biasing and thus also a shorter Stage 1, leading to Part (b).

Part (c) of Corollary \ref{corr:eqm noise} holds by very similar reasoning. Less patient news sources compete over short-run profits (i.e., being chosen in Stage 1), and thus prefer precise signals, while patient sources compete for long-run profits (i.e., long-run proportion), and thus prefer imprecise signals. So the less patient the sources are (larger $r$), the more precise their signals will be in equilibrium (smaller $\zeta^*)$.

\subsection{Application 3: Attention Manipulation}\label{sec:manipulation}

Our analyses have so far assumed that the agent has complete control over how to allocate his attention. In practice, businesses expend substantial effort to divert attention towards their products. Such ``attention grabbing" often takes the form of a one-time intervention (e.g., an ad) rather than a continual shift in exposure, so the value of the attention diversion depends on how it shapes subsequent allocation of attention.  Two questions thus naturally arise: 1) Does a one-time manipulation of attention towards a given source lead to a persistently higher amount of attention devoted to that source, or will the decision-maker quickly ``compensate" for the manipulation? 2) What are the externalities on other sources---in particular, is it the case that manipulating attention towards one source decreases the amount of attention devoted to others?

\citet{GossnerSteinerStewart} (henceforth GSS) recently studied this question in a model in which an agent sequentially learns about the quality of a number of goods by allocating attention to one good each period. One of their main results (Theorem 1) resolves the two questions posed above in the following way: 
\begin{itemize}[noitemsep]
\item[(1)] the cumulative amount of attention paid to that good remains persistently higher following the attention manipulation, and

\item[(2)] the cumulative amount of attention paid to any other good remains persistently lower following the attention manipulation.
\end{itemize}
A key assumption in GSS is that the attention strategy used by the agent satisfies a version of Independence of Irrelevant Alternatives (IIA): Conditional on \emph{not} focusing on the good to which attention is diverted, the agent's belief about that good does not affect the relative probabilities of focusing on the remaining goods. Proposition 5 in \citet{GossnerSteinerStewart} shows that when the agent adopts a class of ``satisficing'' stopping rules, the optimal attention strategy satisfies IIA for independent goods.\footnote{This additional assumption on the stopping rule is not required in our setting, since we focus on the uniformly optimal attention strategy which is independent of stopping behavior.}

We can use our framework and main characterization to study a related but different problem, where the agent learns about multiple attributes of an unknown (one-dimensional) payoff-relevant state, and---importantly---those attribute values can be correlated. Additionally, we differ from GSS by focusing on the optimal attention allocation strategy and how it is affected by attention manipulation. Outside of the special case of independent attributes, the optimal strategy in our setting fails IIA when there are more than two attributes. Nevertheless, we show GSS's finding in (1) holds for flexible patterns of correlation, and the finding in (2) holds under an additional condition, which we make precise.

Formally, suppose attention is manipulated such that the agent only attends to source $1$ from time zero to time $T$, where $T > 0$ is fixed in this section. After time $T$, the agent adopts the optimal attention strategy given his posterior belief at $T$. The dynamic effect of the one-time attention manipulation is then understood by comparing the cumulative attention vectors under the optimal strategy and under the manipulated strategy. We assume throughout that our previous conditions on the prior covariance matrix apply (i.e., Assumption \ref{assmp:K=2} if $K=2$ and Assumption \ref{assmp:Substitutes}, \ref{assmp:Complements} or \ref{assmp:diagonal dominance} if $K>2$).
\begin{proposition} \label{prop:manipulation increases attention}
Let $T^* \geq T$ be the earliest time at which cumulative attention towards source 1 exceeds $T$ under the baseline (unmanipulated optimal) strategy. Then, cumulative attention towards source 1 is strictly larger under the manipulated attention strategy at every moment of time $t \in (T, T^*)$, and equal to the baseline at all later times $t\geq T^*$. 
\end{proposition}

Thus, the finding in (1) holds under arbitrary correlation (so long as our characterizations apply): Attention manipulation towards source 1 has a persistent positive effect on the total amount of attention source 1 receives up to every future time. On the other hand, we show that this increase in cumulative attention vanishes in the long run, with the cumulative attention paid to source 1 under the baseline strategy ``catching up" to the manipulated strategy by time $T^*$.

The proof of this proposition is simple given our previous analysis. The cumulative attention vector under the optimal strategy is the $t$-optimal vector defined as
\begin{equation}\label{eq:t-optimal}
n(t) = (n_1(t), \dots, n_K(t)) = \argmin_{q_1, \dots, q_K \geq 0: ~\sum_{i} q_i = t} V(q),
\end{equation}
where $V$ is the posterior variance function given by \eqref{eq:V(q)}. On the other hand, the manipulated strategy induces the following \emph{constrained} $t$-optimal vector $\hat{n}(t)$ at any time $t > T$:\footnote{As discussed, the ``absorbing property" of our sufficient conditions implies that our characterizations apply to the posterior belief at time $T$ after the agent has paid $T$ units of attention to source $1$. Thus, the cumulative attention vector at time $t \geq T$ must minimize the posterior variance $V$ among feasible attention vectors following the manipulation. This leads to the constrained $t$-optimal vector $\hat{n}(t)$.}
\begin{equation}\label{eq:constrained t-optimal}
\hat{n}(t) = (\hat{n}_1(t), \dots, \hat{n}_K(t)) = \argmin_{q_1, \dots, q_K \geq 0: ~\sum_{i} q_i = t \text{ and } q_1 \geq T} V(q).
\end{equation}

If $n_1(t) \geq T$, then the unconstrained $t$-optimal vector $n(t)$ satisfies the constraint in \eqref{eq:constrained t-optimal}, so it coincides with the constrained $t$-optimal vector $\hat{n}(t)$. Moreover, $n_1(t) \to \infty$ as $t \to \infty$ because our characterization says that source $1$ receives positive and constant attention at every instant in the final stage. Thus while initially source 1 must receive higher cumulative attention under the manipulated strategy,  eventually the cumulative attention devoted to source 1 must be the same under the baseline and manipulated strategies. To show that $T^*$ is this switch point, note that $t \geq T^*$ implies $n_1(t) \geq T$ (by definition of $T^*$ and monotonicity of $n_1(t)$), in which case $n_1(t) = \hat{n}_1(t)$. And if $T<t<T^*$, we have  $n_1(t) < T \leq \hat{n}_1(t)$, so the manipulated amount of attention devoted to source 1 strictly exceeds that of the baseline strategy. This yields the result.

\bigskip

When there are only two attributes, Proposition \ref{prop:manipulation increases attention} also delivers GSS's second finding, namely that diversion of attention towards learning about one attribute weakly reduces cumulative attention towards learning about the remaining attribute at every moment of time. With more than two attributes, however, correlation between the attributes can overturn this result. 

\begin{example}
Suppose there are 3 attributes, the payoff-relevant state is $\omega = \theta_1 + \theta_2 + \theta_3$, and the prior covariance matrix is 
\[
\Sigma = \left(\begin{array}{ccc} 3 & -2 & 0 \\ -2 & 3 & 0 \\ 0 & 0 & 2 \end{array} \right) 
\]
Since $\Sigma^{-1}$ is diagonally-dominant, Theorem \ref{thm:general} applies.

Without attention manipulation, the optimal strategy devotes the first $0.5$ units of attention towards  $\theta_3$. At $t = 0.5$, the three sources have exactly equal marginal values (and equal payoff weights), so equal attention is optimal afterwards. Thus $n(t) = (0, 0, t)$ for $ t < 0.5$ and $n(t) = (\frac{t-0.5}{3}, \frac{t-0.5}{3}, \frac{t+1}{3})$ for $t \geq 0.5$.

Now suppose instead that the agent is forced to attend to source 1 for $0.1$ unit of time. Then after the first $0.1$ units of attention devoted towards $\theta_1$, the agent optimally still begins by learning about $\theta_3$. This lasts until $t^* = \frac{7}{15} < 0.5$, at which time source $2$ has the same marginal value as source $3$. The second stage then involves learning about $\theta_2$ and $\theta_3$ using the constant attention ratio $3:7$. The third stage begins at $t^{**} = 0.8$, after which equal attention across the three sources is optimal.
It can be checked that the manipulation of attention towards source 1 weakly \emph{increases} the cumulative attention towards source 2 at all times, and strictly so during the period $t \in (\frac{7}{15}, 0.8)$. 
\end{example}

Intuitively, in this example sources 1 and 2 provide complementary information (since $\Sigma_{12} < 0$). Manipulating attention towards source 1 thus increases the marginal value of source 2, and the agent begins observing source 2 earlier than he would have otherwise.  In contrast, we might expect that when all sources are substitutes with one another, attention manipulation to source 1 must decrease the amount of attention devoted to every other source. The challenge is understanding what the appropriate notion of ``substitutes" is. This turns out to be the property given earlier in Assumption \ref{assmp:Substitutes}, which guarantees that each pair of attributes has a positive partial correlation coefficient. 

\begin{proposition}\label{prop:manipulation decreases attention} 
Suppose all pairs of sources are substitutes (i.e., Assumption \ref{assmp:Substitutes} is satisfied). Then cumulative attention towards every source $i > 1$ is weakly smaller under the manipulated strategy than under the baseline strategy, at every moment of time. 
\end{proposition}

This result complements GSS by demonstrating a class of correlated attributes for which manipulation of attention towards one reduces cumulative attention towards all others. Together with our previous Proposition \ref{prop:manipulation increases attention}, it shows that GSS's Attention Theorem extends beyond their IIA assumption. Since our environment and GSS's are non-nested, these results collectively point to the possibility of a more general set of sufficient conditions, which we leave to future work.

\section{Discussion} \label{sec:discussion}

Information acquisition is a classic problem within economics, but there are relatively few dynamic models that are simultaneously rich and tractable. In this paper we present a class of dynamic information acquisition problems whose solution can be explicitly characterized in closed-form. We show that a complete analysis is feasible if we assume: (1) Gaussian uncertainty, (2) a one-dimensional payoff-relevant state, and (3) correlation across the unknowns that satisfies certain assumptions (for example if correlation is not too strong). Given these restrictions, a great deal of generality can be accommodated in other aspects of the problem, such as the payoff function. The tractability of the solution and the flexibility of the environment open the door to interesting applications, a number of which we have illustrated here. 

We conclude by briefly mentioning a few other potential extensions and variations.

\smallskip

\textbf{Discrete Time.} Although our main model is in continuous time, our results have direct analogues in a related discrete-time model. Specifically, for the model previously described in Remark \ref{remark:discreteTime}, we have the following result: \emph{Suppose any of Assumption \ref{assmp:Substitutes}, \ref{assmp:Complements}, or \ref{assmp:diagonal dominance} holds. Then at each period $t \in \mathbb{Z}_{+}$, the optimal allocation of precision is $(\pi_1(t), \dots, \pi_K(t))$ where $\pi_i(t) = \int_{t}^{t+1} \beta_i(s) ~ds$ for each $i$, with $\beta_i(s)$ being the optimal attention allocation for the continuous-time model as described in Theorem \ref{thm:general}.}\footnote{In a companion piece, \citet{LiangMuSyrgkanis}, we discretize not only time but also information acquisitions: At each period $t$, the agent has to choose one of $K$ standard normal signals, without the ability to allocate fractional precisions. The necessity of integer approximation complicates the characterization of the full sequence of signal choices. In that paper we instead provide conditions under which myopic acquisition is (eventually) optimal.}

\smallskip

\textbf{Exogenous Stopping.} Although we have assumed that the agent endogenously chooses when to stop acquiring information, our results hold without modification if instead the end date arrives according to an arbitrary exogenous distribution. Additionally, in that alternative model, the optimal path of attention allocations is uniquely characterized by our results so long as the exogenous distribution of end date has full support. 

\smallskip

\textbf{Intertemporal Decision Problems.} Our main model assumes that the agent takes only one action, which simplifies the exposition. But since our analysis based on the notion of uniform optimality is independent of details of the payoff function, it can be easily generalized to a setting where the agent takes $N$ actions $a_1, \dots, a_N$ at times $\tau_1 \leq \dots \leq \tau_N$. Our characterization of the optimal attention strategy extends for any (intertemporal) payoff function $u(\tau_1, \dots, \tau_N, a_1, \dots, a_N, \omega)$ that is decreasing in the decision times $\tau_1, \dots, \tau_N$.



\bigskip

\appendix


\section{Preliminaries}

\subsection{Posterior Variance Function} \label{appx:V}
Given $q_i$ units of attention devoted to learning about each attribute $i$, the posterior variance of $\omega$ can be written in two ways:
\begin{lemma}\label{lemma:V} 
It holds that
\[
V(q_1, \dots, q_K) = \alpha' \left[(\Sigma^{-1} + \diag(q))^{-1}\right] \alpha = \alpha' \left[\Sigma - \Sigma (\Sigma + \diag(1/q))^{-1} \Sigma\right] \alpha 
\]
where $\diag(1/q)$ is the diagonal matrix with entries $1/q_1, \dots, 1/q_K$. 

This function $V$ extends to a rational function (quotient of polynomials) over all of $\mathbb{R}^K$, i.e., even if some $q_i$ are negative.
\end{lemma}

\begin{proof}
The equality $(\Sigma^{-1} + \diag(q))^{-1} = \Sigma - \Sigma (\Sigma + \diag(1/q))^{-1} \Sigma$ is well-known. To see that $V$ is a rational function, simply note that $(\Sigma^{-1} + \diag(q))^{-1}$ can be written as the \emph{adjugate matrix} of $\Sigma^{-1} + \diag(q)$ divided by its determinant. Thus each entry of the posterior covariance matrix is a rational function in $q$. 
\end{proof}

Below we calculate the first and second derivatives of the posterior variance function $V$:
\begin{lemma}\label{lemma:V partials}
Given a cumulative attention vector $q \geq 0$, define 
\[
\gamma:=\gamma(q)= (\Sigma^{-1}+\diag(q))^{-1} \alpha
\]
which is a vector in $\mathbb{R}^{K
}$. Then the first and second derivatives of $V$ are given by
\[
\partial_i V = -\gamma_i^2, \quad \quad \quad \quad  \partial_{ij} V = 2\gamma_i \gamma_j \cdot \left[(\Sigma^{-1}+\diag(q))^{-1}\right]_{ij}. 
\]
\end{lemma}

\begin{proof}
From Lemma \ref{lemma:V} and the formula for matrix derivatives, we have
\[
\partial_i V = - \alpha' (\Sigma^{-1}+\diag(q))^{-1} \Delta_{ii} (\Sigma^{-1}+\diag(q))^{-1} \alpha = -\left[e_i'(\Sigma^{-1}+\diag(q))^{-1} \alpha\right]^2 = -\gamma_i^2
\]
where $e_i$ is the $i$-th coordinate vector in $\mathbb{R}^K$, and $\Delta_{ii} = e_i \cdot e_i'$ is the matrix with ``1" in the $(i,i)$-th entry and ``0" elsewhere. For the second derivative, we compute that
\[
\partial_{ij}V = -2\gamma_i \cdot \frac{\partial \gamma_i}{\partial q_j} = 2\gamma_i \cdot e_i' (\Sigma^{-1}+\diag(q))^{-1} \Delta_{jj} (\Sigma^{-1}+\diag(q))^{-1} \alpha = 2\gamma_i \cdot \left[(\Sigma^{-1}+\diag(q))^{-1}\right]_{ij} \cdot \gamma_j
\]
as we desire to show. The last equality follows by writing $\Delta_{jj} = e_j \cdot e_j'$, and using $e_i'(\Sigma^{-1}+\diag(q))^{-1}e_j = \left[(\Sigma^{-1}+\diag(q))^{-1}\right]_{ij}$ as well as $e_j' (\Sigma^{-1}+\diag(q))^{-1} \alpha = e_j' \gamma = \gamma_j$. 
\end{proof}

\begin{corollary}
$V$ is decreasing and convex in $q_1, \dots, q_K$ whenever $q_i \geq 0$. 
\end{corollary}

\begin{proof}
By Lemma \ref{lemma:V partials}, the partial derivatives of $V$ are non-positive, so $V$ is decreasing. Additionally, its Hessian matrix is
\[
2 \diag(\gamma) \cdot (\Sigma^{-1}+\diag(q))^{-1} \cdot \diag(\gamma),
\]
which is positive semi-definite whenever $q \geq 0$. So $V$ is convex.
\end{proof}

We use these properties to show that for each $t$, the $t$-optimal vector $n(t)$ is unique: 

\begin{lemma}\label{lemma:t-opt unique} 
For each $t \geq 0$, there is a unique $t$-optimal vector $n(t)$. 
\end{lemma} 

\begin{proof}
Suppose for contradiction that two vectors $(r_1, \dots, r_K)$ and $(s_1, \dots, s_K)$ both minimize the posterior variance at time $t$. Relabeling the sources if necessary, we can assume $r_i - s_i$ is positive for $1 \leq i \leq k$, negative for $k+1 \leq i \leq l$ and zero for $l+1 \leq i \leq K$. Since $\sum_{i}r_i = \sum_{i} s_i = t$, the cutoff indices $k, l$ satisfy $1 \leq k < l \leq K$. 

For $\lambda \in [0,1]$, consider the vector $q^{\lambda} = \lambda \cdot r + (1-\lambda) \cdot s$ which lies on the line segment between $r$ and $s$. Then by assumption we have $V(r) = V(s) \leq V(q^\lambda)$. Since $V$ is convex, equality must hold. This means $V(q^\lambda)$ is a constant for $\lambda \in [0,1]$. But $V(q^\lambda)$ is a rational function in $\lambda$, so its value remains the same constant even for $\lambda > 1$ or $\lambda < 0$. In particular, consider the limit as $\lambda \to +\infty$. Then the $i$-th coordinate of $q^{\lambda}$ approaches $+\infty$ for $1 \leq i \leq k$, approaches $-\infty$ for $k+1 \leq i \leq l$ and equals $r_i$ for $i > l$. 

For each $q^{\lambda}$, let us also consider the vector $\vert q^{\lambda} \vert$ which takes the absolute value of each coordinate in $q^{\lambda}$. Note that as $\lambda \to +\infty$, $\diag(1/\vert q^{\lambda} \vert)$ has the same limit as $\diag(1/q^{\lambda})$. Thus by the second expression for $V$ (see Lemma \ref{lemma:V}), $\lim_{\lambda \to \infty} V(\vert q^{\lambda} \vert) = \lim_{\lambda \to \infty} V(q^{\lambda}) = V(r)$. For large $\lambda$, the first $l$ coordinates of $\vert q^{\lambda} \vert$ are strictly larger than the corresponding coordinates of $r$, and the remaining coordinates coincide. So the fact that $V$ is decreasing and  $V(\vert q^{\lambda} \vert) = V(r)$ implies $\partial_i V(r) = 0$ for $1 \leq i \leq l$. 

Consider the vector $\gamma = (\Sigma^{-1}+\diag(r))^{-1} \alpha$. By Lemma \ref{lemma:V partials}, $\partial_i V(r) = -\gamma_i^2$ for $1 \leq i \leq K$. Thus $\gamma_1 = \dots = \gamma_{l} = 0$. Since $\alpha$ and thus $\gamma$ is not the zero vector, there exists $j > l$ s.t. $\gamma_j \neq 0$. It follows that $\partial_1 V(r) = 0 > \partial_j V(r)$. But then the posterior variance $V$ would be reduced if we slightly decreased the first coordinate of $r$ (which is strictly positive since $r_1 > s_1$) and increased the $j$-th coordinate by the same amount. This contradicts the assumption that $r$ is a $t$-optimal vector. Hence the lemma holds. 
\end{proof}

\subsection{Optimality and Uniform Optimality}\label{appx:uniform optimality}
The following result ensures that a strategy that minimizes the posterior variance uniformly at all times is an optimal strategy in any decision problem. 

\begin{lemma}\label{lemma:Greenshtein}
Suppose the payoff function $u(\tau, a, \omega)$ satisfies Assumption \ref{assmp:waiting is costly}, then a uniformly optimal attention strategy is dynamically optimal. 
\end{lemma}
\begin{proof}

Without loss of generality we may assume the prior mean of $\omega$ is zero; otherwise shift $\omega$ by a constant and modify the utility function accordingly. Let $S^*$ be the uniformly optimal attention strategy, and $\{\mathcal{F}^*_{t}\}$ be the induced filtration. Given $S^*$, the optimal stopping rule $\tau$ is a solution to 
\[
\sup_{\tau}~ \mathbb{E}\left[\max_{a}~ \mathbb{E}[u(\tau, a,\omega) \mid \mathcal{F}^*_{\tau}]\right].
\]
Note that the stochastic process of posterior means $M^*_t = \mathbb{E}[\omega \mid \mathcal{F}^*_t]$ is a continuous martingale adapted to the filtration $\{\mathcal{F}^*_{t}\}$, with $M^*_0 = 0$. Moreover, since information is Gaussian, the quadratic variation $\langle M^* \rangle_t$ is simply $v_0 - v^*_t$, where $v^*_t$ is the posterior variance of $\omega$ at time $t$ under the strategy $S^*$, and $v_0$ is the prior variance. By definition of uniform optimality, for each $t$ the random variable $v^*_t$ is deterministic and moreover smallest among possible posterior variances at time $t$. 

Thus, by the Dambis--Dubins--Schwartz Theorem (see Theorem 1.7 in Chapter V of \citet{RevuzYor}), there exists a Brownian motion $(B^*_\nu)_{\nu \in [0, v_0)}$ such that 
\[
B^*_{v_0 - v^*_t} = \mathbb{E}[\omega \mid \mathcal{F}^*_t]. 
\]
This allows us to change variable from the time $t$ to the cumulative precision $v_0 - v^*_t$. 

To formulate the resulting optimization problem, for each $\nu \in [0, v_0)$ we denote by $T^*(\nu)$ the time $t$ such that $v^*_t = v_0 - \nu$; $T^*$ is a deterministic and increasing function of $\nu$. Then, under the attention strategy $S^*$, the agent's optimal payoff can be rewritten as
\begin{equation}\label{eq:payoff under uniform optimality} 
\sup_{\tau}~ \mathbb{E}\left[\max_{a}~ \mathbb{E}[u(\tau, a,\omega) \mid \mathcal{F}^*_{\tau}]\right] = \sup_{\nu}~ \mathbb{E}\left[\max_{a}~ \mathbb{E}[u(T^*(\nu), a,\omega) \mid B^*_{\nu}]\right].
\end{equation}
In other words, instead of optimizing over stopping times $\tau$ adapted to $\{\mathcal{F}^*_{t}\}$, we can think of the agent choosing an optimal $\nu = v_0 - v^*_t$ adapted to the Brownian motion $B^*$.

We will show this payoff is greater than the optimal payoff under any other attention strategy $S$. To do this, let $\{\mathcal{F}_{t}\}$ be the induced filtration under $S$. Similar to the above, we consider the stochastic process $M_t = \mathbb{E}[\omega \mid \mathcal{F}_{t}]$, adapted to $\{\mathcal{F}_{t}\}$. Applying the Dambis-Dubins-Schwartz Theorem again, there exists a Brownian motion $(B_\nu)_{\nu \in [0, v_0)}$ such that 
\[
B_{v_0 - v_t} = \mathbb{E}[\omega \mid \mathcal{F}_t].
\]
Here $v_t$ is the posterior variance under strategy $S$, which is in general random but always satisfies $v_t \geq v^*_t$. Note also that $B$ may not be the same process as $B^*$. 

Observe that for any $t \geq 0$ we have $t = T^*(v_0 - v^*_t) \geq T^*(v_0 - v_t)$. Thus the agent's payoff under strategy $S$ is bounded above by 
\[
\sup_{\tau}~ \mathbb{E}\left[\max_{a}~ \mathbb{E}[u(\tau, a,\omega) \mid \mathcal{F}_{\tau}]\right] \leq \sup_{\tau}~ \mathbb{E}\left[\max_{a}~ \mathbb{E}[u(T^*(v_0 - v_\tau), a,\omega) \mid \mathcal{F}_{\tau}]\right],
\]
where we used Assumption \ref{assmp:waiting is costly}. Now we can change variable again from $\tau$ to $\nu = v_0 - v_\tau$, and rewrite the payoff as 
\begin{equation}\label{eq:payoff under arbitrary strategy}
\sup_{\nu}~ \mathbb{E}\left[\max_{a}~ \mathbb{E}[u(T^*(\nu), a,\omega) \mid B_{\nu}]\right]
\end{equation}
This is the same as the RHS of (\ref{eq:payoff under uniform optimality}), since $B$ and $B^*$ are both Brownian motions. Hence the payoff under $S$ does not exceed the payoff under $S^*$, completing the proof. 
\end{proof}

We also have a simple converse result:
\begin{lemma}\label{lemma:uniform optimal converse}
Fixing $\Sigma$ and $\alpha$. Suppose an information acquisition strategy is optimal for all payoff functions $u(\tau, a, \omega)$ that satisfy Assumption \ref{assmp:waiting is costly}, then it is uniformly optimal. 
\end{lemma}

\begin{proof}
Take an arbitrary time $t$ and consider the payoff function with $u(\tau, a, \omega) = - (a-\omega)^2 - c(\tau)$, where $c(\tau) = 0$ for $\tau \leq t$ and $c(\tau)$ very large for $\tau > t$. Then the agent's optimal stopping rule is to stop exactly at time $t$. Since his information acquisition strategy is optimal for this payoff function, the induced cumulative attention vector must achieve $t$-optimality. Varying $t$ yields the result. 
\end{proof}

\subsection{Sufficient Condition for Assumption \ref{assmp:diagonal dominance}} \label{appx:2K-3 implies diagonal dominance}

\begin{lemma} 
Suppose the prior covariance matrix $\Sigma$ satisfies $\Sigma_{ii} \geq (2K-3) \cdot \vert \Sigma_{ij} \vert$ for all $i \neq j$. Then its inverse matrix satisfies $[\Sigma^{-1}]_{ii} \geq (K-1) \cdot \vert [\Sigma^{-1}]_{ij} \vert$ for all $i \neq j$, and is thus diagonally-dominant. 
\end{lemma}

\begin{proof}
By symmetry, we can focus on $i = 1$. Let $s_j = [\Sigma^{-1}]_{1j}$ for $1 \leq j \leq K$, and without loss assume $s_2$ has the greatest absolute value among $s_2, \dots, s_K$. It suffices to show 
\[
s_1 \geq (K-1) \vert s_2 \vert.
\]
From $\Sigma^{-1} \cdot \Sigma = I$ we have $\sum_{j = 1}^{K} [\Sigma^{-1}]_{1j} \cdot \Sigma_{j2} = 0$. Thus $\sum_{j = 1}^{K} s_j \cdot \Sigma_{2j} = 0$ because $\Sigma_{j2} = \Sigma_{2j}$. Rearranging yields
\[
\vert s_1 \cdot \Sigma_{21} \vert = \vert s_2 \cdot \Sigma_{22} + \sum_{j > 2} s_j \cdot \Sigma_{2j} \vert \geq \vert s_2 \cdot \Sigma_{22} \vert - \sum_{j > 2} \vert s_j \cdot \Sigma_{2j} \vert \geq \vert s_2 \cdot \Sigma_{22} \vert - \sum_{i > 2} \frac{\vert s_2 \cdot \Sigma_{22}\vert}{2K-3},
\]
where the last inequality uses $\vert s_j \vert \leq \vert s_2 \vert$ and $\vert \Sigma_{2j} \vert \leq \frac{1}{2K-3}\vert \Sigma_{22} \vert$ for $j > 2$. The above inequality simplifies to
\[
\vert s_1 \cdot \Sigma_{21} \vert \geq \frac{K-1}{2K-3} \cdot \vert s_2 \cdot \Sigma_{22}\vert.
\]
And since $\Sigma_{21} \leq \frac{1}{2K-3}\vert \Sigma_{22} \vert$, we conclude that $\vert s_1 \vert \geq (K-1) \vert s_2 \vert$ as desired. Note that $s_1 = [\Sigma^{-1}]_{11}$ is necessarily positive, thus $s_1 \geq (K-1) \vert s_2 \vert$. 
\end{proof}

\section{Proof of Theorem \ref{thm:K=2}}\label{appx:theorem K=2}

Define $cov_1, cov_2$ as in the statement of Theorem \ref{thm:K=2}, and define $x_i = \alpha_i \det(\Sigma)$ to ease notation.

Given a cumulative attention vector $q$, let $Q$ be a shorthand for the diagonal matrix $\diag(q)$. Then by direct computation, we have
\begin{align*}
\gamma ~:&= \left(\Sigma^{-1} + Q\right)^{-1} \cdot \alpha = \left(\Sigma^{-1} \cdot (I + \Sigma Q)\right)^{-1} \cdot \alpha \\
&= (I+\Sigma Q)^{-1} \cdot \Sigma \cdot \alpha = (I+\Sigma Q)^{-1} \cdot \left( \begin{array}{c} cov_1 \\ cov_2 \end{array} \right) \\
&= \frac{1}{\det(I + \Sigma Q)} \left( \begin{array}{cc} 1+q_2\Sigma_{22} & -q_2\Sigma_{12} \\ -q_1\Sigma_{21} & 1 + q_1\Sigma_{11} \end{array} \right) \cdot \left( \begin{array}{c} cov_1 \\ cov_2 \end{array} \right) = \frac{1}{\det(I + \Sigma Q)} \left( \begin{array}{c} x_1q_2 + cov_1 \\ x_2q_1 + cov_2 \end{array} \right).
\end{align*}
By Lemma \ref{lemma:V partials}, this implies the marginal values of the two sources are given by:
\begin{equation}\label{eq:K=2partials}
\partial_1 V(q_1, q_2) = \frac{-(x_1q_2 + cov_1)^2}{\det^2(I + \Sigma Q)}; \quad \quad \quad
\partial_2 V(q_1, q_2) = \frac{-(x_2q_1 + cov_2)^2}{\det^2(I + \Sigma Q)}.
\end{equation}

Note that Assumption \ref{assmp:K=2} translates into $cov_1 + cov_2 \geq 0$. Under this assumption, we will characterize the $t$-optimal vector $(n_1(t), n_2(t))$ and show it is increasing over time. Without loss assume $cov_1 \geq cov_2$, then $cov_1$ is non-negative. Let $t_1^* = \frac{cov_1-cov_2}{x_2}$. Then when $q_1 + q_2 \leq t_1^*$ we always have 
\[
x_1q_2 + cov_1 \geq cov_1 \geq x_2q_1 + cov_2,
\]
since $x_1q_2 \geq 0$ and $x_2q_1 \leq x_2(q_1+q_2) \leq x_2t_1^* = cov_1 - cov_2$. We also have
\[
x_1q_2 + cov_1 \geq -(x_2q_1 + cov_2),
\]
since $x_1q_2, x_2q_1 \geq 0$ and by assumption $cov_1 + cov_2 \geq 0$. Thus, (\ref{eq:K=2partials}) implies that $\partial_1 V(q_1, q_2) \leq \partial_2 V(q_1, q_2)$ at such attention vectors $q$. So for any budget of attention $t \leq t_1^*$, putting all attention to source $1$ minimizes the posterior variance $V$. That is, $n(t) = (t,0)$ for $t \leq t_1^*$. 

For $t > t_1^*$, observe that (\ref{eq:K=2partials}) implies $\partial_1V(0,t) < \partial_2V(0,t)$ as well as $\partial_1V(t,0) > \partial_2V(t,0)$. Thus the $t$-optimal vector $n(t)$ is interior (i.e., $n_1(t)$ and $n_2(t)$ are both strictly positive). The first-order condition $\partial_1V = \partial_2V$, together with (\ref{eq:K=2partials}) and the budget constraint $n_1(t) + n_2(t) = t$, yields the solution 
\[
n(t) = \left(\frac{x_1 t + cov_1-cov_2}{x_1+x_2}, \frac{x_2 t - cov_1 + cov_2}{x_1 + x_2}\right).
\]
Hence $n(t)$ is indeed increasing in $t$. The instantaneous attention allocations $\beta(t)$ are the time-derivatives of $n(t)$, and they are easily seen to be described by Theorem \ref{thm:K=2}. In particular, the long-run attention allocation to source $i$ is $\frac{x_i}{x_1 + x_2}$, which simplifies to $\frac{\alpha_i}{\alpha_1 + \alpha_2}$. This completes the proof.

\section{Proof of Theorem \ref{thm:general}}\label{appx:theorem general}

We will first prove the result under Assumption \ref{assmp:diagonal dominance}. The proof is similar under the alternative Assumption \ref{assmp:Substitutes} or \ref{assmp:Complements}, and is presented at the end. 

Given Lemma \ref{lemma:Greenshtein}, it is sufficient to show that the $t$-optimal vector $n(t)$ is weakly increasing in $t$, that its time-derivative is locally constant, and that the time-derivative has an expanding support set (as described in the theorem). The proof is divided into several sections below. 

\subsection{Technical Property of $\gamma$}
We will use the following lemma regarding the marginal values of different sources:
\begin{lemma}\label{lemma:gamma positive} 
Suppose $\Sigma^{-1}$ is diagonally-dominant. Given an arbitrary attention vector $q$, define $\gamma$ as in Lemma \ref{lemma:V partials} and denote by $B$ the set of indices $i$ such that $\vert \gamma_i \vert$ is maximized. Then $\gamma_i$ is the \emph{same positive number} for every $i \in B$. 
\end{lemma}

\begin{proof}
We use $Q$ to denote $\diag(q)$. Since $(\Sigma^{-1}+Q)^{-1} \alpha = \gamma$, we equivalently have 
\[
\alpha = (\Sigma^{-1}+Q) \gamma.
\]
Suppose for contradiction that $\gamma_i \leq 0$ for some $i \in B$. Using the above vector equality for the $i$-th coordinate, we have 
\[
0 < \alpha_i = \sum_{j = 1}^{K} [\Sigma^{-1}+Q]_{ij} \cdot \gamma_j. 
\]
Rearranging, we then have
\[
[\Sigma^{-1}+Q]_{ii} \cdot (-\gamma_i) < \sum_{j \neq i } [\Sigma^{-1}+Q]_{ij} \cdot \gamma_j \leq \sum_{j \neq i } \vert [\Sigma^{-1}+Q]_{ij} \vert \cdot \vert \gamma_j \vert,
\]
which is impossible because $-\gamma_i \geq \vert \gamma_j \vert$ for each $j \neq i$ and $[\Sigma^{-1}+Q]_{ii} \geq \sum_{j \neq i }\vert [\Sigma^{-1}+Q]_{ij} \vert$. Thus $\gamma_i$ is positive for any $i \in B$. The result that these $\gamma_i$ are the same follows from the definition that their absolute values are maximal.
\end{proof}

\subsection{The Last Stage}
To prove Theorem \ref{thm:general}, we first consider those times $t$ when each of the $K$ sources has been sampled. The following lemma shows that after any such time, it is optimal to maintain a constant attention allocation proportional to $\alpha$.

\begin{lemma}\label{lemma:final stage} 
Suppose $\Sigma^{-1}$ is diagonally-dominant. If at some time $\underline{t}$, the $\underline{t}$-optimal vector satisfies $\partial_1V(n(\underline{t})) = \dots = \partial_KV(n(\underline{t}))$, then the $t$-optimal vector at each time $t \geq \underline{t}$ is given by 
\[
n(t) = n(\underline{t}) + \frac{t-\underline{t}}{\alpha_1+\cdots+\alpha_K} \cdot \alpha.\footnote{That is, $n_i(t) = n_i(\underline{t}) + \frac{t}{\alpha_1+\cdots+\alpha_K} \cdot \alpha_i$ for each $i$.}
\]
\end{lemma}

\begin{proof}
Consider increasing $n(\underline{t})$ by a vector proportional to $\alpha$. If we can show the equalities $\partial_1V = \dots = \partial_KV$ are preserved, then the resulting cumulative attention vector must be $t$-optimal. This is because for the convex function $V$, a vector $q$ minimizes $V(q)$ subject to $q_i \geq 0$ and $\sum_{i} q_i = t$ \emph{if and only if} it satisfies the KKT first-order conditions. 

We check the equalities $\partial_1V = \dots = \partial_KV$ by computing the marginal changes of each $\partial_iV$ when the attention vector $q = n(\underline{t})$ increases in the direction of $\alpha$. Denoting $\diag(q)$ by $Q$ to save notation, this marginal change equals
\[
\delta_i:= \sum_{j = 1}^{K} \partial_{ij} V \cdot \alpha_j = 2\sum_{j=1}^{K} \gamma_i \gamma_j \left[(\Sigma^{-1}+Q)^{-1}\right]_{ij} \cdot \alpha_j
\]
by Lemma \ref{lemma:V partials}. Applying Lemma \ref{lemma:gamma positive}, we have $\gamma_1 = \dots = \gamma_K$. Thus the above simplifies to
\[
\delta_i= 2\gamma_1^2 \sum_{j = 1}^{K}\left[(\Sigma^{-1}+Q)^{-1}\right]_{ij} \cdot \alpha_j = 2\gamma_1^2\gamma_i = 2\gamma_1^3.
\]
Hence $\partial_1V = \dots = \partial_KV$ continues to hold, completing the proof.  
\end{proof}

\subsection{Earlier Stages}
In general, we need to show that even when the agent is choosing from a subset of the sources, the $t$-optimal vector $n(t)$ is still increasing over time. This is guaranteed by the following lemma, which says that the agent optimally attends to those sources that maximize the marginal reduction of $V$, until a new source becomes another maximizer. For ease of exposition we work under the stronger assumption that $\Sigma^{-1}$ is \emph{strictly} diagonally-dominant, in the sense that $[\Sigma^{-1}]_{ii} > \sum_{j \neq i} \vert [\Sigma^{-1}]_{ij} \vert$ for all $ 1 \leq i \leq K.$ Later we discuss how the lemma should be modified without this strictness. 

\begin{lemma}\label{lemma:earlier stages}
Suppose $\Sigma^{-1}$ is strictly diagonally-dominant. Choose any time $\underline{t}$ and denote 
\[
B = \textstyle \argmin_{i} \partial_iV(n(\underline{t})) = \argmax_{i} \vert \gamma_i \vert. 
\] 
Then there exists $\beta \in \Delta^{K-1}$ supported on $B$ and $\overline{t} > \underline{t}$ such that $n(t) = n(\underline{t}) + (t-\underline{t})\cdot \beta$ at times $t \in [\underline{t}, \overline{t}]$. 

The vector $\beta$ depends only on $\Sigma, \alpha$ and $B$. The time $\overline{t}$ is the earliest time after $\underline{t}$ at which $\argmin_{i} \partial_i V(n(\overline{t}))$ is a strict superset of $B$. Moreover, when $\vert B \vert < K$ it holds that $\overline{t} < \infty$. Whereas when $\vert B \vert = K$, it holds that $\overline{t} = \infty$ and $\beta$ is proportional to $\alpha$.
\end{lemma}

\begin{proof}
The case when $\vert B \vert = K$ has been proved in Lemma \ref{lemma:final stage}, so we only consider $\vert B \vert < K$. Without loss we assume $B = \{1, \dots, k\}$ with $1 \leq k < K$. Let $q = n(\underline{t})$ and define $\gamma$ as before. By Lemma \ref{lemma:gamma positive}, $\gamma_i$ is the same positive number for $i \leq k$. Moreover, $t$-optimality implies that $q_j = 0$ whenever $j > k$. Otherwise the posterior variance could be reduced by decreasing $q_j$ and increasing $q_1$, as source $1$ has strictly higher marginal value than source $j$.

We now use a trick to deduce the current lemma from the previous Lemma \ref{lemma:final stage}. Specifically, given the prior covariance matrix $\Sigma$, we can choose another basis of the attributes $\theta_1, \dots, \theta_k, \theta_{k+1}^*, \dots, \theta_K^*$ with two properties: 
\begin{enumerate}
    \item each $\theta_j^*$ ($j > k$) is a linear combination of the original attributes $\theta_1, \theta_2, \dots, \theta_K$; 
    \item $\mbox{Cov}[\theta_i, \theta_j^*] = 0$ for all $i \leq k < j$, where the covariance is computed according to the prior belief $\Sigma$.
\end{enumerate}
Denote by $\tilde{\theta}$ the vector $(\theta_1, \dots, \theta_k)'$, and by $\theta^*$ the vector $(\theta_{k+1}^*, \dots, \theta_K^*)'$. The payoff-relevant state $\omega = \alpha' \cdot \theta$ can thus be rewritten as $\tilde{\alpha}' \cdot \tilde{\theta} + \alpha^{*'} \cdot \theta^*$ for some constant coefficient vectors $\tilde{\alpha} \in \mathbb{R}^k$ and $\alpha^* \in \mathbb{R}^{K-k}$. Using property 2 above, we can solve for $\tilde{\alpha}$ from $\Sigma$, $\alpha$ and $B$:
\begin{equation}\label{eq:tilde alpha}
\tilde{\alpha} = (\Sigma_{TL})^{-1} \cdot (\Sigma_{TL}, ~\Sigma_{TR}) \cdot \alpha
\end{equation}
where $\Sigma_{TL}$ is the $k \times k$ top-left sub-matrix of $\Sigma$ and $\Sigma_{TR}$ is the $k \times (K-k)$ top-right block. 

With this transformation, we have reduced the original problem with $K$ sources to a smaller problem with only the first $k$ sources. To see why this reduction is valid, recall that sampling sources $1 \sim k$ only provides information about $\tilde{\theta}$, which is orthogonal to $\theta^*$ according to the prior. So as long as the agent has only looked at the first $k$ sources, the transformed attributes continue to satisfy property 2 above (zero covariances) under any posterior belief. It follows that the posterior variance of $\omega$ is simply the variance of $\tilde{\alpha}'\cdot \tilde{\theta}$ plus the variance of $\alpha^{*'} \cdot \theta^*$. Since the latter uncertainty cannot be reduced, the agent's objective (at those times when only the first $k$ sources are attended to) is equivalent to minimizing the posterior variance of $\tilde{\omega} = \tilde{\alpha}' \cdot \tilde{\theta}$. 

Thus, in this smaller problem, the prior covariance matrix is $\Sigma_{TL}$ and the payoff weights are $\tilde{\alpha}$. Assuming that $\tilde{\alpha}$ has strictly positive coordinates, we can then apply Lemma \ref{lemma:final stage}: As long as the agent attends to the first $k$ sources proportional to $\tilde{\alpha}$, $\partial_1V = \dots = \partial_kV$ continues to hold.\footnote{Lemma \ref{lemma:final stage} implies $\partial_1\tilde{V} = \dots = \partial_k\tilde{V}$, where $\tilde{V}(q_1, \dots, q_k)$ is the posterior variance of $\tilde{\alpha}' \cdot \tilde{\theta}$ in the smaller problem. But as discussed, $\tilde{V}$ differs from $V$ by a constant, so its derivatives are the same as those of $V$.} Moreover, at $q = n(\underline{t})$, the definition of the set $B$ implies that these $k$ partial derivatives have greater magnitude (i.e., more negative) than the rest. By continuity, the same comparison holds until some time $\overline{t} > \underline{t}$. Thus, when $t \in [\underline{t}, \overline{t})$, the cumulative attention vector (under this strategy) still satisfies the first-order condition $B = \argmin_{1\leq i \leq K} \partial_iV$ and $q_j = 0$ for $j \notin B$. Since $V$ is convex, this must be the $t$-optimal vector as we desire to show. It also follows that $\overline{t} < \infty$, because at $t = \infty$ the minimum possible posterior variance is zero, which cannot be achieved by attending only to a subset $B$ of sources.

It remains to prove that $\tilde{\alpha}_i$ is positive for $1 \leq i \leq k$. To this end, define $\tilde{Q} = \diag(q_1, \dots, q_k)$ to be the $k \times k$ top-left sub-matrix of $Q$, and let
\begin{equation}\label{eq:tilde gamma}
\tilde{\gamma} = ((\Sigma_{TL})^{-1} + \tilde{Q})^{-1} \cdot \tilde{\alpha}. 
\end{equation}
We will show that $\tilde{\gamma}$ is just the first $k$ coordinates of $\gamma$. Indeed, for $1 \leq i \leq k$, $\tilde{\gamma}_i$ is by definition the covariance between $\theta_i$ and $\tilde{\omega} = \tilde{\alpha}' \cdot \tilde{\theta}$ under the posterior belief at time $\underline{t}$. Since $\omega = \tilde{\omega} + \alpha^{*'} \cdot \theta^*$, and the vector $\theta^*$ is by construction independent of $\theta_i$, we deduce that $\mbox{Cov}(\theta_i, \tilde{\omega}) = \mbox{Cov}(\theta_i, \omega)$. Thus $\tilde{\gamma}_i = \gamma_i$ as desired.

Given this, Lemma \ref{lemma:gamma positive} tells us that $\tilde{\gamma}_i$ is the same positive number for $1 \leq i \leq k$. Rewriting (\ref{eq:tilde gamma}) as $\tilde{\alpha} = ((\Sigma_{TL})^{-1} + \tilde{Q}) \cdot  \tilde{\gamma}$, we see that $\tilde{\alpha}_i$ is proportional to the $i$-th row sum of the matrix $(\Sigma_{TL})^{-1} + \tilde{Q}$, which is just the row sum of $(\Sigma_{TL})^{-1}$ plus $q_i$. By \citet{CarlsonMarkham}, if $\Sigma^{-1}$ is (strictly) diagonally-dominant, then so is $(\Sigma_{TL})^{-1}$ for any principal sub-matrix $\Sigma_{TL}$ (because $(\Sigma_{TL})^{-1}$ is the Schur complement of $\Sigma^{-1}$ with respect to its bottom-right block). So the row sums of $(\Sigma_{TL})^{-1}$ are all strictly positive, implying $\tilde{\alpha}_i > 0$.
\end{proof}

\subsection{Piecing Together Different Stages}
We now apply Lemma \ref{lemma:earlier stages} repeatedly to prove Theorem \ref{thm:general}. Continuing to assume strict diagonal dominance, we can apply Lemma \ref{lemma:earlier stages} with $\underline{t} = 0$ and deduce that up to some time $t_1 = \overline{t} > 0$, $t$-optimality can be achieved by a constant attention strategy supported on $B_1 = \argmin_{1 \leq i \leq K} \partial_i V(0)$. Applying Lemma \ref{lemma:earlier stages} again with $\underline{t} = t_1$, we know that the agent can maintain $t$-optimality from time $t_1$ to some time $t_2$ with a constant attention strategy supported on $B_2 = \argmin_{1 \leq i \leq K} \partial_i V(n(t_1))$. So on and so forth. The sets $\emptyset = B_0, B_1, B_2, \dots$ are nested by construction, so eventually $B_m = \{1, \dots, K\}$ . This delivers the result. 

\subsection{The Case of Weak Diagonal Dominance}\label{appx:zero weights}

Here we demonstrate how to prove Theorem \ref{thm:general} assuming only that $\Sigma^{-1}$ is weakly diagonally-dominant. The new difficulty is that in the proof of Lemma \ref{lemma:earlier stages}, we cannot conclude the optimal attention allocation (which is proportional to $\widetilde{\alpha}$) has \emph{strictly} positive coordinates on $B$. Thus the agent does not necessarily mix over \emph{all} of the sources that maximize marginal reduction of variance. This might lead to the failure of Theorem \ref{thm:general} for two reasons. First, it is possible that the agent optimally divides attention across a \emph{subset} of the sources that he has paid attention to in the past, which would violate the requirement of nested observation sets. Second, when a new source achieves maximal marginal value, the agent might (not attend to it and) use a different mixture over the sources previously sampled, which would violate the requirement of constant attention allocation for a given observation set. 

We now show that neither occurs in our setting. In response to the first concern above, note that we can still follow the proof of Lemma \ref{lemma:earlier stages} to deduce that the optimal instantaneous attention $\widetilde{\alpha}_i$ given to a source $i \in \argmin_{j} \partial_j V(t)$ is proportional to the $i$-th row sum of $(\Sigma_{TL})^{-1}$ plus $q_i$. Since $(\Sigma_{TL})^{-1}$ is weakly diagonally-dominant, its row sums are weakly positive. Thus $\widetilde{\alpha}_i > 0$ whenever  $q_i > 0$. In words, any source that has received attention in the past will be allocated strictly positive attention at every future instant. 

To address the second concern, consider two times $\widetilde{t} < \widehat{t}$ with $\argmin_{j} \partial_j V(n(\widetilde{t})) \subsetneq \argmin_{j} \partial_j V(n(\widehat{t}))$. Reordering the attributes, we may assume that at time $\widetilde{t}$ the first $\widetilde{k}$ sources have the highest marginal value, whereas at time $\widehat{t}$ this set expands to the first $\widehat{k} > \widetilde{k}$ sources. Let $\widetilde{\alpha} \in \mathbb{R}^{\widetilde{k}}$ and $\widehat{\alpha} \in \mathbb{R}^{\widehat{k}}$ be the optimal attentions associated with these subsets, as given by (\ref{eq:tilde alpha}). We want to show that if $\widehat{\alpha}$ is supported on the same set of sources as $\widetilde{\alpha}$, then $\widehat{\alpha}$ coincides with $\widetilde{\alpha}$ on their support. Indeed, by definition of $\widehat{\alpha}$ (see the proof of Lemma \ref{lemma:earlier stages}), 
\[
\omega = \sum_{i \leq \widehat{k}} \widehat{\alpha}_i \theta_i + \text{ residual term orthogonal to } \theta_1, \dots, \theta_{\widehat{k}}.
\]
If $\widehat{\alpha}$ has the same support as $\widetilde{\alpha}$, then the above implies
\[
\omega = \sum_{i \leq \widetilde{k}} \widehat{\alpha}_i \theta_i + \text{ residual term orthogonal to } \theta_1, \dots, \theta_{\widetilde{k}},
\]
where we use the fact that any term orthogonal to the first $\widehat{k}$ attributes is clearly orthogonal to the first $\widetilde{k}$ attributes. This last representation of $\omega$ reduces to the definition of $\widetilde{\alpha}$. Hence $\widehat{\alpha}_i = \widetilde{\alpha}_i$ for $1 \leq i \leq \widetilde{k}$, as we desire to prove. 

\subsection{The Case of Perpetual Substitutes or Perpetual Complements}\label{appx:alternative sufficient conditions}

We now prove Theorem \ref{thm:general} under Assumption \ref{assmp:Substitutes} or \ref{assmp:Complements}. Our proof above uses the diagonal dominance assumption at two places. It is crucial for proving Lemma \ref{lemma:gamma positive} (that is, the coordinates of $\gamma$ with greatest magnitude are all positive), as well as for showing that the transformed weight vector $\tilde{\alpha}$ is positive in the proof of Lemma \ref{lemma:earlier stages}. Thus we just need to verify these two steps under the alternative assumptions. 

Lemma \ref{lemma:gamma positive} continues to hold because, as we show in the proof of Propositions \ref{prop:substitutes} and \ref{prop:complements} in Online Appendix \ref{appx:substitutes/complements}, these alternative assumptions imply that $\gamma = (\Sigma^{-1} + Q)^{-1} \cdot \alpha$ has non-negative coordinates for any $q \geq 0$. It trivially follows that those coordinates with maximal absolute value must be strictly positive. 

As for $\tilde{\alpha}$ in the proof of Lemma \ref{lemma:earlier stages}, first consider Assumption \ref{assmp:Substitutes} which imposes that $\Sigma^{-1}$ is an $M$-matrix. We use the matrix identity $(\Sigma_{TL})^{-1} \cdot \Sigma_{TR} = - (\Sigma^{-1})_{TR} \cdot [(\Sigma^{-1})_{BR}]^{-1}$, which can be proved using $\Sigma \cdot \Sigma^{-1} = I_K$ and comparing the top-right block. By assumption, $(\Sigma^{-1})_{TR}$ has non-positive entries, since it only consists of off-diagonal entries of $\Sigma^{-1}$. Moreover, $[(\Sigma^{-1})_{BR}]^{-1}$ has non-negative entries, since $(\Sigma^{-1})_{BR}$ is an $M$-matrix. We thus conclude that $(\Sigma_{TL})^{-1} \cdot \Sigma_{TR}$ is a matrix with non-negative entries. From \eqref{eq:tilde alpha} we have $\tilde{\alpha} = (\Sigma_{TL})^{-1} \cdot (\Sigma_{TL}, ~\Sigma_{TR}) \cdot \alpha = (I_k, ~(\Sigma_{TL})^{-1}\Sigma_{TR}) \cdot \alpha$. So each coordinate of $\tilde{\alpha}$ is larger than the corresponding coordinate of $\alpha$, and is thus strictly positive. 

If instead Assumption \ref{assmp:Complements} is satisfied, then $\Sigma$ itself is an $M$-matrix, and so is the principal sub-matrix $\Sigma_{TL}$. Thus $(\Sigma_{TL})^{-1}$ has non-negative entries off the diagonal and strictly positive entries on the diagonal. From \eqref{eq:tilde gamma} we have $\tilde{\alpha} = ((\Sigma_{TL})^{-1} + \tilde{Q}) \cdot \tilde{\gamma}$. Since $\tilde{\gamma}$ is a positive vector (with equal coordinates), we deduce $\tilde{\alpha} \gg 0$. This completes the proof of Theorem \ref{thm:general}.

\section{Algorithm for Computing the Optimal Strategy} \label{appx:algorithm}

Here we provide an algorithm for recursively finding the times $t_k$ and sets $B_k$ in Theorem \ref{thm:general}. Set $Q_0$ to be the $K \times K$ matrix of zeros, and $t_0 = 0$. For each stage $k\geq 1$:

\bigskip

\noindent \textbf{1. Computation of the observation set $B_k$.} Define the $K \times 1$ vector $\gamma^k=(\Sigma^{-1} + Q_{k-1})^{-1} \cdot \alpha$ where $\Sigma$ is the prior covariance matrix, and $\alpha$ is the weight vector. The set of attributes that the agent attends to in stage $k$ is 
\[
    B_k=\textstyle \argmax_i \vert \gamma^k_i \vert.
\]
These sources have highest marginal reduction of posterior variance (see Lemma \ref{lemma:V partials}). 
    
\bigskip

\noindent \textbf{2. Computation of the constant attention allocation in stage $k$.} 
If $B_k$ is the set of all sources, then we are already in the last stage and the algorithm ends. Otherwise let $\ell = \vert B_k \vert < K$. We can re-order the attributes so that the $\ell$ attributes in $B_k$ are the first $\ell$ attributes. In an abuse of notation, let $\Sigma$ be the covariance matrix for the re-ordered attribute vector $\theta$. Define $\Sigma_{TL}$ to be the $\ell \times \ell$ top-left submatrix of $\Sigma$ and $\Sigma_{TR}$ to be the $\ell \times (K-\ell)$ top-right block. Finally let
\[
    \alpha^k = (\Sigma_{TL})^{-1} \cdot (\Sigma_{TL}, ~\Sigma_{TR}) \cdot \alpha
\]
be an $\ell \times 1$ vector. The agent's optimal attention allocation in stage $k$ is proportional to $\alpha^k$: 
\[
    \beta_i^k = \left\{\begin{array}{cl}
    \alpha_i^k/\sum_i \alpha_i^k & \mbox{if } i\leq \ell \\
    0 & \mbox{otherwise}
    \end{array} \right.
\]
As the agent acquires information in this mixture during stage $k$, the marginal values of learning about different attributes in $B_k$ remain the same, and strictly higher than learning about any attribute outside of the set.

\bigskip

\noindent \textbf{3. Computation of the next time $t_k$.} For arbitrary $t$, define 
\[
    Q^k(t):= Q_{k-1} + (t-t_{k-1})\cdot \diag(\beta^k). 
\] 
Let $t_k$ be the smallest $t > t_{k-1}$ such that the coordinates maximizing $(\Sigma^{-1} + Q^k(t))^{-1} \cdot \alpha$ are a strict superset of $B_k$. At this time, the marginal value of some attribute(s) outside of $B_k$ equalizes the attributes in $B_k$, and stage $k+1$ commences, with $Q_k = Q^k(t_k)$. 
    
The time $t_k$ can be computed as follows. For each source $j > \ell$, consider the following (polynomial) equation in $t$:
\begin{equation}\label{eq:computation of next time}
    e_j' \cdot (\Sigma^{-1}+Q^{k}(t))^{-1} \cdot \alpha = \pm e_1' \cdot (\Sigma^{-1}+Q^{k}(t))^{-1} \cdot \alpha.
\end{equation}
Any solution $t > t_{k-1}$ is a time at which source $j$ has the same marginal value as sources $1, \dots, \ell$. So $t_k$ is  the smallest such solution $t$ across all $j > \ell$.

\clearpage 

\input{onlineappendix}

\end{document}

%% file: onlineappendix.tex

\pagenumbering{gobble}
\small
\pagenumbering{arabic}
\renewcommand{\thesection}{\Alph{section}}
\renewcommand\thesubsection{\thesection.\arabic{subsection}}

\setcounter{section}{14}

 


\section{For Online Publication} 

\subsection{Uniqueness of Optimal Information Acquisition}\label{appx:uniqueness}
By Lemma \ref{lemma:Greenshtein}, whenever a uniformly optimal strategy exists, it is the optimal attention strategy regardless of the form of $u(\tau, a, \omega)$. Without further assumptions on $u$, there could exist other optimal attention strategies. For example, consider the payoff function used in the proof of Lemma \ref{lemma:uniform optimal converse}. Under this payoff function, the agent always stops at some fixed time $t$. Hence any strategy that achieves the $t$-optimal vector $n(t)$ would be optimal in this problem. 

Nonetheless, such examples can be ruled out by an assumption on the agent's stopping rule:
\begin{assumption} \label{assmp:stop} 
Given any attention allocation strategy $S$, any history of signal realizations up to time $t$ such that the agent has not stopped, and any $t' > t$, there exists a positive probability of continuation histories such that the agent optimally stops in the interval $(t, t']$.
\end{assumption}



\begin{proposition}
Suppose Assumption \ref{assmp:waiting is costly} holds strictly, and Assumption \ref{assmp:stop} is satisfied. Then, any optimal attention strategy $S$ induces the same posterior variance as the uniformly optimal strategy $S^*$ at every history where the agent has not stopped. Consequently, the two strategies induce the same cumulative attention vectors, and coincide at almost every time before stopping. 
\end{proposition}

\begin{proof}
Suppose $S$ induces larger posterior variance than $S^*$ at some time $t$, then by continuity the same holds from time $t$ to some later time $t' > t$ along this history. By assumption, the agent stops between these times with positive probability. Thus there is positive probability that the agent stops with posterior variance strictly larger than the minimal variance. From the proof of Lemma \ref{lemma:Greenshtein}, we see that payoff under $S$ is strictly below $S^*$, contradicting the optimality of $S$. The second part of the result follows from the uniqueness of $n(t)$, and the fact that $\beta(t)$ integrates to $n(t)$.
\end{proof}

Although Assumption \ref{assmp:stop} is stated in terms of the endogenous stopping rule, it is satisfied in any problem where the agent always stops to take some action when he has an extremely high (or low) expectation about $\omega$. This is guaranteed if extreme values of $\omega$ agree on the optimal action and the marginal cost of delay is bounded away from zero. These conditions on the primitives are rather weak, and are satisfied in many applications such as binary choice with linear waiting cost.

\subsection{Non-existence of Uniformly Optimal Strategy}

\subsubsection{Counterexample for $K=2$} \label{appx:counterexample}
The example below illustrates how and why Theorem \ref{thm:K=2} might fail without Assumption \ref{assmp:K=2}:

\begin{example}\label{ex:failure}
There are two unknown attributes with prior distribution
\[
\left(\begin{array}{c} \theta_1 \\ \theta_2 \end{array}\right) \sim \mathcal{N}\left( \left(\begin{array}{c} \mu_1 \\ \mu_2 \end{array}\right), \left(\begin{array}{cc} 10 & -3 \\ -3 & 1 \end{array} \right) \right).
\]
The agent wants to learn $\theta_1 + 4 \theta_2$.

Given $q_1$ units of attention devoted to learning $\theta_1$, and $q_2$ devoted to $\theta_2$, the agent's posterior variance about $\omega$ is given by (\ref{eq:V(q)}). Simplifying, we have $V(q_1,q_2)=\frac{2+16q_1 + q_2}{(1+q_1)(10+q_2)-9}$.
The $t$-optimal cumulative attention vectors $n(t)$ (see Section \ref{sec:proof sketch}) are defined to minimize $V(q_1,q_2)$ subject to $q_1, q_2 \geq 0$ and the budget constraint $q_1+q_2\leq t$.

These vectors do not evolve monotonically: Initially, the marginal value of learning $\theta_1$ exceeds that of learning $\theta_2$, since the agent has greater prior uncertainty about $\theta_1$ (even accounting for the difference in payoff weights). Thus at all times $t < 1/4$, the $t$-optimal vector is $(t,0)$, and the agent learns only about attribute 1. 

After a quarter-unit of time devoted to learning $\theta_1$, the agent's posterior covariance matrix becomes 
$ \left(\begin{array}{cc} 20/7 & -6/7 \\ -6/7 & 5/14 \end{array} \right)$. Note that the two sources have equal marginal values at $t = 1/4$, since $\gamma_1 = \frac{-4}{7}$ and $\gamma_2 = \frac{4}{7}$ have the same absolute value. However, to maintain equal marginal values at future instants, it would be optimal to take attention away from attribute 1 and re-distribute it to attribute 2. Specifically, at all times $t \in [1/4, 1]$ the $t$-optimal vector is given by $n(t) = \left(\frac{-t+1}{3}, \frac{4t-1}{3}\right)$, and the optimal cumulative attention toward $\theta_1$ is decreasing in this interval.

Consequently, there does not exist a uniformly optimal strategy in this example (Lemma \ref{lemma:monotonicity}). Hence the optimal information acquisition strategy varies according to when the agent expects to stop, and Theorem \ref{thm:K=2} cannot hold independently of the payoff criterion. 



    
    
\end{example}

\subsubsection{Necessity of Assumption \ref{assmp:K=2} for Theorem \ref{thm:K=2}} \label{appx:K=2 necessity}
We show here that when $K = 2$, the assumption $cov_1 + cov_2 \geq 0$ is also \emph{necessary} for the existence of a uniformly optimal strategy. The result generalizes Example \ref{ex:failure} above. 
\begin{proposition}\label{prop:K=2converse}
Suppose $K = 2$ and Assumption \ref{assmp:K=2} is violated. Then a uniformly optimal strategy does not exist. 
\end{proposition}

\begin{proof}
Suppose that $cov_1 + cov_2 < 0$. First note that one of $cov_1, cov_2$ is positive, because $\alpha_1cov_1 + \alpha_2cov_2$ = $\alpha' \Sigma \alpha > 0$. So without loss we can assume $cov_2 > 0 > -cov_2 > cov_1$. Moreover, from $\alpha_1 cov_1 + \alpha_2 cov_2 > 0$ we obtain $\alpha_2 > \alpha_1$ and hence $x_2 > x_1$. Below we show the $t$-optimal attention vector $n(t)$ is non-monotonic. 

Suppose $\frac{-(cov_1+cov_2)}{x_2} < t < \frac{-(cov_1+cov_2)}{x_1}$, then by (\ref{eq:K=2partials}) we have $\partial_1V(0,t) < \partial_2V(0,t)$ and $\partial_1V(t,0) > \partial_2V(t,0)$. These imply that $n(t)$ is interior, and the first-order condition yields 
\[
x_1n_2(t) + cov_1 = -(x_2n_1(t) + cov_2),
\]
where we use the fact that for $t$ in this range  $x_1q_2 + cov_1$ is always negative. Together with $n_1(t) + n_2(t) = t$, we can solve that $n(t) = (\frac{-x_1t-cov_1-cov_2}{x_2-x_1}, \frac{x_2t+cov_1+cov_2}{x_2-x_1})$. Thus as $t$ increases in this range, $n_1(t)$ actually decreases. So a uniformly optimal strategy does not exist. 
\end{proof}

\subsubsection{Counterexample for $K=3$} \label{appx:K=3 counterexample}

We present another example where a uniformly optimal strategy does not exist. In the following example, the three attributes have positive correlation with each other, but not positive \emph{partial correlation}. Thus the example illustrates the subtlety of Assumption \ref{assmp:Substitutes}.

Let the primitives be $K = 3$, $\alpha_1 = \alpha_2 = 1$, $\alpha_3 = 20$ and 
\[
    \Sigma = \left(\begin{array}{ccc} 19 & 3 & 0 \\ 3 & 5 & 3 \\ 0 & 3 & 2 \end{array} \right) \qquad \Sigma^{-1} = \left(\begin{array}{ccc} 1 & -6 & 9 \\ -6 & 38 & -57 \\ 9 & -57 & 86 \end{array} \right)
\]
The choice of $\Sigma_{11} = 19$ is such that $\Sigma$ has determinant exactly $1$. This makes it easier to calculate its inverse matrix, while still ensuring that $\Sigma$ is positive-definite. We highlight the fact that $[\Sigma^{-1}]_{13} = 9$ is positive, so this example does not satisfy Assumption \ref{assmp:Substitutes}. 
    
Consider the cumulative attention vector $q = (1, 14, 0)'$. Simple calculation gives $(\Sigma^{-1} + Q) \cdot (-1, 1, 1)' = (1, 1, 20)' = \alpha$. Thus $\gamma(q) = (\Sigma^{-1} + Q)^{-1}\alpha = (-1, 1, 1)'$. Since the three coordinates of $\gamma$ have equal absolute value, the sources have equal marginal reduction of $V$ at the attention vector $q$. This means $q$ is the $t$-optimal vector for $t = 15$. 

Next consider a different attention vector $\hat{q} = (0, 15, 20)$. We can similarly calculate that $(\Sigma^{-1} + \hat{Q}) \cdot (-1, 1, 1)' = (2, 2, 40)' = 2\alpha$. So $\gamma(\hat{q}) = (-1/2, 1/2, 1/2)$. By the same reasoning, $\hat{q}$ is the $t$-optimal vector for $t = 35$. Hence we see that when $t$ increases from $15$ to $35$, the optimal amount of attention devoted to $\theta_1$ decreases from $1$ to $0$. This implies that a uniformly optimal strategy does not exist in this example.

\subsection{When are Sources Substitutes/Complements?} \label{appx:substitutes/complements}

Since more information \emph{reduces} the posterior variance $V$, we define two sources $i$ and $j$ to be substitutes if the cross-partial $\partial_{ij} V(q)$ is \emph{non-negative} at any cumulative attention vector $q \geq 0$. The following result shows that Assumption \ref{assmp:Substitutes} precisely characterizes when two sources are substitutes. 

\begin{proposition}\label{prop:substitutes} 
Given any positive payoff weight vector $\alpha$. The following conditions on the prior covariance matrix $\Sigma$ are equivalent:
\begin{enumerate}
    \item $\Sigma^{-1}$ has non-positive off-diagonal entries;
    \item Every pair of sources $i \neq j \in \{1, \dots, K\}$ are substitutes in the sense that $\partial_{ij} V(q) \geq 0$ at every cumulative attention vector $q \in \mathbb{R}_{+}^K$.
\end{enumerate}
\end{proposition}

\begin{proof}
In one direction, suppose $\Sigma^{-1}$ has non-positive off-diagonal entries. Then for any vector $q \geq 0$ and corresponding diagonal matrix $Q = \diag(q)$, $\Sigma^{-1} + Q$ satisfies the same property. Thus the positive-definite matrix $\Sigma^{-1} + Q$ is an \emph{$M$-matrix} whose inverse is known to have non-negative entries off the diagonal and strictly positive entries on the diagonal; see e.g.\ \citet{Plemmons1977}. It follows that $\gamma = (\Sigma^{-1} + Q)^{-1} \cdot \alpha$ is a vector with positive entries. Thus by Lemma \ref{lemma:V partials}, 
\[
\partial_{ij} V(q) = 2 \gamma_i \gamma_j \cdot [(\Sigma^{-1} + Q)^{-1}]_{ij} \geq 0.
\]

Conversely, suppose $\partial_{ij}V(q) \geq 0$ for every $q \geq 0$. Let us choose any vector $q$ whose coordinates are all sufficiently large. Then 
\[
(\Sigma^{-1} + Q)^{-1} = [(\Sigma^{-1} \cdot Q^{-1} + I_K) \cdot Q]^{-1} = Q^{-1} (\Sigma^{-1} \cdot Q^{-1} + I_K)^{-1} = Q^{-1} (I_K + o(1)),
\]
where $o(1)$ is a matrix that goes to zero as $q_1, \dots, q_K$ all go to infinity (at any rate). Note that when the $o(1)$ term is sufficiently small, $(I_K + o(1)) \cdot \alpha$ is a vector with positive coordinates. Thus, $\gamma = (\Sigma^{-1} + Q)^{-1} \alpha = Q^{-1} (I_K + o(1)) \alpha$ has positive coordinates. Together with $\partial_{12}V(q) \geq 0$ and Lemma \ref{lemma:V partials}, this implies $[(\Sigma^{-1} + Q)^{-1}]_{12} \geq 0$ for any such $q$. Using the matrix inverse formula, we further obtain that the determinant of the cofactor matrix $[\Sigma^{-1} + Q]_{-21}$ (i.e., the sub-matrix of $\Sigma^{-1} + Q$ with the second row and first column removed) must be non-positive. Expanding this determinant using permutations, it is easy to see that it contains the term $[\Sigma^{-1}]_{12} \cdot \prod_{k = 3}^{K} q_k$, which is in fact the dominant term when each $q_k$ is sufficiently large. Hence we deduce $[\Sigma^{-1}]_{12} \leq 0$, and similarly $[\Sigma^{-1}]_{ij} \leq 0$ for every pair $i \neq j$.
\end{proof}

We next show that Assumption \ref{assmp:Complements} characterizes when two sources are complements.
\begin{proposition}\label{prop:complements} 
Given any positive payoff weight vector $\alpha$. The following conditions on the prior covariance matrix $\Sigma$ are equivalent:
\begin{enumerate}
    \item $\Sigma$ has non-positive off-diagonal entries and $\Sigma \cdot \alpha$ has non-negative coordinates;
    \item Every pair of sources $i \neq j \in \{1, \dots, K\}$ are complements in the sense that $\partial_{ij} V(q) \leq 0$ at every cumulative attention vector $q \in \mathbb{R}_{+}^K$.
\end{enumerate}
\end{proposition}

\begin{proof}
To show the first condition implies the second, we note that since $\Sigma$ is assumed to be an $M$-matrix, $\Sigma^{-1}$ is an inverse $M$-matrix. By Theorem 3 in \citet{Johnson1982}, $\Sigma^{-1}+Q$ is also an inverse $M$-matrix. This implies that for any attention vector $q$, the posterior covariance matrix $(\Sigma^{-1}+Q)^{-1}$ must be an $M$-matrix with non-positive off-diagonal entries. We claim that the vector $\gamma = (\Sigma^{-1}+Q)^{-1} \cdot \alpha$ has non-negative coordinates. Once this is shown, Lemma \ref{lemma:V partials} will imply that $\partial_{ij}V = 2 \gamma_i \gamma_j [(\Sigma^{-1}+Q)^{-1}]_{ij} \leq 0$ for every pair $i \neq j$.

By assumption, if $q$ is the zero vector, then $\gamma = \Sigma \cdot \alpha$ indeed has non-negative coordinates. Our goal below is to show this also holds at any $q \geq 0$. To do this, we first work under the stronger assumption that $\Sigma \cdot \alpha$ has strictly positive coordinates, so that $\gamma(0)$ is a positive vector. Recall from the proof of Lemma \ref{lemma:V partials} that when $\gamma_i$ is viewed as a function of $q$, its partial derivatives are given by
\[
\frac{\partial \gamma_i}{\partial q_j} = - \gamma_j \cdot [(\Sigma^{-1}+Q)^{-1}]_{ij} \quad \text{for each } j \in \{1, \dots, K\}. 
\]
Now since $\Sigma^{-1}$ is positive-definite, we can choose $\epsilon > 0$ such that $\Sigma^{-1} \geq \epsilon I_K$ in the matrix order. Fixing this $\epsilon$, we consider the $K$ functions
\[
f_i(q_1, \dots, q_K) = (q_i + \epsilon) \cdot \gamma_i(q_1, \dots, q_K).
\]
For every $j \neq i$, $\partial f_i / \partial q_j = (q_i + \epsilon) \cdot -\gamma_j \cdot [(\Sigma^{-1}+Q)^{-1}]_{ij}$. This product is non-negative whenever $\gamma_j \geq 0$, since $[(\Sigma^{-1}+Q)^{-1}]_{ij} \leq 0$ as shown above. On the other hand, by the product rule, 
\[
\frac{\partial f_i}{\partial q_i} = \gamma_i \cdot [1 - (q_i + \epsilon)[(\Sigma^{-1}+Q)^{-1}]_{ii}].
\]
This is non-negative whenever $\gamma_i \geq 0$, since $[(\Sigma^{-1}+Q)^{-1}]_{ii} \leq [(\epsilon I_K+Q)^{-1}]_{ii} = (q_i+\epsilon)^{-1}$ where the inequality uses standard properties of the matrix order. 

Hence, we have shown that whenever $f_1(q), \dots, f_K(q)$ are all non-negative, their derivatives with respect to each $q_i$ are also non-negative. Moreover, we know from our stronger assumption that $f_1, \dots, f_K$ are strictly positive at $q = 0$. These together imply $f_1(q), \dots, f_K(q)$ are always strictly positive, by the following argument. Suppose for contradiction that there exists some $i$ and some $q \geq 0$ such that $f_i(q) \leq 0$. By continuity of $f$ and a limit argument, we may assume $q$ is minimal in the sense that at every $q' < q$, $f_j(q')$ is positive for every $j$. Thus, if we let $g(t) = f_i(t \cdot q)$ be the one-variable function defined for $t \in [0, 1]$, we have $g(0) > 0$, $g'(t) \geq 0$ for $t \in [0, 1)$ and $g(1) \leq 0$. This contradicts the Mean Value Theorem. 

It follows that $\gamma(0)$ being positive implies $\gamma(q)$ is positive. We now extend this result to the case of weak inequalities. Suppose $\gamma(0) = \Sigma \cdot \alpha$ has non-negative coordinates. Then by considering $\Sigma + \delta I_K$ instead of $\Sigma$, we see that the corresponding $\gamma$ vector is strictly positive at $q = 0$. The above analysis applied to the $M$-matrix $\Sigma + \delta I_K$ thus implies $[(\Sigma + \delta I_K)^{-1} + Q]^{-1} \cdot \alpha$ is positive for any attention vector $q \geq 0$. Letting $\delta \to 0$ yields $[\Sigma^{-1}+Q]^{-1} \cdot \alpha \geq 0$, as we desire to show. This completes the proof that the first condition in the proposition guarantees complementarity. 

\bigskip

Turning to the converse, we assume at every $q \geq 0$, $\partial_{ij} V = 2\gamma_i(q) \cdot \gamma_j(q) \cdot [(\Sigma^{-1}+Q)^{-1}]_{ij} \leq 0$. Let us choose $q$ such that $\gamma_i(q)$ is nonzero for each $i$; its existence will be verified later. For this $q$, we claim that $\gamma_i(q)$ must be positive for each $i$. Indeed, since $\gamma = (\Sigma^{-1}+Q)^{-1} \alpha$, we have $\alpha' \gamma = \alpha' (\Sigma^{-1}+Q)^{-1} \alpha > 0$. Thus $\gamma$ must have at least one positive coordinate. Suppose for contradiction that $\gamma$ has a negative coordinate, then we can without loss assume $\gamma_i$ is positive for $i \leq k$ and negative for $i > k$, where $k$ satisfies $1 \leq k < K$. Using $\partial_{ij} V \leq 0$, we deduce $[(\Sigma^{-1}+Q)^{-1}]_{ij} \geq 0$ for every pair $i \leq k$ and $j  > k$. Now let us decompose $(\Sigma^{-1}+Q)^{-1}$ into four block sub-matrices: $[(\Sigma^{-1}+Q)^{-1}]_{TL}$, $[(\Sigma^{-1}+Q)^{-1}]_{TR}$, $[(\Sigma^{-1}+Q)^{-1}]_{BL}$ and $[(\Sigma^{-1}+Q)^{-1}]_{BR}$ are the top-left $k \times k$, top-right $k \times (K-k)$, bottom-left $(K-k) \times k$ and bottom-right $(K-k) \times (K-k)$ sub-matrices of $(\Sigma^{-1}+Q)^{-1}$ respectively. Recall that $\gamma = (\Sigma^{-1}+Q)^{-1} \cdot \alpha$. By looking at the last $K-k$ coordinates, we obtain 
\begin{align*}
(\gamma_{k+1}, \dots, \gamma_K)' &= \left([(\Sigma^{-1}+Q)^{-1}]_{BL}, [(\Sigma^{-1}+Q)^{-1}]_{BR} \right) \cdot \alpha \\
&= [(\Sigma^{-1}+Q)^{-1}]_{BL} \cdot (\alpha_1, \dots, \alpha_k)' + [(\Sigma^{-1}+Q)^{-1}]_{BR} \cdot (\alpha_{k+1}, \dots, \alpha_{K})'.
\end{align*}
The preceding analysis tells us that $[(\Sigma^{-1}+Q)^{-1}]_{BL}$ has non-negative entries, so the vector $[(\Sigma^{-1}+Q)^{-1}]_{BL} \cdot (\alpha_1, \dots, \alpha_k)'$ is non-negative. In addition, $[(\Sigma^{-1}+Q)^{-1}]_{BR}$ is positive-definite, so the vector $[(\Sigma^{-1}+Q)^{-1}]_{BR} \cdot (\alpha_{k+1}, \dots, \alpha_{K})'$ has at least one positive coordinate. Thus, the above displayed equation contradicts the assumption that $\gamma_{k+1}, \dots, \gamma_K$ are all negative. 

We thus know that if $\gamma(0)$ has nonzero coordinates, then it is in fact a positive vector. Complementarity further requires $\Sigma_{ij} \leq 0$ for all $i \neq j$, which would complete the proof. In the general case, $\gamma(0)$ may have some coordinates equal to zero, so we instead look for $q$ close to the zero vector such that $\gamma(q)$ has nonzero coordinates. To see why such $q$ exists, note that we can calculate the matrix inverse $(\Sigma^{-1} + Q)^{-1}$ using the determinants of cofactor matrices. From this we see that modulo a multiplicative factor of $[\det(\Sigma^{-1} + Q)]^{-1}$, each $\gamma_i(q) = (\Sigma^{-1} + Q)^{-1} \alpha$ is a nonzero multi-linear polynomial in the $K-1$ variables $\{q_{j}\}_{j \neq i}$ (with leading term $\alpha_i \cdot \prod_{j \neq i} q_j$). Thus, the set of vectors $q$ that make $\gamma_i(q)$ equal to zero has measure zero. It follows that we can choose $q$ with arbitrarily small coordinates, such that $\gamma_i(q)$ is nonzero for all $i$. By the earlier analysis, $\gamma(q)$ is a positive vector, and $[(\Sigma^{-1} + Q)^{-1}]_{ij} \leq 0$ for all $i \neq j$. Letting $q \to 0$, we conclude by continuity that $\gamma(0) \geq 0$ and $\Sigma_{ij} \leq 0$. This is what we desire to show.  
\end{proof}  



\subsection{Proof of Proposition \ref{prop:binary choice FSS}}\label{appx:binary choice}

\subsubsection{Preliminary Analysis}
Suppose the agent's prior about the two unknown payoffs is normal with covariance matrix $ \left(\begin{array}{cc}
\Sigma_{11} & -\Sigma_{12} \\
-\Sigma_{21} & \Sigma_{22} \end{array}\right)$. Throughout we assume $\Sigma_{11} \geq \Sigma_{22}$. The objective is to maximize
\[
\mathbb{E}[\mathbb{E}[\max\{v_1, v_2\} \mid \mathcal{F}_{\tau} ] - c\tau ],
\]
To reduce this problem to our main model, we use $\max\{v_1, v_2\} = \max\{v_1 - v_2, 0\} + v_2$ to rewrite the objective as 
\[
\mathbb{E}[\mathbb{E}[\max\{v_1 - v_2, 0\}  \mid \mathcal{F}_{\tau} ] - c\tau] + \mathbb{E}[ \mathbb{E}[v_2 \mid F_{\tau}]]
\]
The posterior expectations of $v_2$, $M_t = \mathbb{E}[v_2 \mid \mathcal{F}_{t}]$, form a continuous martingale with continuous paths. Moreover, the family $\{M_t\}$ are uniformly integrable because they are conditional expectations of an integrable random variable $v_2$. Thus we can apply Doob's Optional Sampling Theorem to deduce 
$
\mathbb{E}[ \mathbb{E}[v_2 \mid \mathcal{F}_{\tau}]] = \mathbb{E}[v_2],
$
which is just the prior expectation of $v_2$ (and does not depend on the agent's strategy). It follows that the agent simply maximizes 
\[
\mathbb{E}[\mathbb{E}[\max\{v_1 - v_2, 0\} \mid \mathcal{F}_{\tau} ] - c\tau ].
\]

As a corollary, the payoff difference $v_1 - v_2$ is a sufficient statistic for the agent's decision. Now if we let $\theta_1 = v_1$, $\theta_2 = -v_2$, then the prior covariance matrix about $\theta$ is simply $\Sigma := \left(\begin{array}{cc}
\Sigma_{11} & \Sigma_{12} \\
\Sigma_{21} & \Sigma_{22} \end{array}\right)$. This returns our main model with prior covariance matrix $\Sigma$ and payoff-relevant state $\omega = v_1 - v_2 = \theta_1 + \theta_2$. Since $\alpha_1 = \alpha_2 = 1$, our Theorem \ref{thm:K=2} applies and yields Corollary \ref{corr:binary choice attention}. 

For the subsequent analysis, we need to keep track of how fast the posterior variance of $\omega$ evolves over time. These (minimal) posterior variances are given below:
\begin{lemma}\label{lemma:binary choice variances}
Suppose $\Sigma_{11} \geq \Sigma_{22}$. When adopting the optimal information acquisition strategy, the agent's posterior variance of $\omega = \theta_1 + \theta_2$ at time $t$ is given by 
\[
\sigma_t^2 = \begin{cases} \frac{\Sigma_{11} + \Sigma_{22} + 2 \Sigma_{12} + \det(\Sigma)t}{1 + \Sigma_{11} t} &\text{if } t \leq t_1^*; \\
\frac{4 \det(\Sigma)}{\Sigma_{11} + \Sigma_{22} - 2\Sigma_{12} + \det(\Sigma)t} &\text{if } t \geq t_1^*.
\end{cases}
\]
\end{lemma}

\begin{proof}
At time $t \leq t_1^*$, the posterior covariance matrix is 
\[
\left[\left(\begin{array}{cc}
\Sigma_{11} & \Sigma_{12} \\
\Sigma_{21} & \Sigma_{22} \end{array}\right)^{-1} + \left(\begin{array}{cc}
t & 0 \\
0 & 0 \end{array}\right)\right]^{-1} = \frac{1}{1 + \Sigma_{11}t} \left(\begin{array}{cc}
\Sigma_{11} & \Sigma_{12} \\
\Sigma_{21} & \Sigma_{22} + \det(\Sigma) t \end{array}\right).
\]
This gives the first part of the lemma. 

At time $t = t_1^*$, the levels of uncertainty about $\theta_1$ and $\theta_2$ have equalized. From this time on, each unit of time produces two normal signals $\theta_1 + \epsilon_1$ and $\theta_2 + \epsilon_2$ with $\epsilon_1$ and $\epsilon_2$ independently and identically distributed according to $\mathcal{N}(0,2)$. These two signals are informationally equivalent to their sum and difference $\theta_1 + \theta_2 + \epsilon_1 + \epsilon_2$ and $\theta_1 - \theta_2 + \epsilon_1 - \epsilon_2$, with i.i.d.\ noise terms $\epsilon_1 + \epsilon_2, \epsilon_1 - \epsilon_2 \sim \mathcal{N}(0,4)$. Note that as long as the agent's uncertainty about $\theta_1$ and $\theta_2$ are the same, $\theta_1 + \theta_2$ and $\theta_1 - \theta_2$ are independent. Thus, in terms of learning about $\theta_1 + \theta_2$, it is as if the agent receives only the signal $\theta_1 + \theta_2 + \mathcal{N}(0,4)$ over each unit of time. This observation enables us to calculate the posterior variance of $\theta_1 + \theta_2$ as follows: For any $t \geq t_1^*$,
\[
\sigma_t^2 = \left(\frac{1}{\sigma_{t_1^*}^2} + \frac{t-t_1^*}{4}\right)^{-1},
\]
where $\frac{1}{\sigma_{t_1^*}^2}$ is the belief precision of $\theta_1 + \theta_2$ at time $t_1^*$ and $\frac{t-t_1^*}{4}$ is the signal precision from those $\theta_1 + \theta_2 + \mathcal{N}(0,4)$ signals between time $t_1^*$ and time $t$. Plugging in $t_1^* = \frac{\Sigma_{11} - \Sigma_{22}}{\det(\Sigma)}$ and $\sigma_{t_1^*}^2 = \frac{2 \det(\Sigma)}{\Sigma_{11} - \Sigma_{12}}$ then yields the second part of the lemma.
\end{proof}

\subsubsection{Stopping Boundaries and Choice Accuracy}

Using the posterior variances characterized above, we can write down the process for the posterior expectation of $\omega$, which we denote by $Y_t = \mathbb{E}[\theta_1 + \theta_2 \mid \mathcal{F}_t]$: 
\begin{equation}\label{eq:binary choice belief process}
Y_t = Y_0 + \int_{0}^{t} \sqrt{\frac{\partial \sigma_s^2}{\partial s}} \cdot d B_s,
\end{equation}
where $B_s$ is a standard Brownian motion with respect to the filtration of the agent's information. The volatility term $\sqrt{\frac{\partial \sigma_s^2}{\partial s}}$ is such that the variance of posterior expectation $Y_t$ matches the reduction in posterior variance $\sigma_0^2 - \sigma_t^2$. This representation is a direct generalization of Lemma 1 in \cite{FudenbergStrackStrzalecki} and follows from standard results. 

Therefore, given any prior covariance matrix $\Sigma$ and any prior expectation $Y_0 = y$, the agent's problem can be rewritten as 
\begin{equation}\label{eq:binary choice objective}
\max_{\tau} \mathbb{E}[\max\{Y_{\tau}, 0\} - c \tau],
\end{equation}
where $\tau$ can be any stopping time adapted to the $Y$ process given above. We now define the value function $U(y, c, \Sigma)$ to be the agent's maximal payoff in this problem. It is easy to see that $U$ is non-negative, increasing in $y$ and decreasing in $c$ (we refer to weak monotonicity, unless otherwise specified). In addition, just as in \cite{FudenbergStrackStrzalecki}, the \emph{stopping boundary} at time $t = 0$ is symmetric and given by
\[
k^*(c, \Sigma) = \min\{x > 0:~ U(-x, c, \Sigma ) = 0\}.
\]
What this means is that if the prior expectation satisfies $\vert Y_0 \vert \geq k^*(c, \Sigma)$, then the agent optimally stops immediately and chooses good 1 or good 2 depending on whether $Y_0$ is positive or negative. Whereas if $\vert Y_0 \vert < k^*(c, \Sigma)$, then the optimal stopping time $\tau$ is strictly positive. 

To study how the agent's choice accuracy changes over time, we need to also consider the stopping boundaries at later times $t > 0$. For this we let $\Sigma_t$ be the posterior covariance matrix of $\theta$ at time $t$, given the prior covariance matrix $\Sigma$ and given the optimal attention strategy. Then the stopping boundary at time $t$ is simply $k^*(c, \Sigma_t)$ (this implicitly uses the fact that starting from the prior $\Sigma_t$, the posterior at time $s$ would be $\Sigma_{t+s}$). 

The next lemma (essentially Theorem 2 in \cite{FudenbergStrackStrzalecki}) characterizes a necessary and sufficient condition for choice accuracy to decrease over time, in terms of these stopping boundaries:

\begin{lemma}\label{lemma:binary choice accuracy boundary}
Given $c$ and $\Sigma$, and suppose $\vert Y_0 \vert < k^*(c, \Sigma)$ (so that the agent does not immediately stop). Let $p_t$ be the conditional probability that good 1 is better than good 2 when the agent stops at time $t$ and chooses good 1 (and vice verse, by symmetry). Then $p_t$ decreases in $t$ if and only if $\frac{k^*(c, \Sigma_t)}{\sigma_t}$ decreases in $t$. 
\end{lemma}

\begin{proof}
If the agent stops at time $t$ and chooses good 1, then $Y_t = k^*(c, \Sigma_t)$ (and $t$ is the earliest time this happens). So by definition, the posterior belief of $\omega$ is normal with mean $k^*(c, \Sigma_t)$ and standard deviation $\sigma_t$. Thus, the conditional probability that $\omega > 0$ (i.e., good 1 is better) is the normal c.d.f.\ evaluated at $\frac{k^*(c, \Sigma_t)}{\sigma_t}$. This yields the result. 
\end{proof}

Note that $\sigma_t^2$, being the posterior variance of $\theta_1 + \theta_2$, is just the sum of all the entries in the matrix $\Sigma_t$. Thus the condition in Lemma \ref{lemma:binary choice accuracy boundary} is ultimately about how $k^*(c, \Sigma)$ varies with $\Sigma$. The next section studies this change in detail.

\subsubsection{Effect of $\Sigma$ on Stopping Boundary}

We will say two covariance matrices $\widetilde{\Sigma}$ and $\widehat{\Sigma}$ \emph{induce the same prior uncertainty}, if the prior variances $\widetilde{\sigma}_0^2 = \widetilde{\Sigma}_{11} + 2\widetilde{\Sigma}_{12} + \widetilde{\Sigma}_{22}$ and $\widehat{\sigma}_0^2 = \widehat{\Sigma}_{11} + 2 \widehat{\Sigma}_{12} + \widehat{\Sigma}_{22}$ of $\theta_1 + \theta_2$ are equal. 

The result below provides a sufficient condition for stopping boundaries under two different prior covariance matrices to be comparable: 
\begin{lemma}\label{lemma:binary choice general change}
Let $\widetilde{\Sigma}$ and $\widehat{\Sigma}$ be two covariance matrices that induce the same prior uncertainty. Suppose further that the following two conditions holds:
\begin{enumerate}
    \item[(a)] $\widetilde{\Sigma}_{11} - \widetilde{\Sigma}_{22} \geq \widehat{\Sigma}_{11} - \widehat{\Sigma}_{22} \geq 0$;
    \item[(b)] $(\widetilde{\sigma}_{\widetilde{t}_1^*})^2 = \frac{2 \det (\widetilde{\Sigma})}{\widetilde{\Sigma}_{11} - \widetilde{\Sigma}_{12}} \leq \frac{2 \det (\widehat{\Sigma})}{\widehat{\Sigma}_{11} - \widehat{\Sigma}_{12}} = (\widehat{\sigma}_{\widehat{t}_1^*})^2 $.
\end{enumerate}
Then the posterior variances satisfy $\widetilde{\sigma}_t^2 \leq \widehat{\sigma}_t^2$ for all $t \geq 0$. Consequently $k^*(c, \widetilde{\Sigma}) \geq k^*(c, \widehat{\Sigma})$.
\end{lemma}

Part (a) says that there is greater asymmetry in the agent's uncertainty about the two attributes in the prior covariance matrix $\widetilde{\Sigma}$ compared to $\widehat{\Sigma}$. Part (b) says that the agent's uncertainty about $\theta_1+\theta_2$ at the optimal switchpoint $\widetilde{t}^*_1$ given prior $\widetilde{\Sigma}$ is lower than the agent's uncertainty about $\theta_1+\theta_2$ at the optimal switchpoint $\hat{t}^*_1$ given prior $\widehat{\Sigma}$, i.e., the agent has learned more about $\theta_1+\theta_2$ in (the endogenous) Stage 1 starting from $\widetilde{\Sigma}$. The lemma says that these conditions imply that the agent's uncertainty about $\theta_1+\theta_2$ is lower at every moment of time starting from prior $\widetilde{\Sigma}$, i.e., the agent learns faster.

\begin{proof}
Given any prior covariance matrix $\Sigma$ and resulting path of posterior variances $\sigma_t$, we can define for each $v \in [0, \sigma_0^2)$ the hitting time $T(v)$ such that $\sigma_{T(v)}^2 = \sigma_0^2 - v$. $T(v)$ is well-defined because $\sigma_t^2$ decreases strictly and continuously in $t$. This monotonicity also implies that the comparison $\widetilde{\sigma}_t^2 \leq \widehat{\sigma}_t^2$ for each $t$ is equivalent to $\widetilde{T}(v) \leq \widehat{T}(v)$ for each $v < \widetilde{\sigma}_0^2 = \widehat{\sigma}_0^2$. Thus in what follows we study the properties of $T(v)$. 

From Lemma \ref{lemma:binary choice variances}, it is not difficult to derive the following formula for $T(v)$: 
\[
T(v) = \begin{cases} \frac{v}{(\Sigma_{11}+\Sigma_{12})^2 - \Sigma_{11}v} &\text{if } v \in [0, v^*]; \\
\frac{4}{\sigma_0^2 - v} - \frac{\Sigma_{11} + \Sigma_{22} - 2\Sigma_{12}}{\det(\Sigma)} &\text{if } v \in [v^*, \sigma_0^2).
\end{cases}
\]
Above, the switchpoint $v^*$ is given by
\[
v^* = \sigma_0^2 - \sigma_{t_1^*}^2 = (\Sigma_{11} + \Sigma_{22} + 2\Sigma_{12}) - \frac{2 \det(\Sigma)}{\Sigma_{11} - \Sigma_{12}} = \frac{(\Sigma_{11}-\Sigma_{22})(\Sigma_{11}+\Sigma_{12})}{\Sigma_{11} - \Sigma_{12}} .
\]
Not surprisingly, at $v = v^*$ either formula for $T(v)$ yields the time $t_1^* = \frac{\Sigma_{11}-\Sigma_{22}}{\det(\Sigma)}$.

We now compute the (right) derivative of $T(v)$:
\begin{equation}\label{eq:binary choice T'}
T'(v) = \begin{cases} \frac{1}{\left(\Sigma_{11}+\Sigma_{12} - \frac{\Sigma_{11}}{\Sigma_{11}+\Sigma_{12}} v\right)^2} &\text{if } v \in [0, v^*); \\
\frac{4}{(\sigma_0^2-v)^2} &\text{if } v \in [v^*, \sigma_0^2).
\end{cases}
\end{equation}
This time perhaps more surprisingly, both formulae for $T'(v)$ yield the same value at $v = v^*$. Moreover, it can be checked that $T'(v) \leq \frac{4}{(\sigma_0^2-v)^2}$ for all $v$.

We claim that under the two stated conditions on $\widetilde{\Sigma}$ and $\widehat{\Sigma}$, it holds that $\widetilde{T}'(v) \leq \widehat{T}'(v)$ for every $v$. Since $\widetilde{T}(0) = \widehat{T}(0) = 0$, this would imply the desired comparison $\widetilde{T}(v) \leq \widehat{T}(v)$. To compare those derivatives, note that given $\widetilde{\Sigma}_{11} + \widetilde{\Sigma}_{22} + 2 \widetilde{\Sigma}_{12} = \widehat{\Sigma}_{11} + \widehat{\Sigma}_{22} + 2 \widehat{\Sigma}_{12}$, the assumption $\widetilde{\Sigma}_{11} - \widetilde{\Sigma}_{22} \geq \widehat{\Sigma}_{11} - \widehat{\Sigma}_{22}$ is equivalent to $\widetilde{\Sigma}_{11} + \widetilde{\Sigma}_{12} \geq \widehat{\Sigma}_{11} + \widehat{\Sigma}_{12}$. Thus $\widetilde{T}'(0) \leq \widehat{T}'(0)$ holds. Moreover, the second assumption in the lemma translates into $\widetilde{v}^* \geq 
\widehat{v}^*$. Below we show these are sufficient to imply $\widetilde{T}'(v) \leq \widehat{T}'(v)$. 

Indeed, for $v \geq \widehat{v}^*$, we deduce from \eqref{eq:binary choice T'} that 
\[
\widetilde{T}'(v) \leq \frac{4}{(\sigma_0^2-v)^2} = \widehat{T}'(v). 
\]
On the other hand, for $v \leq \widehat{v}^*$, $T'(v)$ is given by the first term in \eqref{eq:binary choice T'} for both $\widetilde{\Sigma}$ and $\widehat{\Sigma}$. Thus the comparison between $\widetilde{T}'(v)$ and $\widehat{T}'(v)$ reduces to a comparison between two linear functions of $v$: We want to show $\Sigma_{11}+\Sigma_{12} - \frac{\Sigma_{11}}{\Sigma_{11}+\Sigma_{12}} v$ is larger when $\Sigma = \widetilde{\Sigma}$ than when $\Sigma = \widehat{\Sigma}$. We already know this holds at $v = 0$ and $v = \widetilde{v}^*$, so by linearity, it also holds at any $v$ in between, completing the proof. 

Hence we have shown that $\widetilde{\sigma}_t^2 \leq \widehat{\sigma}_t^2$ for every $t$. It remains to show that the comparison of posterior variances implies the comparison of stopping boundaries. For this we observe that lower posterior variances under $\widetilde{\Sigma}$ imply that the value function $U(y, c, \widetilde{\Sigma})$ is weakly larger than $U(y, c, \widehat{\Sigma})$ for any cost $c$ and any prior expectation $y$. This follows from the same time-change argument as in the proof of Lemma \ref{lemma:Greenshtein}, and the idea is simply that any stopping time under prior $\widehat{\Sigma}$ can be replicated under prior $\widetilde{\Sigma}$ at an earlier stopping time. Thus 
\[
0 = U(-k^*(c, \widetilde{\Sigma}), c, \widetilde{\Sigma}) \geq U(-k^*(c, \widetilde{\Sigma}), c, \widehat{\Sigma}).
\]
In fact we have equality since $U$ is non-negative. Hence $k^*(c, \widehat{\Sigma})$, being the smallest $x$ such that $U(-x, c, \widehat{\Sigma}) = 0$, must be smaller than $k^*(c, \widetilde{\Sigma})$.
\end{proof}

Two useful corollaries follow from Lemma \ref{lemma:binary choice general change} (the proofs are immediate and thus omitted): 
\begin{lemma}[Effect of correlation]\label{lemma:binary choice correlation change}
Let $\widetilde{\Sigma}$ and $\widehat{\Sigma}$ be two covariance matrices that induce the same prior uncertainty. Suppose further that $\widetilde{\Sigma}_{11} = \widetilde{\Sigma}_{22}$ and $\widehat{\Sigma}_{11} = \widehat{\Sigma}_{22}$ (symmetric priors). Then $k^*(c, \widetilde{\Sigma}) = k^*(c, \widehat{\Sigma})$. 
\end{lemma}

\begin{lemma}[Effect of asymmetry]\label{lemma:binary choice asymmetry change}
Let $\widetilde{\Sigma}$ and $\widehat{\Sigma}$ be two covariance matrices that induce the same prior uncertainty. Suppose further that $\widetilde{\Sigma}_{11} > \widetilde{\Sigma}_{22}$ while $\widehat{\Sigma}_{11} = \widehat{\Sigma}_{22}$ (asymmetric versus symmetric). Then $k^*(c, \widetilde{\Sigma}) \geq k^*(c, \widehat{\Sigma})$. 
\end{lemma}

While the above results hold fixed the prior variance $\sigma_0^2$, we also need a result that considers a change in overall prior uncertainty, and characterizes its effect on the stopping boundary. 

\begin{lemma}[Effect of scaling $\Sigma$]\label{lemma:binary choice uncertainty change}
For any $c, \Sigma$ and any $\lambda \in (0, 1)$, it holds that $\lambda k^*(c, \Sigma) \geq k^*(c, \lambda^2 \Sigma)$.
\end{lemma}

\begin{proof}
We will show that for any $\lambda > 0$, 
\begin{equation}\label{eq:binary choice k* identity}
k^*(c, \lambda^2 \Sigma) = \lambda k^*(c\lambda^{-3}, \Sigma). 
\end{equation}
This implies the lemma because the cost $c\lambda^{-3}$ is higher than $c$ whenever $\lambda < 1$, which decreases the value function and thus also decreases the stopping boundary, resulting in $k^*(c\lambda^{-3}, \Sigma) \leq k^*(c, \Sigma)$. 

We note that the identity \eqref{eq:binary choice k* identity} is a direct generalization of Equation (A6) in \cite{FudenbergStrackStrzalecki}. Nonetheless, we provide a proof below for completeness. The key insight is that we can identify the belief processes under prior $\Sigma$ and under $\lambda^2\Sigma$ via a scaling of time and space. Specifically, let $Y_t$ denote the belief process under $\Sigma$ as given by \eqref{eq:binary choice belief process}, with $\sigma_t^2(\Sigma)$ denoting the posterior variance at time $t$ (under the optimal path). Similarly define $Z_t$ for the belief process starting from the prior $\lambda \Sigma$, with  $\sigma_t^2(\lambda^2 \Sigma)$ denoting the posterior variance at time $t$. From Lemma \ref{lemma:binary choice variances}, we have the relation
\[
\sigma_t^2(\lambda^2 \Sigma) = \lambda^2 \cdot \sigma_{\lambda^2t}^2(\Sigma). 
\]
Assuming $Z_0 = \lambda Y_0$, then the process $Z_t$ has the same distribution as the process $\lambda \cdot Y_{\lambda^2 t}$. Intuitively, receiving standard normal signals about $\lambda \omega$ for one unit of time is equivalent to receiving standard normal signals about $\omega$ for $\lambda^2$ units of time. 

Now that we identify $Z_t$ with $\lambda \cdot Y_{\lambda^2 t}$, we can rewrite the stopping problem with respect to $Z_t$ in terms of the $Y$ process instead. Specifically, the agent's problem is 
\[
\max_{\tau} \mathbb{E}[\max\{\lambda Y_{\lambda^2\tau}, 0\} - c \tau] = \lambda \max_{\tau'} \mathbb{E}[\max\{Y_{\tau'}, 0\} - \frac{c}{\lambda^3} \tau'],
\]
with $\tau' = \lambda^2\tau$. Thus it is as if the agent chooses an optimal stopping time with respect to $Y_t$, but with transformed marginal cost $\frac{c}{\lambda^3}$. The agent should stop at time 0 in this problem if and only if  $\vert Y_0 \vert \geq k^*(\frac{c}{\lambda^3}, \Sigma)$, which is equivalent to $\vert Z_0 \vert \geq \lambda k^*(\frac{c}{\lambda^3}, \Sigma)$. Thus \eqref{eq:binary choice k* identity} holds. 
\end{proof}

\subsubsection{Main Proof of Proposition \ref{prop:binary choice FSS}} 

Given Lemma \ref{lemma:binary choice accuracy boundary}, we just need to show that for any times $t < t'$, 
$\frac{k^*(c, \Sigma_t)}{\sigma_t} \geq \frac{k^*(c, \Sigma_{t'})}{\sigma_{t'}}$. Define $\lambda \in (0,1)$ such that $\sigma_{t'} = \lambda \sigma_{t}$. Then the inequality becomes $\lambda k^*(c, \Sigma_t) \geq k^*(c, \Sigma_{t'})$. From Lemma \ref{lemma:binary choice uncertainty change} we have $\lambda k^*(c, \Sigma_t) \geq k^*(c, \lambda^2 \Sigma_t)$, so it is sufficient to show 
\begin{equation}
   k^*(c, \lambda^2 \Sigma_t) \geq k^*(c, \Sigma_{t'}). \tag{to be shown}
\end{equation}
Note that, by the definition of $\lambda$, the matrices $\lambda^2 \Sigma_t$ and $\Sigma_{t'}$ induce the same uncertainty about $\theta_1 + \theta_2$. There are three cases to consider: 

\smallskip

\textbf{Case 1: $t_1^* \leq t < t'$.} In this case $\Sigma_t$ and $\Sigma_{t'}$ are posterior covariance matrices in Stage 2, so they induce symmetric uncertainty about $\theta_1$ and $\theta_2$. Lemma \ref{lemma:binary choice correlation change} thus applies to the symmetric priors $\lambda^2 \Sigma_t$ and $\Sigma_{t'}$, and yields $k^*(c, \lambda^2 \Sigma_t) = k^*(c, \Sigma_{t'})$. Intuitively, the belief process in Stage 2 is the same as in \cite{FudenbergStrackStrzalecki} since correlation does not matter for symmetric priors. Thus the result in this case follows from the result in that paper. 

\smallskip

\textbf{Case 2: $t < t_1^* \leq t'$.} Here Lemma \ref{lemma:binary choice correlation change} no longer applies. We instead apply Lemma \ref{lemma:binary choice asymmetry change} to deduce $k^*(c, \lambda^2 \Sigma_t) \geq k^*(c, \Sigma_{t'})$, since $\lambda^2\Sigma_t$ is an asymmetric prior while $\Sigma_{t'}$ is a symmetric prior. This proves the result, and as discussed, the intuition is that asymmetry increases the stopping boundary relative to symmetric priors. 

\smallskip

\textbf{Case 3: $t < t' < t_1^*$.} To prove the key comparison $k^*(c, \lambda^2 \Sigma_t) \geq k^*(c, \Sigma_{t'})$, we now need to invoke the more general Lemma \ref{lemma:binary choice general change} since $\lambda^2\Sigma_t$ and $\Sigma_{t'}$ are both asymmetric. Thus, we have to check that $\widetilde{\Sigma} = \lambda^2\Sigma_t$ and $\widehat{\Sigma} = \Sigma_{t'}$ satisfy the two conditions stated in Lemma \ref{lemma:binary choice general change}. 

To do this, we let $\Sigma_t = \left(\begin{array}{cc}
p & r \\
r & q \end{array}\right)$ with $p > q > 0$ and $r^2 < pq$. The posterior covariance matrix $\Sigma_{t'}$ at the later time $t'$ can be calculated from the ``prior" covariancee matrix $\Sigma_{t}$, after focusing on $\theta_1$ for $t' - t$ units of time. Thus
\[
\Sigma_{t'} = \left( \Sigma_t^{-1} + \left(\begin{array}{cc}
t'-t & 0 \\
0 & 0 \end{array}\right) \right)^{-1} = \left(\frac{1}{pq-r^2}\left(\begin{array}{cc}
q + (t'-t)(pq-r^2) & -r \\
-r & p \end{array}\right)  \right)^{-1}.
\]
Let $q' = q + (t'-t)(pq-r^2)$ with $q < q' < p$ (the latter inequality holds because $t' < t_1^*$). Then the above simplifies to 
\[
\Sigma_{t'} = \frac{pq - r^2}{pq' - r^2} \left(\begin{array}{cc}
p & r \\
r & q' \end{array}\right)
\]
The scaling factor $\lambda$ is thus given by
\begin{equation}\label{eq:binary choice lambda}
\lambda^2 (p+q+2r) = \frac{pq - r^2}{pq' - r^2} (p+q'+2r). 
\end{equation}

We now check the first condition in Lemma \ref{lemma:binary choice general change}, which for $\widetilde{\Sigma} = \lambda^2 \Sigma_t$ and $\widehat{\Sigma} = \Sigma_{t'}$ becomes $\lambda^2(p-q) \geq \frac{pq - r^2}{pq' - r^2}(p - q')$. Using \eqref{eq:binary choice lambda} to eliminate $\lambda$, it suffices to show that 
\[
\frac{p-q}{p+q+2r} \geq \frac{p-q'}{p+q'+2r}. 
\]
This inequality holds simply because $q' > q$. 

We then turn to the second condition in Lemma \ref{lemma:binary choice accuracy boundary}, which in the current setting becomes $2\lambda^2 (\frac{pq-r^2}{p-r}) \leq 2(\frac{pq-r^2}{pq'-r^2}) \cdot (\frac{pq'-r^2}{p-r})$. This simplifies to $\lambda^2 \leq 1$, which clearly holds. Intuitively, since $\Sigma_t$ and $\Sigma_{t'}$ are both posterior beliefs following the prior $\Sigma$, they ``become symmetric" at the same posterior belief $\Sigma_{t_1^*}$. Thus the second condition in Lemma \ref{lemma:binary choice general change} holds with equality when we consider $\Sigma_t$ versus $\Sigma_{t'}$. It follows that when comparing $\lambda^2 \Sigma_t$ and $\Sigma_{t'}$, the former prior belief leads to lower uncertainty when entering Stage 2. 

Hence Lemma \ref{lemma:binary choice general change} applies, and we again have 
\[
\lambda k^*(c, \Sigma_t) \geq k^*(c, \lambda^2 \Sigma_t) \geq k^*(c, \Sigma_{t'}). 
\]
This proves $\frac{k^*(c, \Sigma_t)}{\sigma_t} \geq \frac{k^*(c, \Sigma_{t'})}{\sigma_{t'}}$ and the proposition.

\subsubsection{Generalization to Unequal Learning Speeds} \label{appx:binary choice informativeness} 

In this section we show that the conclusion of Proposition \ref{prop:binary choice FSS} further generalizes to situations where the information about the two unknown payoffs arrives with different levels of precision. Formally, fix $\alpha_1, \alpha_2 > 0$ and suppose a unit of time devoted to learning about the payoff $v_i$ produces the signal $v_i + \mathcal{N}(0,\alpha_i^2)$. Any prior covariance matrix over $(v_1, v_2)$ can be written as $\left(\begin{array}{cc}
\alpha_1^2\Sigma_{11} & -\alpha_1\alpha_2\Sigma_{12} \\
-\alpha_1\alpha_2\Sigma_{21} & \alpha_2^2\Sigma_{22} \end{array}\right)$ where $\Sigma = \left(\begin{array}{cc}
\Sigma_{11} & \Sigma_{12} \\
\Sigma_{21} & \Sigma_{22} \end{array}\right)$ is also a positive definite matrix. This formulation of the problem will make the results below easier to state. 

To transform this problem into our main model, let $\theta_1 = v_1 / \alpha_1$ and $\theta_2 = - v_2 / \alpha_2$, so that each unit of time devoted to $\theta_i$ produces a standard normal signal about it. Moreover, the prior covariance matrix over $(\theta_1, \theta_2)$ is simply $\Sigma$, and the payoff-relevant state is $\omega = v_1 - v_2 = \alpha_1 \theta_1 + \alpha_2 \theta_2$. Throughout this section, we assume $\Sigma$ and $\alpha$ satisfy Assumption \ref{assmp:K=2}, so that we can apply Theorem \ref{thm:K=2} to characterize optimal attention allocation. 

\begin{corollary}
Suppose Assumption \ref{assmp:K=2} holds and $\alpha_1\Sigma_{11} + \alpha_2\Sigma_{12} \geq \alpha_1\Sigma_{12} + \alpha_2\Sigma_{22}$. The agent's optimal information acquisition strategy $(\beta_1(t),\beta_2(t))$ in this generalized binary choice problem consists of two stages: 
\begin{itemize}[noitemsep]
\item \textbf{Stage 1:} At all times 
\[
t < t_1^* = \frac{\alpha_1\Sigma_{11} + \alpha_2\Sigma_{12} - \alpha_1\Sigma_{21} - \alpha_2\Sigma_{22}}{\alpha_2\det(\Sigma)},
\]
the agent optimally allocates all attention to $\theta_1$.
\item \textbf{Stage 2:} At times $t \geq t_1^*$, the agent optimally uses the constant mixture $\left(\frac{\alpha_1}{\alpha_1+\alpha_2}, \frac{\alpha_2}{\alpha_1+\alpha_2} \right)$.
\end{itemize}
\end{corollary}

From this we can compute the agent's posterior variance of $\alpha_1 \theta_1 + \alpha_2 \theta_2$ at each time $t$. Generalizing Lemma \ref{lemma:binary choice variances}, we have 
\[
\sigma_t^2 = \begin{cases} \frac{\alpha_1^2\Sigma_{11} + \alpha_2^2 \Sigma_{22} + 2\alpha_1\alpha_2 \Sigma_{12} + \alpha_2^2 \det(\Sigma)t}{1 + \Sigma_{11} t} &\text{if } t \leq t_1^*; \\
\frac{(\alpha_1+\alpha_2)^2 \det(\Sigma)}{\Sigma_{11} + \Sigma_{22} - 2\Sigma_{12} + \det(\Sigma)t} &\text{if } t \geq t_1^*.
\end{cases}
\]
In particular, $\sigma_{t_1^*}^2 = \frac{(\alpha_1+\alpha_2)\alpha_2\det(\Sigma)}{\Sigma_{11}-\Sigma_{12}}$. We omit the detailed calculations.

Next, note that Lemma \ref{lemma:binary choice accuracy boundary} holds without change, so we just need to study how the stopping boundary $k^*(c, \Sigma)$ varies with $\Sigma$ in this more general setting. Naturally, we say that two prior covariance matrices $\widetilde{\Sigma}$ and $\widehat{\Sigma}$ \emph{induce the same prior uncertainty} if the prior variances $\alpha_1^2\widetilde{\Sigma}_{11} + \alpha_2^2 \widetilde{\Sigma}_{22} + 2\alpha_1\alpha_2 \widetilde{\Sigma}_{12}$ and $\alpha_1^2\widehat{\Sigma}_{11} + \alpha_2^2 \widehat{\Sigma}_{22} + 2\alpha_1\alpha_2 \widehat{\Sigma}_{12}$ are equal. The key Lemma \ref{lemma:binary choice general change} above is then generalized as follows:

\begin{lemma}
Let $\widetilde{\Sigma}$ and $\widehat{\Sigma}$ be two covariance matrices that induce the same prior uncertainty. Suppose further that the following two conditions holds:
\begin{enumerate}
    \item[(a)] $\alpha_1\widetilde{\Sigma}_{11} + \alpha_2\widetilde{\Sigma}_{12} - \alpha_1\widetilde{\Sigma}_{21} - \alpha_2\widetilde{\Sigma}_{22} \geq \alpha_1\widehat{\Sigma}_{11} + \alpha_2\widehat{\Sigma}_{12} - \alpha_1\widehat{\Sigma}_{21} - \alpha_2\widehat{\Sigma}_{22} \geq 0$;
    \item[(b)] $(\widetilde{\sigma}_{\widetilde{t}_1^*})^2 = \frac{(\alpha_1+\alpha_2)\alpha_2 \det (\widetilde{\Sigma})}{\widetilde{\Sigma}_{11} - \widetilde{\Sigma}_{12}} \leq \frac{(\alpha_1+\alpha_2)\alpha_2 \det (\widehat{\Sigma})}{\widehat{\Sigma}_{11} - \widehat{\Sigma}_{12}} = (\widehat{\sigma}_{\widehat{t}_1^*})^2 $.
\end{enumerate}
Then the posterior variances satisfy $\widetilde{\sigma}_t^2 \leq \widehat{\sigma}_t^2$ for all $t \geq 0$. Consequently $k^*(c, \widetilde{\Sigma}) \geq k^*(c, \widehat{\Sigma})$.
\end{lemma}
That is, with potentially unequal payoff weights $\alpha_1$ and $\alpha_2$, prior ``asymmetry" is not simply measured by the difference between $\Sigma_{11}$ and $\Sigma_{22}$. Rather, it is measured by the difference in initial marginal values $cov_1 = \alpha_1\Sigma_{11} + \alpha_2\Sigma_{12}$ and $cov_2 = \alpha_1\Sigma_{21} + \alpha_2\Sigma_{22}$. Condition (a) thus requires this asymmetry to be larger under $\widetilde{\Sigma}$ than under $\widehat{\Sigma}$. This turns out to be equivalent to $\widetilde{T}'(0) \leq \widehat{T}'(0)$, where the hitting time $T$ is the same as defined in the proof of Lemma \ref{lemma:binary choice general change}.\footnote{From the formulae for $\sigma_t^2$ we can compute $T(v)$ and $T'(v)$. Writing $v^*= \sigma_0^2 - \sigma_{t_1^*}^2 = \frac{(cov_1 - cov_2)cov_1}{\Sigma_{11}-\Sigma_{12}}$, we have
\begin{equation*}
T'(v) = \begin{cases} \frac{1}{\left(cov_1 - \frac{\Sigma_{11}}{cov_1} v\right)^2} &\text{if } v \in [0, v^*]; \\
\frac{(\alpha_1+\alpha_2)^2}{(\sigma_0^2-v)^2} &\text{if } v \in [v^*, \sigma_0^2).
\end{cases}
\end{equation*}
}
Together with condition (b), this implies $\widetilde{T}'(v) \leq \widehat{T}'(v)$ for all $v$ and thus $\widetilde{T}(v) \leq \widehat{T}(v)$, just as in that proof. 

Using the same notion of asymmetry, we obtain direct generalizations of Lemma \ref{lemma:binary choice correlation change} and Lemma \ref{lemma:binary choice asymmetry change} as well. Finally, Lemma \ref{lemma:binary choice uncertainty change} continues to hold since its proof does not depend on payoff weights. These lemmata allow us to replicate the proof of Proposition \ref{prop:binary choice FSS} with only minor modifications. Thus, to summarize, we have the following result: 
\begin{proposition}
Consider the binary choice problem with general signal variances $\alpha_1^2$ and $\alpha_2^2$. Denote the agent's prior covariance matrix over the payoffs $(v_1, v_2)$ as $\left(\begin{array}{cc}
\alpha_1^2\Sigma_{11} & -\alpha_1\alpha_2\Sigma_{12} \\
-\alpha_1\alpha_2\Sigma_{21} & \alpha_2^2\Sigma_{22} \end{array}\right)$. Then whenever $\Sigma$ and $\alpha$ satisfy Assumption \ref{assmp:K=2}, the agent's choice accuracy in this problem is (weakly) higher at earlier stopping times. 
\end{proposition}
We reiterate that Assumption \ref{assmp:K=2} is guaranteed if $\alpha_1 = \alpha_2$ (as we assumed previously), or if $\Sigma_{12} \geq 0$ (i.e., the unknown payoffs $v_1$ and $v_2$ are negatively correlated).

\subsection{Proof of Proposition \ref{prop:eqm}}

\subsubsection{Reduction to the Main Model} 

We first show it is without loss to consider $\sigma_b = 1$. Suppose the result holds for $\sigma_b = 1$, then for a general $\sigma_b$, we can write $b = \sigma_b \cdot \hat{b}$ and $\omega = \sigma_b \cdot \hat{\omega}$ for some random variables $\hat{b}$ and $\hat{\omega}$, where $\hat{b}$ has unit variance. Note that the two signals $\omega + \phi_1 b + \mathcal{N}(0, \zeta_1^2)$ and $\omega - \phi_2 b + \mathcal{N}(0, \zeta_2^2)$ are informationally equivalent to $\hat{\omega} + \phi_1 \hat{b} + \mathcal{N}\left(0, (\frac{\zeta_1}{\sigma_b})^2\right)$ and $\hat{\omega} - \phi_2 \hat{b} + \mathcal{N}\left(0, (\frac{\zeta_2}{\sigma_b})^2\right)$. We have thus transformed the general problem into one with payoff-relevant state $\hat{\omega}$ and unknown benefit $\hat{b}$, where $\hat{b}$ has unit variance. The result for this case then pins down equilibrium choices of $\phi_i$ and $\frac{\zeta_i}{\sigma_b}$, which then yield the equilibrium in the general case.

Hence, for the rest of the proof, we assume $\sigma_b = 1$. Define $\theta_1  = \frac{1}{\zeta_1}(\omega + \phi_1 b)$ and $\theta_2 = \frac{1}{\zeta_2}(\omega - \phi_2 b)$. Observe that $\omega + \phi_1 b + \mathcal{N}(0, \zeta_1^2)$ is informationally equivalent to $\frac{1}{\zeta_1}(\omega + \phi_1 b)+ \mathcal{N}(0, 1)$. Thus a unit of time spent on source $i$ produces a standard normal signal about the corresponding $\theta_i$, which returns our main model. The prior covariance matrix for $(\theta_1, \theta_2)$ is 
\[
\Sigma = \left(\begin{array}{cc}
\frac{\sigma_\omega^2 + \phi_1^2}{\zeta_1^2} & \frac{\sigma_\omega^2 - \phi_1 \phi_2}{\zeta_1\zeta_2} \\
\frac{\sigma_\omega^2 - \phi_1 \phi_2}{\zeta_1\zeta_2} & \frac{\sigma_\omega^2 + \phi_2^2}{\zeta_2^2} \end{array}\right)
\]
and the payoff-relevant state can be written as $\omega = \alpha_1 \theta_1 + \alpha_2 \theta_2$
with payoff weights $\alpha_1  = \zeta_1 \cdot \frac{\phi_2}{\phi_1 + \phi_2}$ and $\alpha_2 = \zeta_2 \cdot \frac{\phi_1}{\phi_1 + \phi_2}$.

Below we derive the reader's optimal attention allocation using our main results. Throughout we assume $\zeta_1 \leq \zeta_2$, so that source 1 is more precise. Under this assumption, we have 
$
cov_1 = \sigma_\omega^2 / \zeta_1 \geq \sigma_\omega^2 / \zeta_2 = cov_2,
$
which implies that source 1 is attended to in the first stage. Moreover, by Theorem \ref{thm:K=2} (which applies since $cov_1$ and $cov_2$ are both positive), the length of the first stage is 
\begin{equation}\label{eq:media_t1*}
t_1^* = \frac{cov_1 - cov_2}{\alpha_2 \det(\Sigma)} = \frac{\zeta_1(\zeta_2 - \zeta_1)}{\phi_1(\phi_1 + \phi_2)},
\end{equation}
where we used $\det(\Sigma) = \sigma_\omega^2 \left(\frac{\phi_1 + \phi_2}{\zeta_1 \zeta_2}\right)^2$.

Thus, the reader optimally attends only to source 1 until time $t_1^*$, and afterwards gives $\frac{\alpha_1}{\alpha_1 + \alpha_2}$ fraction of his attention to source 1. We can then write the two sources' payoffs as follows:
\begin{equation}\label{eq:media_payoffs}
\begin{split}
    U_1 & = \int_0^{t_1^*} re^{-rt}dt + \int_{t_1^*}^\infty re^{-rt} \left(\frac{\zeta_1\phi_2}{\zeta_1 \phi_2 + \zeta_2 \phi_1}\right) dt - \lambda(\kappa-\phi_1)^2 \\
    & = 1 - e^{-rt_1^*} \left(\frac{\zeta_2 \phi_1}{\zeta_1 \phi_2 + \zeta_2 \phi_1}\right) - \lambda(\kappa-\phi_1)^2; \\
    U_2 & = e^{-rt_1^*} \left(\frac{\zeta_2 \phi_1}{\zeta_1 \phi_2 + \zeta_2 \phi_1}\right) - \lambda(\kappa-\phi_2)^2.
\end{split}
\end{equation}
These payoffs define a stage game between the two sources, and our goal is to characterize its equilibrium. Our strategy below is to use first-order conditions to pin down what an equilibrium must be. We will then verify the equilibrium by checking all possible deviations.

\subsubsection{Solving for Equilibrium Precisions $\zeta_1^*$, $\zeta_2^*$}

We show here that the precision choices $\zeta_1$, $\zeta_2$ must be equal in any equilibrium. Suppose not, then small changes in $\zeta_1$ and $\zeta_2$ do not affect our standing assumption that $\zeta_1 \leq \zeta_2$. Thus we can take the first-order conditions for $\zeta_1$ and $\zeta_2$. Observe that 
$
\frac{\partial t_1^*}{\partial \zeta_1} = \frac{\zeta_2- 2\zeta_1}{\phi_1(\phi_1 + \phi_2)}.
$
So we can compute that
\begin{align*}
    \frac{\partial U_1}{\partial \zeta_1} = e^{-r t_1^*} \left(\frac{\zeta_2 \phi_1}{\zeta_1 \phi_2 + \zeta_2 \phi_1}\right) \left(r \frac{\zeta_2 - 2\zeta_1}{\phi_1(\phi_1 + \phi_2)} + \frac{\phi_2}{\zeta_1 \phi_2 + \zeta_2 \phi_1}\right).
\end{align*}
The FOC then requires that
\begin{equation}\label{eq:media_FOC_zeta1}
r\frac{2\zeta_1 - \zeta_2}{\phi_1(\phi_1 + \phi_2)} = \frac{\phi_2}{\zeta_1\phi_2 + \zeta_2 \phi_1}.
\end{equation}

Now consider the FOC for source 2. Observe that 
$\frac{\partial t_1^*}{\partial \zeta_2} = \frac{\zeta_1}{\phi_1(\phi_1 + \phi_2)}.$
So
\begin{align*}
    \frac{\partial U_2}{\partial \zeta_2} = e^{-r t_1^*} \left(\frac{\zeta_1 \phi_1}{\zeta_1 \phi_2 + \zeta_2 \phi_1}\right) \left(-r \frac{\zeta_2}{\phi_1(\phi_1 + \phi_2)} + \frac{\phi_2}{\zeta_1 \phi_2 + \zeta_2 \phi_1}\right).
\end{align*}
The FOC then requires that
\begin{equation}\label{eq:media_FOC_zeta2}
r\frac{\zeta_2}{\phi_1(\phi_1 + \phi_2)} = \frac{\phi_2}{\zeta_1\phi_2 + \zeta_2 \phi_1}.
\end{equation}
\eqref{eq:media_FOC_zeta1} and \eqref{eq:media_FOC_zeta2} together imply 
\[
r\frac{2\zeta_1 - \zeta_2}{\phi_1(\phi_1 + \phi_2)}  = r \frac{\zeta_2}{\phi_1(\phi_1 + \phi_2)},
\]
which simplifies to $\zeta_1 = \zeta_2$ and leads to a contradiction. 

Hence $\zeta_1 = \zeta_2$ must hold in equilibrium. In this case the first-order conditions derived above need not hold with equality, because the payoffs in \eqref{eq:media_payoffs} are derived under the assumption that $\zeta_1 \leq \zeta_2$, so that the same payoff expressions apply only to downward deviations of $\zeta_1$ and upward deviations of $\zeta_2$. Given this, the first-order conditions become inequalities $ \frac{\partial U_1}{\partial \zeta_1} \geq 0$ and $ \frac{\partial U_2}{\partial \zeta_2} \leq 0$ (evaluated at the equilibrium choices). These translate into the following inequality versions of \eqref{eq:media_FOC_zeta1} and \eqref{eq:media_FOC_zeta2}:
\begin{align*}
r\frac{2\zeta_1 - \zeta_2}{\phi_1(\phi_1 + \phi_2)} &\leq  \frac{\phi_2}{\zeta_1\phi_2 + \zeta_2 \phi_1}; \\
r\frac{\zeta_2}{\phi_1(\phi_1 + \phi_2)} &\geq \frac{\phi_2}{\zeta_1\phi_2 + \zeta_2 \phi_1}.
\end{align*}
Since we already know $\zeta_1 = \zeta_2$, the two inequalities above must both hold equal, and we further deduce that
\begin{equation}\label{eq:media_zeta and phi}
\zeta_1 = \zeta_2 = \sqrt{\frac{\phi_1 \phi_2}{r}}.
\end{equation}
Note also that given $\zeta_2 = \sqrt{\frac{\phi_1 \phi_2}{r}}$, choosing $\zeta_1$ to be any smaller number cannot be profitable for source 1. This is because as $\zeta_1$ decreases, the term $r \frac{\zeta_2 - 2\zeta_1}{\phi_1(\phi_1 + \phi_2)} + \frac{\phi_2}{\zeta_1 \phi_2 + \zeta_2 \phi_1}$ appearing in $\frac{\partial U_1}{\partial \zeta_1}$ increases and remains positive. So the choice $\zeta_1 = \sqrt{\frac{\phi_1 \phi_2}{r}}$ is robust to any downward deviation (in this variable). Similarly, given this $\zeta_1$, choosing any larger $\zeta_2$ is not profitable for source 2. By symmetry, source 1 (respectively source 2) also cannot profit from upward (respectively downward) deviations in precision.

\subsubsection{Solving for Equilibrium Biases $\phi_1^*$, $\phi_2^*$} 

We now fix precision choices and characterize equilibrium levels of bias. Since $\zeta_1 = \zeta_2$ in equilibrium, we have $t_1^* = 0$, meaning that there is no stage 1. Hence the two sources' payoffs simplify to 
\begin{equation}\label{eq:media_payoffs_phi}
\begin{split}
U_1 &=  \frac{\phi_2}{\phi_1 + \phi_2} - \lambda(\kappa - \phi_1)^2; \\
U_2 &=  \frac{\phi_1}{\phi_1 + \phi_2} - \lambda(\kappa - \phi_2)^2.
\end{split}
\end{equation}
In this smaller game, we will show that there is a (unique) pure strategy equilibrium if and only if $\lambda \kappa^2 \geq \frac{9}{16}$, in which case the equilibrium involves $\phi_1 = \phi_2 = \frac{1}{2}\left(\kappa + \sqrt{\kappa^2 - \frac{1}{2\lambda}}\right)$.

The first-order conditions $\frac{\partial U_i}{\partial \phi_i} = 0$ give
\begin{equation}\label{eq:media_FOC_phi}
2\lambda(\kappa-\phi_1) =  \frac{\phi_2}{(\phi_1 + \phi_2)^2}; \quad \quad \quad \quad 
2\lambda(\kappa-\phi_2) =  \frac{\phi_1}{(\phi_1 + \phi_2)^2}.
\end{equation}
In addition, the second-order conditions $\frac{\partial^2 U_i}{\partial \phi_i^2} \leq 0$ give
\begin{equation}\label{eq:media_SOC_phi}
2\lambda \geq \frac{2\phi_2}{(\phi_1 + \phi_2)^3}; \quad \quad \quad \quad \quad \quad \quad
2\lambda \geq  \frac{2\phi_1}{(\phi_1 + \phi_2)^3}.
\end{equation}
Comparing each equality in \eqref{eq:media_FOC_phi} with the corresponding one in \eqref{eq:media_SOC_phi} yields 
\begin{equation}\label{eq:media_phi1 and phi2}
2(\kappa-\phi_1) \leq \phi_1 + \phi_2; \quad \quad \quad \quad \quad
2(\kappa-\phi_2) \leq \phi_1 + \phi_2. 
\end{equation}
Moreover, multiplying the two equalities in \eqref{eq:media_FOC_phi} yields $(\kappa-\phi_1)\phi_1 = (\kappa-\phi_2)\phi_2$. So either $\phi_1 = \phi_2$ or $\phi_1 + \phi_2 = \kappa$. In the former case \eqref{eq:media_phi1 and phi2} implies $\phi_1 = \phi_2 \geq \frac{\kappa}{2}$. In the latter case \eqref{eq:media_phi1 and phi2} implies $\phi_1 = \phi_2 = \frac{\kappa}{2}$. Thus we always have $\phi_1 = \phi_2 \geq \frac{\kappa}{2}$. 

Now since $\phi_1 = \phi_2$, we can solve from \eqref{eq:media_FOC_phi} that $\phi_1$ satisfies $4 \phi_1 (\kappa - \phi_1) = \frac{1}{2\lambda}$, which is equivalent to 
\[
(2\phi_1 - \kappa)^2 = \kappa^2 - \frac{1}{2\lambda}. 
\]
Using $\phi_1 \geq \frac{\kappa}{2}$, we deduce that the only possible equilibrium is $\phi_1 = \phi_2 = \frac{1}{2}\left(\kappa + \sqrt{\kappa^2 - \frac{1}{2\lambda}}\right)$, which we denote by $\phi^*$. Clearly, for this solution to be a real number, a necessary condition is $\lambda \kappa^2 \geq \frac{1}{2}$.

Below we show these choices are an equilibrium in this smaller game (involving only $\phi_1, \phi_2$) if and only if $\lambda \kappa^2 \geq \frac{9}{16}$. Indeed, with these choices source 1's payoff is 
$\frac{1}{2} - \lambda(\kappa - \phi^*)^2$. This must be higher than choosing $\phi_1$ close to $0$, which would yield a payoff close to $1 - \lambda \kappa^2$. Thus we can deduce the inequality $\phi^* (2\kappa - \phi^*) \geq \frac{1}{2\lambda}$. But recall that $\phi^*$ satisfies the equation $4\phi^* (\kappa - \phi^*) = \frac{1}{2\lambda}$. So we deduce that $2\kappa - \phi^* \geq 4(\kappa - \phi^*)$, or $\phi^* \geq \frac{2}{3} \kappa$. It follows that $\frac{1}{2\lambda} = 4\phi^* (\kappa - \phi^*) \leq \frac{8}{9} \kappa^2$, which implies $\lambda \kappa^2 \geq \frac{9}{16}$ as a necessary condition.

Conversely, suppose $\lambda \kappa^2 \geq \frac{9}{16}$ holds. We need to show that 
\[
f(\phi_1) = \frac{\phi^*}{\phi_1 + \phi^*} - \lambda(\kappa - \phi_1)^2
\]
is maximized at $\phi_1 = \phi^*$. Since we have derived the solution via the first- and second-order conditions, it holds that $f'(\phi^*) = 0$ and $f''(\phi^*) \leq 0$. Moreover, note that $f'''(\phi_1) = \frac{-6\phi^*}{(\phi_1 + \phi^*)^4} < 0$, so $f'(\phi_1)$ is a strictly concave function. The fact that $f'(\phi^*) = 0$ and $f''(\phi^*) \leq 0$ thus imply that $f'(\phi_1) \leq 0$ for all $\phi_1 \geq \phi^*$. Hence $f(\phi_1)$ is decreasing for $\phi_1 \geq \phi^*$. 

On the other hand, since $f'(\phi_1)$ is concave and $f'(\phi^*) = 0$, there are two possibilities for the behavior of $f'$ on the interval $[0, \phi^*]$: Either $f'(0) \geq 0$ and thus $f'$ is non-negative on this whole interval, or $f'(0) < 0$ and $f'$ crosses zero exactly once from below. This means $f$ is either increasing on $[0, \phi^*]$, or first decreasing and then increasing. Hence the maximum of $f$ on this interval must occur at the extreme points. When $\lambda \kappa^2 \geq \frac{9}{16}$, we have $\phi^* = \frac{1}{2}\left(\kappa + \sqrt{\kappa^2 - \frac{1}{2\lambda}}\right) \geq \frac{2}{3} \kappa$. Thus 
\[
f(\phi^*) = \frac{1}{2} - \lambda(\kappa - \phi^*)^2 \geq \frac{1}{2} - \frac{1}{9} \lambda \kappa^2 \geq 1 - \lambda \kappa^2 = f(0).
\]
This shows that the function $f$ is maximized at $\phi^*$ whenever $\lambda \kappa^2 \geq \frac{9}{16}$.

\subsubsection{Verifying the Equilibrium} 

Summarizing the above analysis, we have shown that the only possible pure strategy equilibrium is $\phi_1 = \phi_2 = \phi^*$ and $\zeta_1 = \zeta_2 = \frac{\phi^*}{\sqrt{r}}$, where the latter follows from \eqref{eq:media_zeta and phi}. We have also shown that if $\lambda \kappa^2 \geq \frac{9}{16}$, then given source 2's equilibrium choices, source 1 does not have a profitable deviation in $\zeta_1$ or in $\phi_1$ alone. However, without further assumptions, it is possible for source 1 to profit from choosing $\zeta_1$ and $\phi_1$ both away from the target equilibrium. To illustrate, let $\kappa = 1$, $\lambda = \frac{9}{16}$ and $r = 1$. Then source 2's equilibrium choices are $\phi_2 = \zeta_2 = \phi^* = \frac{2}{3}$. By also choosing $\phi_1 = \zeta_1 = \frac{2}{3}$, source 1 obtains payoff $\frac{1}{2} - \lambda(\kappa - \phi^*)^2 = \frac{7}{16} = 0.4375$. Suppose instead source 1 chooses $\phi_1 = \frac{1}{6}$ and $\zeta_1 = \frac{1}{3}$. Then $t_1^* = \frac{\zeta_1(\zeta_2-\zeta_1)}{\phi_1(\phi_1+\phi_2)} = \frac{4}{5}$, and source 1 receives long-run attention $\frac{\zeta_1\phi_2}{\zeta_1\phi_2+\zeta_2\phi_1} = \frac{2}{3}$. In this case source 1's payoff is higher:
\[
1 - \frac{1}{3} e^{-\frac{4}{5}} - \frac{9}{16}(1-\frac{1}{6})^2 \approx 0.4596.
\]
So the target equilibrium would not be an equilibrium when double deviations are considered. 


However, we will show that when $\lambda \kappa^2 \geq 1.6$, then given $\phi_2 = \phi^*$ and $\zeta_2 = \frac{\phi^*}{\sqrt{r}}$, source 1 cannot profitably deviate to any $\phi_1$ and $\zeta_1 \leq \frac{\phi^*}{\sqrt{r}}$. And given $\phi_1 = \phi^*$ and $\zeta_1 = \frac{\phi^*}{\sqrt{r}}$, source 2 cannot profitably deviate to any $\phi_2$ and $\zeta_2 \geq \frac{\phi^*}{\sqrt{r}}$. Thanks to the symmetry of the target equilibrium, verifying these will be sufficient to show it is indeed an equilibrium. 

First consider source 1. We have the elementary inequality
\[
e^{-r t_1^*} \geq 1 - rt_1^* = \frac{\phi_1(\phi_1+\phi_2) - r \zeta_1(\zeta_2-\zeta_1)}{\phi_1(\phi_1+\phi_2)}, 
\]
which is tight if and only if $\zeta_1 = \zeta_2$ and $t_1^*=0$. Thus from \eqref{eq:media_payoffs} we have
\[
U_1 \leq 1 - \frac{\phi_1(\phi_1+\phi_2) - r \zeta_1(\zeta_2-\zeta_1)}{\phi_1(\phi_1+\phi_2)} \cdot \frac{\zeta_2\phi_1}{\zeta_1\phi_2+\zeta_2\phi_1} - \lambda(\kappa-\phi_1)^2.
\]
Plugging in $\phi_2 = \phi^*$ and $\zeta_2 = \frac{\phi^*}{\sqrt{r}}$, the above simplifies to
\begin{align*}
U_1 &\leq 1 - \frac{\phi_1(\phi_1+\phi^*) - \sqrt{r} \zeta_1 \phi^* + r\zeta_1^2}{(\phi_1+\phi^*)(\phi_1+\sqrt{r}\zeta_1)} - \lambda(\kappa-\phi_1)^2 \\
&= \frac{(\phi_1 + 2\phi^* - \sqrt{r}\zeta_1)\sqrt{r}\zeta_1}{(\phi_1+\phi^*)(\phi_1+\sqrt{r}\zeta_1)} - \lambda(\kappa-\phi_1)^2.
\end{align*}

We now optimize this upper-bound over $\zeta_1$, and then over $\phi_1$. Let $y = \phi_1 + \sqrt{r}\zeta_1$. Note that the range of $y$ is $y \in [\phi_1, \phi_1 + \phi^*]$, since $\sqrt{r}\zeta_1 \leq \sqrt{r}\zeta_2 = \phi^*$. Using $y$, we can rewrite the above inequality as
\begin{equation}\label{eq:media_U1 upper bound}
U_1 \leq \frac{(2\phi_1 + 2\phi^* - y)(y-\phi_1)}{(\phi_1+\phi^*)y} - \lambda(\kappa-\phi_1)^2.
\end{equation}
Thus, in terms of $y$, we would like to maximize 
\[
\frac{(2\phi_1 + 2\phi^* - y)(y-\phi_1)}{y} = - y - \frac{2(\phi_1 + \phi^*)\phi_1}{y} + 2\phi^* + 3\phi_1
\]
This is a single-peaked function in $y$, with global maximum occurring at $y = \sqrt{2(\phi_1 + \phi^*)\phi_1} \geq \phi_1$. We distinguish two cases. If $\sqrt{2(\phi_1 + \phi^*)\phi_1} \geq \phi_1 + \phi^*$, then the maximum of this function of $y$ over the interval $[\phi_1, \phi_1 + \phi^*]$ occurs at $y = \phi_1 + \phi^*$. Plugging back into \eqref{eq:media_U1 upper bound}, we have
\[
U_1 \leq \frac{\phi^*}{\phi_1 + \phi^*} - \lambda(\kappa - \phi_1)^2. 
\]
But in previous analysis we have already shown that the above function of $\phi_1$ is maximized at $\phi^*$ whenever $\lambda \kappa^2 \geq \frac{9}{16}$, so we are done in this case. 

The more difficult case is $y = \sqrt{2(\phi_1 + \phi^*)\phi_1} \leq \phi_1 + \phi^*$, which corresponds to $\phi_1 \leq \phi^*$. Here the global maximum is achievable on the interval $y \in [\phi_1, \phi_1 + \phi^*]$, so we instead deduce from \eqref{eq:media_U1 upper bound} the following:
\[
U_1 \leq \frac{3\phi_1 + 2\phi^* - 2 \sqrt{2(\phi_1 + \phi^*)\phi_1}}{\phi_1+\phi^*} - \lambda(\kappa-\phi_1)^2.
\]
We denote the function on the RHS as $g(\phi_1)$, and aim to show $g(\phi_1) \leq g(\phi^*)$ for all $\phi_1 \leq \phi^*$. Note that
\[
g'(\phi_1) = \frac{\phi^*}{(\phi_1+\phi^*)^2} \cdot \left(1 - \sqrt{\frac{2(\phi_1+\phi^*)}{\phi_1}}\right) + 2\lambda(k-\phi_1).
\]
Thus $g'(\phi^*)= \frac{-1}{4\phi^*} + 2\lambda(k - \phi^*) = 0$. Moreover, we claim that $g'$ is concave on the interval $[0, \phi^*]$. To see this, let us write 
\[
\phi^* g'(\phi_1) = \frac{1}{(\frac{\phi_1}{\phi^*} + 1)^2} \cdot \left(1 - \sqrt{2 + \frac{2\phi^*}{\phi_1}}\right) + 2\phi^*\lambda(k-\phi_1).
\]
The second term $2\phi^*\lambda(k-\phi_1)$ is linear in $\phi_1$ and thus does not affect convexity/concavity. As for the first term, we can rewrite it as 
\[
h(z) = \frac{1}{(1+z)^2} \cdot \left(1 - \sqrt{2 + \frac{2}{z}}\right)
\]
with $z = \frac{\phi_1}{\phi^*}$. It thus suffices to show $h$ is concave on the interval $[0, 1]$. For this we compute that
\[
h''(z) = -3\frac{-8\sqrt{z^5(z+1)} + 8\sqrt{2}z^3 + 12\sqrt{2}z^2 + 5\sqrt{2}z + \sqrt{2}}{4z^{2.5}(z+1)^{4.5}},
\]
which is negative for $z \leq 1$ because $8\sqrt{z^5(z+1)} \leq 8\sqrt{2z^5} \leq 8\sqrt{2}z^2$. 

Hence $g'(\phi^*) = 0$ and $g'$ is concave on $[0, \phi^*]$. It follows that either $g'$ is non-negative on this whole interval, or it is first negative then positive. So $g$ is either increasing, or first decreasing and then increasing. It thus remains to show $g(0) \leq g(\phi^*)$, which reduces to 
\[
2 - \lambda \kappa^2 \leq \frac{1}{2} - \lambda (\kappa - \phi^*)^2. 
\]
Since $\lambda \kappa^2 \geq 1.6$, we have $\phi^* = \frac{1}{2}\left(\kappa + \sqrt{\kappa^2 - \frac{1}{2\lambda}}\right) > \frac{3}{4} \kappa$. Thus we indeed have
\[
2 - \lambda \kappa^2 \leq \frac{1}{2} - \frac{1}{16} \lambda \kappa^2 < \frac{1}{2} - \lambda (\kappa - \phi^*)^2,
\]
where the first inequality again uses the assumption $\lambda \kappa^2 \geq 1.6$. This completes the proof that source 1 does not have any profitable deviation.

\bigskip

We now turn to source 2's incentives. We use another elementary inequality 
\[
e^{rt_1^*} \geq 1 + rt_1^* = 1 + r \frac{\zeta_1(\zeta_2 - \zeta_1)}{\phi_1(\phi_1 + \phi_2)}.
\]
Given $\phi_1 = \phi^*$ and $\zeta_1 = \frac{\phi^*}{\sqrt{r}}$, we can simplify the above as
\[
e^{rt_1^*} \geq 1 + \frac{\sqrt{r}(\zeta_2 - \zeta_1)}{\phi_1 + \phi_2} = \frac{\phi_2 + \sqrt{r} \zeta_2}{\phi_2 + \phi^*}.
\]
So from \eqref{eq:media_payoffs} we obtain an upper bound for $U_2$:
\begin{equation}\label{eq:media_U2 upper bound}
\begin{split}
U_2 &= e^{-rt_1^*} \left( \frac{\zeta_2\phi_1}{\zeta_1\phi_2 + \zeta_2\phi_1}\right) - \lambda(\kappa - \phi_2)^2  \\
&\leq \frac{\phi_2 + \phi^*}{\phi_2 + \sqrt{r} \zeta_2} \cdot \frac{\sqrt{r}\zeta_2}{\phi_2 + \sqrt{r}\zeta_2} - \lambda(\kappa - \phi_2)^2\\
&= (\phi_2 + \phi^*)\frac{\sqrt{r}\zeta_2}{(\phi_2 + \sqrt{r}\zeta_2)^2} - \lambda(\kappa - \phi_2)^2
\end{split}
\end{equation}

Again we will optimize this upper bound first over $\zeta_2$, then over $\phi_2$. In terms of $\zeta_2$, we want to minimize
\[
\frac{(\phi_2 + \sqrt{r}\zeta_2)^2}{\sqrt{r}\zeta_2} = \sqrt{r} \zeta_2 + \frac{\phi_2^2}{\sqrt{r}\zeta_2} + 2\phi_2. 
\]
The global minimum of this single-dipped function occurs at $\sqrt{r}\zeta_2 = \phi_2$, but since $\sqrt{r}\zeta_2 \geq \sqrt{r}\zeta_1 = \phi^*$, there are two cases to consider. In the first case, $\phi_2 \leq \phi^*$. Then the minimum subject to $\sqrt{r}\zeta_2 \geq \phi^*$ occurs precisely at $\sqrt{r}\zeta_2 = \phi^*$. In this case the upper bound \eqref{eq:media_U2 upper bound} becomes $U_2 \leq \frac{\phi^*}{\phi_2 + \phi^*} - \lambda(\kappa-\phi_2)^2$, and as we have shown previously this is maximized at $\phi_2 = \phi^*$.

In the other case, we have $\phi_2 \geq \phi^*$. Then the optimal $\zeta_2$ that maximizes the upper bound \eqref{eq:media_U2 upper bound} is $\sqrt{r}\zeta_2 = \phi_2$. \eqref{eq:media_U2 upper bound} then becomes 
\[
U_2 \leq \frac{\phi_2 + \phi^*}{4\phi_2} - \lambda(\kappa - \phi_2)^2.
\]
We will show that the derivative of this function of $\phi_2$ is negative for any $\phi_2 \geq \phi^*$, so that $\phi_2 = \phi^*$ is again the optimal choice. This comparison reduces to $2\lambda(\kappa-\phi_2) \leq \frac{\phi^*}{4\phi_2^2}$, which is equivalent to $8\lambda(\kappa-\phi_2)\phi_2^2 \leq \phi^*$. Since equality obtains when $\phi_2 = \phi^*$, we just need to show
\[
8\lambda(\kappa-\phi_2)\phi_2^2 \leq 8\lambda(\kappa-\phi^*)(\phi^*)^2.
\]
After factoring out $\phi_2 - \phi^*$, this inequality reduces to 
\[
\kappa(\phi_2 + \phi^*) \leq \phi_2^2 + \phi_2 \phi^* + (\phi^*)^2.
\]
This holds whenever $\phi^* \geq \frac{2}{3}\kappa$, since $\phi_2 \geq \phi^*$ implies $\phi_2^2 + \phi_2 \phi^* + (\phi^*)^2 \geq \frac{3}{2}\phi^*(\phi_2 + \phi^*)$. Hence whenever $\lambda k^2 \geq \frac{9}{16}$, source 2 does not have profitable deviations either.

\subsection{Proof of Proposition \ref{prop:manipulation decreases attention}}

For $t \leq T$, $\hat{n}_i(t) = 0$ for every $i > 1$, so the result $\hat{n}_i(t) \leq n_i(t)$ trivially holds. Below we consider a fixed time $t > T$. We can further assume $n_1(t) < T$, because otherwise the proof of Proposition \ref{prop:manipulation increases attention} (in the main text) shows that $\hat{n}(t)$ coincides with $n(t)$. Given this assumption, we claim that $\hat{n}_1(t)$ is exactly equal to $T$. Indeed, let $\overline{t} \geq t$ be the first time at which $n_1(\overline{t}) = T$; such a time exists by monotonicity and continuity of $n_1(\cdot)$. Then  $\hat{n}(\overline{t})$ coincides with $n(\overline{t})$, which in particular implies $\hat{n}_1(\overline{t}) = n_1(\overline{t}) = T$. Monotonicity of $\hat{n}_1(\cdot)$ then implies $\hat{n}_1(t) \leq T$. But by assumption $\hat{n}_1(t) \geq T$, so equality must hold. 

Next, we connect the two vectors $n(t)$ and $\hat{n}(t)$ by a continuous path. For each $x \in [n_1(t), T]$, we define $q^x$ as the cumulative attention vector (at time $t$) resulting from a hypothetical attention manipulation that forces the agent to observe source 1 for $x$ units of time. That is,
\[
q^x = (q^x_1, \dots, q^x_K) = \argmin_{q_1, \dots, q_K \geq 0: ~\sum_{i} q_i = t \text{ and } q_1 \geq x} V(q). 
\]
Clearly, $n(t) = q^{n_1(t)}$ and $\hat{n}(t) = q^T$. By the same argument as in the previous paragraph, $q^x_1 = x$ holds for $x$ in this range. So in defining $q^x$ we can replace the constraint $q^x_1 \geq x$ with equality. 

To prove the proposition, it suffices to show that as $x$ decreases from $T$ to $n_1(t)$, $q^x_i$ weakly increases for each $i > 1$. Similar to our proof of Theorem \ref{thm:general}, the proof strategy here will be to use the Hessian matrix of $V$ to compute the derivative of the vector $q^x$ with respect to $x$. For this we fix $x > n_1(t)$, and assume for now that $q^x_i$ is strictly positive for each $i > 1$. Then the first-order condition for the constrained optimality of $q^x$ yields $\partial_2 V(q^x) = \cdots = \partial_K V(q^x)$. If the vector $q^x$ is left-differentiable at $x$, then for any $y$ slightly smaller than $x$, $q^y$ must also satisfy the above equal marginal value condition (for every source $i > 1$). 

Thus the left derivative of $q^x$ is a vector $u \in \mathbb{R}^K$ that satisfies $u_1 = 1$, $u_1 + \dots + u_K = 0$ and
\begin{equation}\label{eq:attention manipulation derivative}
Hess_V(q^x) \cdot u = \lambda (c, 1, \dots, 1)',
\end{equation}
for some $\lambda, c \in \mathbb{R}$ that will be determined later. 

Under the differentiability assumption, we can solve for $u$ as follows. Note from Lemma \ref{lemma:V partials} that $\partial_{ij} V = 2 \gamma_i \gamma_j [(\Sigma^{-1} + Q)^{-1}]_{ij}$, where we save notation by writing $Q = \diag(q^x)$ from now on. Then, we have the matrix identity 
\begin{equation}\label{eq:attention manipulation hessian}
Hess_V(q^x) = 2 \diag(\gamma) \cdot (\Sigma^{-1} + Q)^{-1} \cdot \diag(\gamma), 
\end{equation}
where $\gamma$ is as usual the vector $(\Sigma^{-1} + Q)^{-1} \cdot \alpha$. As shown in the proof of Proposition \ref{prop:substitutes}, $\gamma$ has strictly positive coordinates. Recalling $\partial_i V = -\gamma_i^2$, we thus have $\gamma_2 = \dots = \gamma_K  > 0$. $\gamma_1$ cannot be larger, since then $q^x_1$ should be larger than $x$ to minimize $V$. $\gamma_1$ cannot be equal to the other sources either, since in that case the vector $q^x$ would satisfy the first-order condition for the unconstrained variance minimization problem, leading to $q^x = n(t)$ and $x = n_1(t)$. So $\gamma_2 = \dots = \gamma_K > \gamma_1 > 0$, which will be useful below. 

Now, using \eqref{eq:attention manipulation derivative} and \eqref{eq:attention manipulation hessian}, we have 
\begin{align*}
u = \lambda \cdot Hess^{-1} \cdot (c, 1, \dots, 1)' & = 2 \lambda \cdot \diag(1/\gamma) \cdot (\Sigma^{-1}+Q) \cdot \diag(1/\gamma) \cdot (c, 1, \dots, 1)' \\
&= 2 \lambda \cdot \diag(1/\gamma) \cdot (\Sigma^{-1}+Q) \cdot \left(\frac{c}{\gamma_1}, \frac{1}{\gamma_2}, \dots, \frac{1}{\gamma_2}\right)' \\
&= \frac{2\lambda}{\gamma_2^2}\cdot \diag(1/\gamma) \cdot (\Sigma^{-1}+Q) \cdot \left(\frac{c\gamma_2^2}{\gamma_1}, \gamma_2, \dots, \gamma_2\right)'.
\end{align*}
If we rewrite $\frac{c\gamma_2^2}{\gamma_1}$ as $\gamma_1 - b$ for some $b \in \mathbb{R}$, then the vector $\left(\frac{c\gamma_2^2}{\gamma_1}, \gamma_2, \dots, \gamma_2\right)'$ differs from $\gamma$ only in that the first coordinate is smaller by $b$. Using $(\Sigma^{-1}+Q) \gamma = \alpha$, we thus obtain
\[
(\Sigma^{-1}+Q) \cdot \left(\frac{c\gamma_2^2}{\gamma_1}, \gamma_2, \dots, \gamma_2\right)' = (\alpha_1 - b[\Sigma^{-1}+Q]_{11}, \alpha_2 - b[\Sigma^{-1}+Q]_{21}, \dots, \alpha_K - b[\Sigma^{-1}+Q]_{K1})'
\]
By the symmetry of $\Sigma^{-1}+Q$, it follows that
\begin{equation}\label{eq:attention manipulation u}
u_i = \frac{2\lambda}{\gamma_2^2} \cdot \frac{\alpha_i - b[\Sigma^{-1}+Q]_{1i}}{\gamma_i}.
\end{equation}

Recall $u_1 = 1$, which reflects $q^y_1 = y$ for every $y$. Thus $\lambda \neq 0$. Moreover, recall $\sum_{i=1}^{K}u_i = 0$, a consequence of the fact that $\sum_{i=1}^{K}q^y_i = t$ for every $y$. Thus, we obtain
\begin{equation}\label{eq:attention manipulation b}
b \sum_{i=1}^{K} \frac{[\Sigma^{-1}+Q]_{1i}}{\gamma_i} = \sum_{i=1}^{K}\frac{\alpha_i}{\gamma_i}.
\end{equation}
The RHS is clearly positive. On the LHS, using $\gamma_2 = \dots = \gamma_K > \gamma_1 > 0$ and $[\Sigma^{-1}+Q]_{11} > 0$, 
\begin{align*}
\sum_{i=1}^{K} \frac{[\Sigma^{-1}+Q]_{1i}}{\gamma_i} &= \frac{[\Sigma^{-1}+Q]_{11}}{\gamma_1} + \sum_{i > 1} \frac{[\Sigma^{-1}+Q]_{1i}}{\gamma_i} \\
&> \frac{[\Sigma^{-1}+Q]_{11} \cdot \gamma_1}{\gamma_2^2} + \sum_{i > 1}\frac{[\Sigma^{-1}+Q]_{1i} \cdot \gamma_i}{\gamma_2^2} \\
&= \frac{1}{\gamma_2^2} \sum_{i=1}^{K}[\Sigma^{-1}+Q]_{1i} \cdot \gamma_i \\&
=\frac{\alpha_1}{\gamma_2^2} > 0,
\end{align*}
where the last equality uses $\sum_{i = 1}^{K} (\Sigma^{-1} + Q)_{1i} \cdot \gamma_i = \alpha_1 > 0$ as a consequence of $(\Sigma^{-1} + Q) \gamma = \alpha$.

Thus, \eqref{eq:attention manipulation b} implies the crucial inequality $b > 0$. It follows that for this $b$ and any $i > 1$, $\frac{\alpha_i - b[\Sigma^{-1}+Q]_{1i}}{\gamma_i}$ is strictly positive (by assumption $[\Sigma^{-1}+Q]_{1i} \leq 0$). From $\sum_{i = 1}^{K}u_i = 0$, we further know that $\frac{\alpha_1 - b[\Sigma^{-1}+Q]_{11}}{\gamma_1}$ is strictly negative. We can then use \eqref{eq:attention manipulation u} to determine the unique value of $\lambda < 0$ that makes $u_1 = 1$. Since $\lambda < 0$ and $b > 0$, \eqref{eq:attention manipulation u} yields $u_i < 0$ for every $i > 0$. Hence, starting from $q^x$, any local decrease in the amount of attention towards source 1 increases the optimal amount of attention towards every other source. This is what we desire to show. 

\bigskip



To complete the proof, we now dispense with two previous assumptions in the analysis. First, we have assumed $q^x_i > 0$ for each $i > 1$. In general, relabeling the attributes if necessary, we can assume that among the sources $2, \dots, K$, the ones with maximal marginal value (i.e., those with maximal $\gamma_i$) are sources $2 \sim k$ for some $k$. Then the first-order condition for the constrained optimality of $q^x_i$ requires $q^x_i = 0$ for $i > k$. Moreover, by the same argument as above, $\gamma_1$ is strictly smaller. By continuity, sources $2 \sim k$ also maximize the marginal value at any $q^y_i$ where $y$ is slightly smaller than $x$. So $q^y_i = 0$ for $i > k$ also holds. Thus, locally we can reduce the problem with $K$ sources to a smaller problem with only the first $k$ sources, similar to the proof of Theorem \ref{thm:general}. 

In this smaller problem, the payoff weight vector $\tilde{\alpha}$ is given by \eqref{eq:tilde alpha}, and it is strictly positive as shown in the proof of Theorem \ref{thm:general}. The prior covariance matrix becomes $\Sigma_{TL}$, which is the $k \times k$ top-left principal sub-matrix of $\Sigma$. Since $\Sigma^{-1}$ is an $M$-matrix, so is $(\Sigma_{TL})^{-1}$.\footnote{Rewriting $(\Sigma_{TL})^{-1}$ as the Schur complement of $\Sigma^{-1}$ with respect to its bottom-right block, the result follows from the fact that $M$-matrices are closed under Schur complements. This fact can be proved by the same induction argument as in \citet{CarlsonMarkham}.} The vector $\tilde{\gamma} = [(\Sigma_{TL})^{-1} + \tilde{Q}]^{-1} \cdot \tilde{\alpha} \in \mathbb{R}^k$ is the first $k$ coordinates of $\gamma$, so we have $\tilde{\gamma}_2 = \dots = \tilde{\gamma}_k > \tilde{\gamma}_1 > 0$. Hence, the above procedure for finding $u$ directly applies to this smaller problem, and yields a vector $\tilde{u} \in \mathbb{R}^k$ with $\tilde{u}_1 = 1$, $\tilde{u}_i < 0$ for $1 < i \leq k$ and $\tilde{u}_1 + \dots + \tilde{u}_k = 0$. In the original problem, the left-derivative of $q^x$ is thus the vector $(\tilde{u}_1, \dots, \tilde{u}_k, 0, \dots, 0)'$. Once again, locally decreasing attention towards source 1 weakly increases attention towards every other source. To be fully rigorous, we also note that by the Maximum Theorem, $q^x$ varies continuously with $x$. So the preceding form of local monotonicity implies global monotonicity (without continuity, $q^x_i$ may only be piece-wise monotone).

Our second simplifying assumption is left-differentiablity, which we now argue is without loss. Let $x$ be the infimum of those numbers $\hat{x} \in [n_1(t), T]$ such that $q^y$ is left-differentiable in $y$ for $y > \hat{x}$. By properties of the infimum, $x$ is in fact the \emph{smallest} number $\geq n_1(t)$ such that $q^y$ is left-differentiable for $y > x$. Now suppose $x > n_1(t)$, and we will deduce a contradiction. Specifically, let sources $2 \sim k$ have the maximal marginal value at $q^x$. We can use the above procedure to find the vectors $\tilde{u} \in \mathbb{R}^k$ and $(\tilde{u}_1, \dots, \tilde{u}_k, 0, \dots, 0)' \in \mathbb{R}^K$, such that if $q^x$ is perturbed slightly in the direction $-u$, the resulting attention vector maintains the equal marginal value property across sources $2 \sim k$. Formally, for any attention vector $q$ with total attention $t$ such that sources $2 \sim k$ have the maximal marginal value, we can solve for $b, \lambda \in \mathbb{R}$ and $\tilde{u} \in \mathbb{R}^k$ from \eqref{eq:attention manipulation u} and \eqref{eq:attention manipulation b} (with $\tilde{\alpha}$, $\Sigma_{TL}$ and $\diag(q)$ replacing $\alpha$, $\Sigma$ and $Q$ in those equations). It is easy to see that these solutions vary continuously with $q$, so we can write $u = f(q)$ for a continuous function $f$. This allows us to define the following system of ordinary differential equations with a right boundary condition (where the derivative at $x$ is interpreted as the left-derivative). 
\[
q'(y) = f(q(y)) \text{ for every } y \in (x- \epsilon, x], \text{ and } q(x) = q^x.
\]
By Peano's Existence Theorem, this system of ODE admits a solution when $\epsilon$ is sufficiently small. Note that by construction, for any $y$ in the interval $(x-\epsilon, x]$, we have $q_1(y) = y$, $\sum_{i = 1}^{K} q_i(y) = t$ and $q_i(y) = 0$ for $i > k$. Moreover, at the vector $q(y)$ sources $2 \sim k$ have equal marginal values, which are maximal if $\epsilon$ is sufficiently small. Hence $q(y)$ satisfies the Kuhn-Tucker conditions for minimizing $V(q)$ subject to $q_1 \geq y$, $q_i \geq 0$ and $\sum_{i = 1}^{K} q_i = t$. Since $V$ is convex, these conditions are sufficient, so that $q(y)$ coincides with $q^y$. But then we see that $q^y$ is left-differentiable for any $y > x-\epsilon$, contradicting the definition of $x$. 

This completes the entire proof of Proposition \ref{prop:manipulation decreases attention}.


%% file: main.bbl
\clearpage

\begin{thebibliography}{45}
\newcommand{\enquote}[1]{``#1''}
\expandafter\ifx\csname natexlab\endcsname\relax\def\natexlab#1{#1}\fi

\bibitem[\protect\citeauthoryear{Aghion, Bolton, Harris, and Jullien}{Aghion
  et~al.}{1991}]{Aghion}
\textsc{Aghion, P., P.~Bolton, C.~Harris, and B.~Jullien} (1991):
  \enquote{Optimal Learning by Experimentation,} \emph{Review of Economic
  Studies}, 58, 621--654.

\bibitem[\protect\citeauthoryear{Alarie and Bronsard}{Alarie and
  Bronsard}{1990}]{AlarieBronsard}
\textsc{Alarie, Y. and C.~Bronsard} (1990): \enquote{Preferences and Normal
  Goods: A Necessary and Sufficient Condition,} \emph{Journal of Economic
  Theory}, 51, 423--430.

\bibitem[\protect\citeauthoryear{Azevedo, Deng, Olea, Rao, and Weyl}{Azevedo
  et~al.}{2020}]{AzevedoEtAl}
\textsc{Azevedo, E., A.~Deng, J.~L.~M. Olea, J.~Rao, and G.~Weyl} (2020):
  \enquote{A/B Testing with Fat Tails,} \emph{Journal of Political Economy},
  forthcoming.

\bibitem[\protect\citeauthoryear{Bardhi}{Bardhi}{2020}]{Bardhi}
\textsc{Bardhi, A.} (2020): \enquote{Attributes: Selective Learning and
  Influence,} Working Paper.

\bibitem[\protect\citeauthoryear{Bergemann and V{\"a}lim{\"a}ki}{Bergemann and
  V{\"a}lim{\"a}ki}{2008}]{Bergemann2008}
\textsc{Bergemann, D. and J.~V{\"a}lim{\"a}ki} (2008): \enquote{Bandit
  Problems,} in \emph{The New Palgrave Dictionary of Economics}, ed. by S.~N.
  Durlauf and L.~E. Blume, Basingstoke: Palgrave Macmillan.

\bibitem[\protect\citeauthoryear{Bilancini and Boncinelli}{Bilancini and
  Boncinelli}{2010}]{BilanciniBoncinelli}
\textsc{Bilancini, E. and L.~Boncinelli} (2010): \enquote{Preferences and
  normal goods: An easy-to-check necessary and sufficient condition,}
  \emph{Economics Letters}, 108, 13--15.

\bibitem[\protect\citeauthoryear{Blackwell}{Blackwell}{1951}]{Blackwell}
\textsc{Blackwell, D.} (1951): \enquote{Comparison of Experiments,} in
  \emph{Proceedings of the Second Berkeley Symposium on Mathematical Statistics
  and Probability}, Berkeley and Los Angeles: University of California Press.

\bibitem[\protect\citeauthoryear{Bubeck, Munos, and Stoltz}{Bubeck
  et~al.}{2009}]{Bubeck2009}
\textsc{Bubeck, S., R.~Munos, and G.~Stoltz} (2009): \enquote{Pure Exploration
  in Multi-armed Bandits Problems,} in \emph{Algorithmic Learning Theory. ALT
  2009. Lecture Notes in Computer Science, vol 5809}, Springer, Berlin,
  Heidelberg.

\bibitem[\protect\citeauthoryear{Callander}{Callander}{2011}]{Callander}
\textsc{Callander, S.} (2011): \enquote{Searching and Learning by Trial and
  Error,} \emph{American Economic Review}, 101, 2277--2308.

\bibitem[\protect\citeauthoryear{Carlson and Markham}{Carlson and
  Markham}{1979}]{CarlsonMarkham}
\textsc{Carlson, D. and T.~L. Markham} (1979): \enquote{Schur Complements of
  Diagonally Dominant Matrices,} \emph{Czechoslovak Mathematical Journal}, 29,
  246--251.

\bibitem[\protect\citeauthoryear{Che and Mierendorff}{Che and
  Mierendorff}{2019}]{CheMierendorff}
\textsc{Che, Y.-K. and K.~Mierendorff} (2019): \enquote{Optimal Dynamic
  Allocation of Attention,} \emph{American Economic Review}, 108, 2993--3029.

\bibitem[\protect\citeauthoryear{Chen and Suen}{Chen and Suen}{2019}]{ChenSuen}
\textsc{Chen, H. and W.~Suen} (2019): \enquote{Competition for Attention and
  News Quality,} Working Paper.

\bibitem[\protect\citeauthoryear{Chick and Frazier}{Chick and
  Frazier}{2012}]{ChickFrazier}
\textsc{Chick, S.~E. and P.~Frazier} (2012): \enquote{Sequential Sampling with
  Economics of Selection Procedures,} \emph{Management Science}, 58, 550--569.

\bibitem[\protect\citeauthoryear{Chipman}{Chipman}{1977}]{Chipman}
\textsc{Chipman, J.~S.} (1977): \enquote{An Empirical Implication of
  Auspitz-Lieben-Edgeworth-Pareto Complementarity,} \emph{Journal of Economic
  Theory}, 14, 228--231.

\bibitem[\protect\citeauthoryear{Easley and Kiefer}{Easley and
  Kiefer}{1988}]{EasleyKiefer}
\textsc{Easley, D. and N.~M. Kiefer} (1988): \enquote{Controlling a Stochastic
  Process with Unknown Parameters,} \emph{Econometrica}, 56, 1045--1064.

\bibitem[\protect\citeauthoryear{Frazier, Powell, and Dayanik}{Frazier
  et~al.}{2008}]{FrazierEtAl2008}
\textsc{Frazier, P., W.~Powell, and S.~Dayanik} (2008): \enquote{A
  Knowledge-Gradient Policy for Sequential Information Collection,} \emph{SIAM
  Journal of Control and Optimization}, 47, 2410--2439.

\bibitem[\protect\citeauthoryear{Frazier, Powell, and Dayanik}{Frazier
  et~al.}{2009}]{FrazierEtAl2009}
---\hspace{-.1pt}---\hspace{-.1pt}--- (2009): \enquote{The Knowledge-Gradient
  Policy for Correlated Normal Beliefs,} \emph{Informs Journal on Computing},
  21, 599--613.

\bibitem[\protect\citeauthoryear{Fudenberg, Strack, and Strzalecki}{Fudenberg
  et~al.}{2018}]{FudenbergStrackStrzalecki}
\textsc{Fudenberg, D., P.~Strack, and T.~Strzalecki} (2018): \enquote{Speed,
  Accuracy, and the Optimal Timing of Choices,} \emph{American Economic
  Review}, 108, 3651--3684.

\bibitem[\protect\citeauthoryear{Galperti and Trevino}{Galperti and
  Trevino}{2020}]{GalpertiTrevino}
\textsc{Galperti, S. and I.~Trevino} (2020): \enquote{Coordination Motives and
  Competition for Attention in Information Markets,} \emph{Journal of Economic
  Theory}, 188.

\bibitem[\protect\citeauthoryear{Garfagnini and Strulovici}{Garfagnini and
  Strulovici}{2016}]{GarfagniniStrulovici}
\textsc{Garfagnini, U. and B.~Strulovici} (2016): \enquote{Social
  Experimentation with Interdependent and Expanding Technologies,} \emph{Review
  of Economic Studies}, 83, 1579--1613.

\bibitem[\protect\citeauthoryear{Gentzkow and Shapiro}{Gentzkow and
  Shapiro}{2008}]{GentzkowShapiro}
\textsc{Gentzkow, M. and J.~Shapiro} (2008): \enquote{Competition and Truth in
  the Market for News,} \emph{Journal of Economic Perspectives}, 22, 133--154.

\bibitem[\protect\citeauthoryear{Gittins}{Gittins}{1979}]{Gittins1979}
\textsc{Gittins, J.~C.} (1979): \enquote{Bandit Processes and Dynamic
  Allocation Indices,} \emph{Journal of the Royal Statistical Society, Series
  B}, 148--177.

\bibitem[\protect\citeauthoryear{Gossner, Steiner, and Stewart}{Gossner
  et~al.}{2020}]{GossnerSteinerStewart}
\textsc{Gossner, O., J.~Steiner, and C.~Stewart} (2020): \enquote{Attention
  Please!} \emph{Econometrica}, forthcoming.

\bibitem[\protect\citeauthoryear{Greenshtein}{Greenshtein}{1996}]{Greenshtein}
\textsc{Greenshtein, E.} (1996): \enquote{Comparison of Sequential
  Experiments,} \emph{The Annals of Statistics}, 24, 436--448.

\bibitem[\protect\citeauthoryear{Hansen and Torgersen}{Hansen and
  Torgersen}{1974}]{HansenTorgersen}
\textsc{Hansen, O.~H. and E.~N. Torgersen} (1974): \enquote{Comparison of
  Linear Normal Experiments,} \emph{The Annals of Statistics}, 2, 367--373.

\bibitem[\protect\citeauthoryear{H\'{e}bert and Woodford}{H\'{e}bert and
  Woodford}{2020}]{HebertWoodford}
\textsc{H\'{e}bert, B. and M.~Woodford} (2020): \enquote{Rational Inattention
  When Decisions Take Time,} Working Paper.

\bibitem[\protect\citeauthoryear{Johnson}{Johnson}{1982}]{Johnson1982}
\textsc{Johnson, C.~R.} (1982): \enquote{Inverse M-Matrices,} \emph{Linear
  Algebra and Its Applications}, 47, 195--216.

\bibitem[\protect\citeauthoryear{Ke and Villas-Boas}{Ke and
  Villas-Boas}{2019}]{KeVillas-Boas}
\textsc{Ke, T.~T. and J.~M. Villas-Boas} (2019): \enquote{Optimal Learning
  Before Choice,} \emph{Journal of Economic Theory}, 180, 383--437.

\bibitem[\protect\citeauthoryear{Keller, Rady, and Cripps}{Keller
  et~al.}{2005}]{KellerRadyCripps}
\textsc{Keller, G., S.~Rady, and M.~Cripps} (2005): \enquote{Strategic
  Experimentation with Exponential Bandits,} \emph{Econometrica}, 73,
  2787--2802.

\bibitem[\protect\citeauthoryear{Klabjan, Olszwski, and Wolinsky}{Klabjan
  et~al.}{2014}]{KlabjanOlszewskiWolinsky}
\textsc{Klabjan, D., W.~Olszwski, and A.~Wolinsky} (2014):
  \enquote{Attributes,} \emph{Games and Economic Behavior}, 88, 190--206.

\bibitem[\protect\citeauthoryear{Liang, Mu, and Syrgkanis}{Liang
  et~al.}{2017}]{LiangMuSyrgkanis}
\textsc{Liang, A., X.~Mu, and V.~Syrgkanis} (2017): \enquote{Optimal and Myopic
  Information Acquisition,} Working Paper.

\bibitem[\protect\citeauthoryear{Mackowiak, Mat\v{e}jka, and
  Wiederholt}{Mackowiak et~al.}{2020}]{RISurvey}
\textsc{Mackowiak, B., F.~Mat\v{e}jka, and M.~Wiederholt} (2020):
  \enquote{Rational Inattention: A Review,} Working Paper.

\bibitem[\protect\citeauthoryear{Mayskaya}{Mayskaya}{2019}]{Mayskaya}
\textsc{Mayskaya, T.} (2019): \enquote{Dynamic Choice of Information Sources,}
  Working Paper.

\bibitem[\protect\citeauthoryear{Morris and Strack}{Morris and
  Strack}{2019}]{MorrisStrack}
\textsc{Morris, S. and P.~Strack} (2019): \enquote{The Wald Problem and the
  Equivalence of Sequential Sampling and Static Information Costs,} Working
  Paper.

\bibitem[\protect\citeauthoryear{Perego and Yuksel}{Perego and
  Yuksel}{2020}]{PeregoYuksel}
\textsc{Perego, J. and S.~Yuksel} (2020): \enquote{Media Competition and Social
  Disagreement,} Working Paper.

\bibitem[\protect\citeauthoryear{Plemmons}{Plemmons}{1977}]{Plemmons1977}
\textsc{Plemmons, R.} (1977): \enquote{M-Matrix
  Characterizations.I---Nonsingular M-Matrices,} \emph{Linear Algebra and Its
  Applications}, 18, 175--188.

\bibitem[\protect\citeauthoryear{Quah}{Quah}{2007}]{Quah}
\textsc{Quah, J. K.-H.} (2007): \enquote{The Comparative Statics of Constrained
  Optimization Problems,} \emph{Econometrica}, 75, 401--431.

\bibitem[\protect\citeauthoryear{Ratcliff and McKoon}{Ratcliff and
  McKoon}{2008}]{RatcliffMcKoon}
\textsc{Ratcliff, R. and G.~McKoon} (2008): \enquote{The Diffusion Decision
  Model: Theory and Data for Two-Choice Decision Tasks,} \emph{Neural
  Computation}, 20, 873--922.

\bibitem[\protect\citeauthoryear{Revuz and Yor}{Revuz and Yor}{1999}]{RevuzYor}
\textsc{Revuz, D. and M.~Yor} (1999): \emph{Continuous Martingales and Brownian
  Motion}, New York: Springer.

\bibitem[\protect\citeauthoryear{Russo}{Russo}{2016}]{Russo}
\textsc{Russo, D.} (2016): \enquote{Simple Bayesian Algorithms for Best-Arm
  Identification,} \emph{Operations Research}, 68, 1625--1647.

\bibitem[\protect\citeauthoryear{Sanjurjo}{Sanjurjo}{2017}]{Sanjurjo}
\textsc{Sanjurjo, A.} (2017): \enquote{Search with Multiple Attributes: Theory
  and Empirics,} \emph{Games and Economic Behavior}, 104, 535--562.

\bibitem[\protect\citeauthoryear{Steiner, Stewart, and Mat\v{e}jka}{Steiner
  et~al.}{2017}]{SteinerStewartMatejka}
\textsc{Steiner, J., C.~Stewart, and F.~Mat\v{e}jka} (2017): \enquote{Rational
  Inattention Dynamics: Inertia and Delay in Decision-Making,}
  \emph{Econometrica}, 85, 521--553.

\bibitem[\protect\citeauthoryear{Yang}{Yang}{2015}]{Yang}
\textsc{Yang, M.} (2015): \enquote{Coordination with Flexible Information
  Acquisition,} \emph{Journal of Economic Theory}, 158, 721--738.

\bibitem[\protect\citeauthoryear{Yong and Zhou}{Yong and Zhou}{1999}]{YongZhou}
\textsc{Yong, J. and X.~Y. Zhou} (1999): \emph{Stochastic Controls: Hamiltonian
  Systems and HJB Equations}, New York: Springer.

\bibitem[\protect\citeauthoryear{Zhong}{Zhong}{2019}]{Zhong}
\textsc{Zhong, W.} (2019): \enquote{Optimal Dynamic Information Acquisition,}
  Working Paper.

\end{thebibliography}
